\documentclass[journal,onecolumn,twoside]{IEEEtran}
\normalsize
\ifCLASSINFOpdf
\else
\fi
\usepackage{setspace}
\usepackage{color}
\usepackage{cite}
\usepackage{amsmath}
\usepackage{amsfonts}
\usepackage{amssymb}
\usepackage{amsthm}
\usepackage{tikz}
\usepackage{cases}
\usepackage{bm}
\usepackage{graphicx}
    \graphicspath{{../}}
    \DeclareGraphicsExtensions{.pdf}
\usepackage[caption=true,font=footnotesize]{subfig}
\usepackage{multirow}
\usepackage{makecell}
\usepackage{mathdots}
\usepackage{booktabs}
\usepackage{url}
\usepackage{bm}
\usepackage{xtab}
\usepackage{pifont}
\usepackage{array}
\usepackage{algorithm}
\usepackage{algorithmic}
\usepackage{enumerate}
\usetikzlibrary{arrows}
\usepackage{caption}
\usepackage{soul}
\usepackage{xcolor}

\allowdisplaybreaks[4]
\usepackage[colorlinks,
            linkcolor=blue,
            anchorcolor=blue,
            citecolor=blue]{hyperref}
\theoremstyle{plain}

\newtheorem{theorem}{Theorem}
\newtheorem{lemma}[theorem]{Lemma}

\newtheorem{definition}[theorem]{Definition}
\newtheorem{corollary}[theorem]{Corollary}
\theoremstyle{definition}
\newtheorem{example}{Example}

\newtheorem{remark}{Remark}

\begin{document}

\title{On Sequence Reconstruction Problem for $q$-ary Deletion Channels}

\author{Xiang Wang, Han Li, and Fang-Wei Fu
        % <-this % stops a space
\thanks{X. Wang is with the School of Mathematics, Statistics and Mechanics, Beijing University of Technology, Beijing, 100124, China (e-mail: xwang@bjut.edu.cn). Han Li is with the Chern Institute of Mathematics and LPMC, Nankai University, Tianjin 300071, China (e-mail: hli@mail.nankai.edu.cn).  F.-W. Fu is with the Chern Institute of Mathematics and LPMC, Nankai University, Tianjin 300071, China (e-mail: fwfu@nankai.edu.cn).}}% <-this % stops a space
%\thanks{Manuscript received October 26, 2023; revised December 8, 2023.}}

% The paper headers
%\markboth{Journal of \LaTeX\ Class Files,~Vol.~1, No.~2, December~2023}%
%{Shell \MakeLowercase{\textit{et al.}}: A Sample Article Using IEEEtran.cls for IEEE Journals}

%\IEEEpubid{0000--0000~\copyright~2023 IEEE}
% Remember, if you use this you must call \IEEEpubidadjcol in the second
% column for its text to clear the IEEEpubid mark.

\maketitle

\begin{abstract}
The sequence reconstruction problem for $q$-ary deletion channels, introduced by Levenshtein in 2001, concerns the minimum number of channels required to uniquely recover a transmitted sequence from noisy outputs, where each channel introduces exactly $t$ deletions. Combinatorially, this problem is equivalent to determining $N_q(n,d,t)$, the maximum size of the intersection of two $t$-radius deletion balls whose centers are at Levenshtein distance at least $d$, for $q$-ary sequences of length $n$ over the alphabet \(\Sigma_q=\{0,1,\dots,q-1\}\). Levenshtein solved the uncoded case \(N_q(n,1,t)\) for all \(n \ge t\); subsequently, Gabrys and Yaakobi established \(N_2(n,2,t)\) for binary one-deletion-correcting codes, and Wang et al. extended this to the ternary case, determining \(N_3(n,2,t)\).

In this paper, we study the problem for \(q\)-ary sequences under minimum Levenshtein distance \(d=2\) with channels that introduce exactly \(t\) deletions. We determine the exact value of \(N_q(n,2,t)\) for all \(t\ge 2, q\geq 4\), and for sufficiently large \(n\), and show that the maximum intersection is attained by a family of explicitly constructed pairs of sequences. Furthermore, for each \(q\ge3\), we fully characterize the structure of $q$-ary sequence pairs that achieve this extremal intersection. In particular, if the intersection size matches the first two  terms of \(N_q(n,2,t)\), then the two center sequences must contain, at the same positions, length-5 blocks of the forms \((a,b,c,a,b)\) and \((b,a,c,b,a)\) for some distinct \(a,b,c\in\Sigma_q\); for \(t\ge q+2\), the exact maximum \(N_q(n,2,t)\) is attained precisely by \(2q!\) unordered pairs of sequences with a specific block structure.

Asymptotically, we prove that for \(q\ge 4\) and \(t\ge 2\),
\[
N_q(n,2,t)=\frac{6}{(t-2)!}n^{t-2}-\frac{3t+13}{(t-3)!}n^{t-3}+\frac{3t^2+25t+64}{4(t-4)!}n^{t-4}+O(n^{t-5}).
\]
Moreover, \(N_q(n,2,t)\) and \(N_{q-1}(n,2,t)\) share their first \(q-1\)  terms, and for \(t\ge q\) the coefficient of \(n^{t-q}\) in their difference is \(\frac{6t-6q+5}{(t-q)!}\).  
\end{abstract}


\section{Introduction}
\IEEEPARstart{L}{evenshtein} introduced the sequence reconstruction problem in 2001~\cite{L1}. In this framework, a codeword $\mathbf{x} \in \mathcal{C}$ is transmitted over several noisy channels, and the decoder receives the distinct outputs from each channel, aiming to uniquely recover $\mathbf{x}$. The problem originally emerged from biological and chemical applications, and has recently gained renewed interest due to advances in racetrack memories~\cite{Parkin,Chee} and DNA-based storage systems~\cite{Church,Yazdi,Lenz}.

Let $\Sigma_q = \{0,1,\dots,q-1\}$ denote the alphabet, and let $\Sigma_q^n$ be the set of all $q$-ary sequences of length $n$. Suppose $\rho: \Sigma_q^n \times \Sigma_q^n \rightarrow \mathbb{N}$ is a metric, and $\mathcal{C} \subseteq \Sigma_q^n$ is a code with minimum distance $d$ under $\rho$. When each channel introduces at most $t$ errors, Levenshtein~\cite{L1} showed that the minimum number of channels required is $N_q(n,d,t)+1$, where $N_q(n,d,t)$ is the size of the largest possible
intersection of two radius-$t$ balls centered at distinct codewords at distance at least $d$:
\begin{equation*}
N_q(n,d,t)=\max_{\mathbf{x}_1,\mathbf{x}_2\in \Sigma_q^n,\,\rho(\mathbf{x}_1,\mathbf{x}_2)\geq d } |B_t(\mathbf{x}_1)\cap B_t(\mathbf{x}_2)|,
%\label{intr::eq1}
\end{equation*}
where $B_t(\mathbf{x}) = \{\mathbf{y}\in \Sigma_q^n \mid \rho(\mathbf{x},\mathbf{y})\leq t\}$. Determining $N_q(n,d,t)$ is known as the \emph{sequence reconstruction problem}.

Levenshtein's early work~\cite{L1} examined this problem under various metrics, including Hamming distance, Johnson graphs, and others. Subsequent studies explored settings such as permutations~\cite{K1,Yaakobi,Wang1,Wang2}, general error graphs~\cite{L3,L4}, Grassmann graphs~\cite{Yaakobi}, and deletion/insertion channels~\cite{Sala1,Lan,Sun1,Zhang,Song,Wang, Wang4, Gabrys, Pham, Pham2, Sun0, Sun2}.

Deletion channels have attracted significant attention. For the uncoded case $\mathcal{C} = \Sigma_q^n$ with $q \ge 2$, Levenshtein~\cite{L1,L2} determined $N_q(n,1,t)$. Gabrys and Yaakobi~\cite{Gabrys} later solved the binary one-deletion-correcting case, obtaining $N_2(n,2,t)$. Subsequently, Pham, Goyal, and Kiah~\cite{Pham,Pham2} gave an asymptotic characterization of $N_2(n,d,t)$ for $d \ge 2$, showing $N_2(n,d,t) = \frac{\binom{2d}{d}}{(t-d)!} n^{t-d} - O(n^{t-d-1})$ for $2 \le d \le t < n$, and $N_2(n,t,t) = \binom{2t}{t}$. Focusing on specific parameter settings, Wang et al.~\cite{Wang4} proved that $N_2(n,3,4)=20n-166$ holds for all $n\ge 13$. In addition, Wang et al.~\cite{Wang} determined the exact value of $N_3(n,2,t)$ for the ternary case. For the binary $3$-deletion channel, Zhang et al.~\cite{Zhang} characterized sequence pairs $(\mathbf{x},\mathbf{y})$ with distance at least $3$ such that $|D_3(\mathbf{x})\cap D_3(\mathbf{y})|\in\{19,20\}$, where $D_3(\cdot)$ denotes the deletion ball of radius $3$.  More generally, Sun and Ge~\cite{Sun0} proved the asymptotic upper bound $N_q(n,d,t)\le \binom{2d}{d}\binom{n-d}{t-d} = \frac{\binom{2d}{d}}{(t-d)!}n^{t-d} - O(n^{t-d-1})$ for $1\le d\le t<n$.

In this paper, we study the sequence reconstruction problem over deletion channels for $q$-ary sequences of length $n$, where the transmitted codeword belongs to a one‑deletion‑correcting code and each channel introduces exactly $t$ deletions. Our goal is to determine, for general $q\ge 3$, the minimum number of channel outputs required to uniquely recover the codeword. For a sequence $\mathbf{x} \in \Sigma_q^n$, the \emph{deletion ball} of radius $t$ is
\begin{equation*}
D_t(\mathbf{x}) = \{ \mathbf{y} \in \Sigma_q^{n-t} \mid \mathbf{y} \text{ is a subsequence of } \mathbf{x} \},
\end{equation*}
with $D_t(\mathbf{x}) = \emptyset$ for $t < 0$ or $t > n$. Let $D_q(n,t)$ denote the maximum size of such a ball over all sequences of length $n$, where we set $D_q(n,t)=0$ for $t<0$ or $t>n$. Throughout this paper, we adopt the conventions that \(\sum_{i=a}^{b} f(i)=0\) for \(a>b\) and \(1/k! = 0\) for any integer \(k<0\). It is known from~\cite{Hirschberg} that
\begin{equation}
D_q(n,t) = \sum_{i=0}^{t} \binom{n-t}{i} D_{q-1}(t,t-i)=\frac{n^t}{t!}-O(n^{t-1}).
\label{intr:eq2}
\end{equation}
In particular, $D_2(n,t) = \sum_{i=0}^{t} \binom{n-t}{i}$ and $D_3(n,t) = \sum_{i=0}^{t} \binom{n-t}{i} \sum_{j=0}^{t-i} \binom{i}{j}$.

Our main contribution is the exact determination of $N_q(n,2,t)$ for all $q \geq 3$ and $t \geq 2$, provided $n$ is sufficiently large. More precisely, for $t\geq 2, q\geq 3$, and sufficiently large $n$, we have
$N_q(n,2,t)=N_q(n,t),$
where $N_q(n,t)$ is defined as
\begin{align*}
N_q(n,t)\triangleq&
\sum_{j=0}^{q-3} \binom{q-3}{j} \Bigl[ 2D_q(n-q+1,\,t-j-1) + D_q(n-q+1,\,t-j-2)+ 2D_q(n-q-2,\,t-j-2) \nonumber\\
&\quad + D_q(n-q-2,\,t-j-3) + D_q(n-2q-2,\,t-j-2-q) - 2D_q(n-q-1,\,t-j-1) \nonumber\\
&\quad - D_q(n-q-1,\,t-j-2) - D_q(n-q-1,\,t-j-3)- D_q(n-q-2,\,t-j-4)  \nonumber\\
&\quad - D_q(n-2q-1,\,t-j-1-q) \Bigr] - \sum_{m=0}^{q-4} \sum_{j=0}^{m} \binom{m}{j} D_q(n-3-q-m,\,t-2-q-j).%\label{intr::eq3}
\end{align*}
Moreover, for all $t\ge 2$, $q\ge 4$, and sufficiently large $n$, the following asymptotic expansion holds:
\begin{equation*}
N_q(n,2,t)=\frac{6}{(t-2)!}n^{t-2}-\frac{3t+13}{(t-3)!}n^{t-3}+\frac{3t^2+25t+64}{4(t-4)!}n^{t-4}+O(n^{t-5}).
% \label{intr::eq4}
\end{equation*}
In particular, when $t$ is small, we obtain the following exact values:
\begin{align*}
N_q(n,2,t)=N_q(n,t)=
\begin{cases}
 6,&\text{if~} t=2,\\
6n-22,&\text{if~} t=3,\\
3n^2-25n+53,&\text{if~} t=4,
\end{cases}
\end{align*}
for $q\geq 4$ and sufficiently large $n$.

The remainder of the paper is organized as follows. Section~\ref{sec2} introduces notation and preliminary results. Section~\ref{sec3} derives key properties of $D_q(n,t)$ and characterizes the sequences that attain this quantity, either exactly or asymptotically. Section~\ref{sec4} establishes lower bounds on $N_q(n,2,t)$. Section~\ref{sec5} proves the matching upper bounds, yielding the exact value. Section~\ref{sec6} concludes the paper with a summary and future research directions.

\section{Definitions and Preliminaries}
\label{sec2}

For integers \(i\) and \(j\), we define \([i,j]\) to be the set \(\{i,i+1,\dots,j\}\) if \(i \le j\), and the empty set otherwise.  In particular, \([n]=[1,n]=\{1,2,\dots,n\}\) for every positive integer $n$.
Scalars are denoted by lowercase letters, while vectors or sequences appear in bold lowercase letters. Let $S_q$ be the set of all permutations over $\Sigma_q$. We denote by $\pi:=[\pi(0),\pi(1),\dots,\pi(q-1)]$ a \emph{permutation} over $\Sigma_q$.
Given $\mathbf{x} = (x_1, x_2, \dots, x_n) \in \Sigma_q^n$, the \emph{projection} onto the index interval $[i,j]$ is denoted by $\mathbf{x}_{[i,j]}$; the \emph{reversal} of \(\mathbf{x}\) is denoted by \(\mathbf{x}^R=(x_n,\dots,x_1)\); for any $\pi\in S_q$, the sequence $\pi(\mathbf{x})$ is obtained by applying $\pi$ to each component of $\mathbf{x}$, i.e., $\pi(\mathbf{x})=(\pi(x_1),\pi(x_2),\dots,\pi(x_n))$. The length of $\mathbf{x}$ is $|\mathbf{x}|$. A \emph{run} is a maximal substring of identical symbols, and the number of runs in a string $\mathbf{x}\in \Sigma_q^n$ is denoted by $r(\mathbf{x})$.

For two sequences $\mathbf{u} = (u_1,\dots,u_m)$ and $\mathbf{v} = (v_1,\dots,v_n)$, we write $\mathbf{u} \circ \mathbf{v} = (u_1,\dots,u_m,v_1,\dots,v_n)$ for their concatenation. If $\mathcal{C} \subseteq \Sigma_q^n$ and $\mathbf{u}$ has length at most $n$, the subset of $\mathcal{C}$ consisting of sequences that start with $\mathbf{u}$ is written as $\mathcal{C}^{\mathbf{u}}$. Similarly, $\mathcal{C}_{\mathbf{u}_2}^{\mathbf{u}_1}$ denotes the subset of $\mathcal{C}$ consisting of sequences that begin with $\mathbf{u}_1$ and end with $\mathbf{u}_2$. For $\mathbf{u} \in \Sigma_q^m$ and $\mathcal{C} \subseteq \Sigma_q^n$, we define
\[
\mathbf{u} \circ \mathcal{C} = \{\mathbf{u} \circ \mathbf{c} \mid \mathbf{c} \in \mathcal{C}\}.
\]

For $\mathbf{x} \in \Sigma_q^{n+k}, \mathbf{y} \in \Sigma_q^n$, and $t \ge 0$, define $D_q(\mathbf{x},\mathbf{y};t+k,t) = D_{t+k}(\mathbf{x}) \cap D_t(\mathbf{y})$. The \emph{Levenshtein distance} between $\mathbf{x}$ and $\mathbf{y}$ is
\[
d_L(\mathbf{x},\mathbf{y}) = \min\{s \ge 0 \mid D_q(\mathbf{x},\mathbf{y};s+k,s) \ne \emptyset\}.
\]
When $\mathbf{x}$ and $\mathbf{y}$ have the same length $n$, $d_L(\mathbf{x},\mathbf{y})$ equals the half of the minimum total number of insertions and deletions required to transform $\mathbf{x}$ into $\mathbf{y}$. A code $\mathcal{C} \subseteq \Sigma_q^n$ has \emph{minimum Levenshtein distance} at least $d$ if $d_L(\mathbf{x},\mathbf{y}) \ge d$ for any two distinct $\mathbf{x},\mathbf{y} \in \mathcal{C}$; such a code can correct up to $d-1$ deletions.

We define $N_q(n,d,t)$ as the maximum size of the intersection of two deletion balls of radius $t$ centered at sequences in $\Sigma_q^n$ with Levenshtein distance at least $d$:
\[
N_q(n,d,t) = \max\{|D_t(\mathbf{x}) \cap D_t(\mathbf{y})| : \mathbf{x},\mathbf{y} \in \Sigma_q^n,\ d_L(\mathbf{x},\mathbf{y}) \ge d\}.
\]
For sequences of possibly different lengths we set
\[
N_q(n,d,t+k,t) = \max\{|D_{t+k}(\mathbf{x}) \cap D_t(\mathbf{y})| : \mathbf{x} \in \Sigma_q^{n+k},\ \mathbf{y} \in \Sigma_q^n,\ d_L(\mathbf{x},\mathbf{y}) \ge d\}.
\]
Suppose a codeword from a $(d-1)$-deletion-correcting code $\mathcal{C} \subseteq \Sigma_q^n$ is transmitted over $N$ channels, each introducing exactly $t$ deletions, and the decoder receives $N$ distinct outputs. Levenshtein~\cite{L1} proved that the minimum $N$ guaranteeing successful reconstruction equals $N_q(n,d,t) + 1$.

Let $\sigma = (\sigma_1,\dots,\sigma_q)$ be an ordering of $\Sigma_q$. The periodic sequence $\mathbf{c}_q(n,\sigma)$ of length $n$ is defined by
\[
\mathbf{c}_q(n,\sigma) = (c_1,\dots,c_n) \quad \text{with} \quad c_i = \sigma_{((i-1) \bmod q) + 1}, \quad 1 \le i \le n.
\]
When $\sigma$ is the identity ordering $I_q = (0,1,\dots,q-1)$, we denote $\mathbf{a}_n = \mathbf{c}_q(n,I_q)$ and write $\mathbf{a}_n(j)$ for its $j$-th coordinate. 
The collection of all such periodic sequences over all orderings of $\Sigma_q$ is denoted by 
\[
\mathbf{C}_q(n) = \{\mathbf{c}_q(n,\sigma) \mid \sigma \text{ is an ordering of } \Sigma_q\}.
\]

As shown in~\cite{Hirschberg}, for every $\mathbf{x} \in \mathbf{C}_q(n)$ we have
\[
D_q(n,t) = |D_t(\mathbf{x})| = \sum_{i=0}^{t} \binom{n-t}{i} D_{q-1}(t,t-i).
\]
Moreover, $D_q(n,t)$ satisfies the recurrence
\begin{equation}
D_q(n,t) = D_q(n-1,t) + D_q(n-2,t-1)+\cdots + D_q(n-q,t-q+1), \quad n \ge t+1,  \label{def:eq1}
\end{equation}
with the conventions that $D_q(n,t)=0$ whenever $t>n$, $n<0$, or $t<0$, and $D_q(n,n)=1$ for all $n\ge 1$. 

Several results concerning $N_q(n,d,t)$ have been established. For $d=1$, Levenshtein~\cite{L1} proved the following theorem.

\begin{theorem}[Levenshtein \cite{L1}]
For any $n \ge t+1$ and $q \ge 2$,
\[
N_q(n,1,t) = \sum_{i=1}^{q-1} D_q(n-i-1,t-i) + D_q(n-2,t-1) = D_q(n,t) - D_q(n-1,t) + D_q(n-2,t-1).
\]
This is achieved by taking $\mathbf{x} = (c_1,c_2,c_3,\dots,c_n) \in \mathbf{C}_q(n)$ and $\mathbf{y} = (c_2,c_1,c_3,\dots,c_n)$. That is, $d_L(\mathbf{x},\mathbf{y}) = 1$ and $|D_t(\mathbf{x}) \cap D_t(\mathbf{y})| = N_q(n,1,t)$.
\label{def:thm1}
\end{theorem}

For the binary alphabet, Gabrys and Yaakobi~\cite{Gabrys} determined $N_2(n,2,t)$.

\begin{theorem}[Gabrys and Yaakobi \cite{Gabrys}]
For $t \ge 2$ and $n \ge \max\{8, 2t+1\}$,
\begin{align*}
N_2(n,2,t) &= 2D_2(n-4,t-2) + 2D_2(n-5,t-2) + 2D_2(n-7,t-2)+ D_2(n-6,t-3) + D_2(n-7,t-3).
\end{align*}
\label{def:thm2}
\end{theorem}

In the binary case, Wang et al.~\cite{Wang4} proved that $N_2(n,3,4)=20n-166$ for all $n\geq 13$, and established the following result.

\begin{theorem}[Wang, Fang, Li, and Fu \cite{Wang4}]
For all $n \ge 13$,
\begin{align*}
N_2(n,3,4) &= 20n-166.
\end{align*}
Moreover, for $n\geq \max\{13,t+8\}$ and $t\geq 4$,
\begin{align*}
N_2(n,3,t)\geq  &6D_2(n-6,t-3)+4D_2(n-8,t-3)+6D_2(n-9,t-3)+4D_2(n-11,t-3)\\
&~~~~~+2D_2(n-13,t-5)+D_2(n-13,t-6).
\end{align*}
\label{def:thm3}
\end{theorem}

For the ternary alphabet, Wang et al.~\cite{Wang} determined the exact value of $N_3(n,2,t)$.

\begin{theorem}[Wang, Li, and Fu \cite{Wang}]
For $t \ge 2$ and $q = 3$, the following hold:
\begin{itemize}
    \item If $t \le n \le \frac{3t}{2}$, then $N_3(n,2,t) = 3^{n-t}$.
    \item If $2 \le t \le 5$ and $n \ge \max\{6, \lfloor 3t/2 \rfloor + 1\}$, then $N_3(n,2,t) = M_1(n,t)$.
    \item If $t \ge 6$ and $\lfloor 3t/2 \rfloor + 1 \le n \le 3t-1$, then $N_3(n,2,t) = \max\{M_0(n,t), M_1(n,t)\}$.
    \item If $t \ge 6$ and $n \ge 3t$, then $N_3(n,2,t) = M_1(n,t)$.
\end{itemize}
Here,
\begin{align*}
M_0(n,t)\triangleq& D_3(n-4,t-2)+3D_3(n-5,t-2)+4D_3(n-5,t-3)+3D_3(n-6,t-3)+D_3(n-6,t-4)\\
&+2D_3(n-7,t-3)+2D_3(n-7,t-4)+D_3(n-8,t-4)+D_3(n-12,t-7)-D_3(n-10,t-5),
\end{align*}
and
$$ M_1(n,t)\triangleq D_3(n-4,t-2)+5D_3(n-5,t-2)+4D_3(n-5,t-3)+3D_3(n-6,t-3)+D_3(n-6,t-4)+D_3(n-8,t-5).$$
Furthermore, for any fixed $t \ge 3$,
\[
N_3(n,2,t) = \frac{6}{(t-2)!} n^{t-2} - \frac{3t+13}{(t-3)!} n^{t-3} + O(n^{t-4}).
\]
\label{def:thm4}
\end{theorem}

In the binary case, Pham, Goyal, and Kiah~\cite{Pham} gave asymptotically tight bounds for $N_2(n,d,t)$ with $1 \le d \le t$.

\begin{theorem}[Pham, Goyal, and Kiah \cite{Pham}]
For $n > t \ge d \ge 1$,
\[
N_2(n,d,t) = \frac{\binom{2d}{d}}{(t-d)!} n^{t-d} - O(n^{t-d-1}).
\]
In the special case $t = d \ge 1$ and $n \ge 4t-2$, we have $N_2(n,d,d) = \binom{2d}{d}$.
\label{def:thm5}
\end{theorem}

Sun and Ge~\cite{Sun0} extended this result to the $q$-ary alphabet, strengthening the asymptotic upper bound.

\begin{theorem}[Sun and Ge \cite{Sun0}]\label{def:thm6}
 For $n > t \ge d \ge 1$,
 \begin{align*}
     N_q(n,d,t) \le \binom{2d}{d}\binom{n-d}{t-d}=\frac{\binom{2d}{d}}{(t-d)!}n^{t-d}-O(n^{t-d-1}).
 \end{align*}
\end{theorem}

%\textcolor{red}{In this work, we investigate the case $d = 2$ for general $q$-ary alphabets. Our main result provides an exact formula for $N_q(n,2,t)$ for all $t \ge 2, q\geq 4$, and sufficiently large $n$.}

\section{Properties of Deletion Balls and Their Intersections}
\label{sec3}

In this section, we study properties of deletion balls $D_t(\mathbf{x})$ and their intersections $D_t(\mathbf{x}) \cap D_t(\mathbf{y})$ for $\mathbf{x}, \mathbf{y} \in \Sigma_q^n$. We also investigate properties of the maximum size $D_q(n,t)$ of a $t$-deletion ball and characterize the sequences that attain this extremal value, both exactly and asymptotically.

To analyze $|D_t(\mathbf{x})|$ for general $q\ge 2$, we extend the combinatorial decompositions introduced in~\cite{Gabrys} to arbitrary alphabets. 
The following lemmas enable us to partition deletion balls according to prescribed prefixes and suffixes.

\begin{lemma}
Let $n,m_1,m_2$ be positive integers with $m_1+m_2 \le n$, and let $q \ge 2$. For any $\mathbf{u} \in \Sigma_q^n$ and $0 \le t \le n$,
\[
|D_t(\mathbf{u})| = \sum_{\mathbf{u}_1 \in \Sigma_q^{m_1}} \sum_{\mathbf{u}_2 \in \Sigma_q^{m_2}} |D_t(\mathbf{u})_{\mathbf{u}_2}^{\mathbf{u}_1}|.
\]
\label{Pro:lm1}
\end{lemma}

\begin{lemma}
Let $n,m_1,m_2,t$ be positive integers with $m_1+m_2 \le n$, and $q\geq 2$, and let $\mathbf{u} = (u_1,\dots,u_n) \in \Sigma_q^n$, $\mathbf{u}_1 \in \Sigma_q^{m_1}$, $\mathbf{u}_2 \in \Sigma_q^{m_2}$. Define $k_1$ as the smallest index for which $\mathbf{u}_1$ is a subsequence of $(u_1,\dots,u_{k_1})$,
and $k_2$ as the largest index for which $\mathbf{u}_2$ is a subsequence of $(u_{k_2},\dots,u_n)$.  If $k_1 < k_2$, then
\[
D_t(\mathbf{u})_{\mathbf{u}_2}^{\mathbf{u}_1} = \mathbf{u}_1 \circ D_{t^*}(u_{k_1+1},\dots,u_{k_2-1}) \circ \mathbf{u}_2,
\]
where $t^* = t - (k_1 - m_1) - (n - k_2 + 1 - m_2)$ (with the understanding that $D_{t^*}(\cdot)=\emptyset$ when $t^*<0$). Consequently,
\[
|D_t(\mathbf{u})_{\mathbf{u}_2}^{\mathbf{u}_1}| = |D_{t^*}(u_{k_1+1},\dots,u_{k_2-1})|.
\]
\label{Pro:lm2}
\end{lemma}

To describe the intersections $D_t(\mathbf{x}) \cap D_t(\mathbf{y})$ for $\mathbf{x}, \mathbf{y} \in \Sigma_q^n$, we use a result from~\cite{Chrisnata} (Proposition~16), which naturally generalizes to arbitrary $q$-ary alphabets ($q\ge 2$) as follows.

\begin{lemma}
\label{Pro:lm3}
Let \( \mathbf{x} = (\mathbf{u},\mathbf{c},\mathbf{v}) \) and \( \mathbf{y} = (\mathbf{u},\mathbf{d},\mathbf{v}) \), where \( \mathbf{u}, \mathbf{c}, \mathbf{d},\) and \( \mathbf{v}\) are $q$-ary sequences with $q\geq 2$, and let $t,s \ge 0$. Then
\[
D_t(\mathbf{x}) \cap D_s(\mathbf{y}) = \bigcup_{a+b \leq \min\{t,s\}} D_a(\mathbf{u}) \circ \left( D_{t-a-b}(\mathbf{c}) \cap D_{s-a-b}(\mathbf{d}) \right) \circ D_b(\mathbf{v}),
\]
with the understanding that $D_r(\cdot)=\emptyset$ whenever $r<0$.
\end{lemma}

Let \(\mathbf{x} = (x_1,\dots,x_{n+k}) \in \Sigma_q^{n+k}\), \(\mathbf{y} = (y_1,\dots,y_n) \in \Sigma_q^n\), and define \(\mathcal{X} = D_{k+t}(\mathbf{x}) \cap D_t(\mathbf{y})\) with \(n \ge t \ge 1\) and \(k \ge 0\). By Lemmas \ref{Pro:lm1} and \ref{Pro:lm2}, for any \(a \in \Sigma_q\) we have
\[
|\mathcal{X}^a| = \bigl| D_{t+k-\ell}(x_{2+\ell},\dots,x_{n+k}) \cap D_{t-\ell^*}(y_{2+\ell^*},\dots,y_n) \bigr|,
\]
where \(\ell+1\) and \(\ell^*+1\) are non‑negative integers denoting the first occurrence of the symbol \(a\) in \(\mathbf{x}\) and \(\mathbf{y}\), respectively (i.e., \(x_{\ell+1}=a\) and \(y_{\ell^*+1}=a\)). Unless otherwise stated, \(\ell\) and \(\ell^*\) are defined in this way during the decomposition of \(\mathcal{X}\).

For any two sequences $\mathbf{x},\mathbf{y}\in\Sigma_q^n$, applying a permutation or reversal to both sequences does not affect their Levenshtein distance or the intersection of their deletion balls.

\begin{lemma}\label{Pro:lm4}
For any two sequences $\mathbf{x},\mathbf{y}\in\Sigma_q^n$, we have
\begin{align*}
d_L(\mathbf{x},\mathbf{y}) = d_L(\mathbf{x}^R,\mathbf{y}^R) = d_L(\pi(\mathbf{x}),\pi(\mathbf{y})),\qquad
|D_t(\mathbf{x})\cap D_t(\mathbf{y})| = |D_t(\mathbf{x}^R)\cap D_t(\mathbf{y}^R)| = |D_t(\pi(\mathbf{x}))\cap D_t(\pi(\mathbf{y}))|. 
\end{align*}
\end{lemma}
\begin{proof}
Let \(F\) denote either the reversal map \(\mathbf{z}\mapsto\mathbf{z}^R\) or a componentwise permutation \(\mathbf{z}\mapsto\pi(\mathbf{z})\) for some \(\pi\in S_q\) and any $\mathbf{z}\in \Sigma_q^n$. In both cases \(F:\Sigma_q^n\to\Sigma_q^n\) is a bijection and commutes with deletions in the following sense: for any \(\mathbf{z}\in\Sigma_q^n\),  
\(
F(D_t(\mathbf{z})) = D_t(F(\mathbf{z})).
\)
Hence \(F\) maps the intersection \(D_t(\mathbf{x})\cap D_t(\mathbf{y})\) bijectively onto \(D_t(F(\mathbf{x}))\cap D_t(F(\mathbf{y}))\), so their cardinalities are equal.

For the Levenshtein distance, recall that it equals the minimum number of deletions needed to make the two sequences identical, i.e.,  
\(
d_L(\mathbf{x},\mathbf{y}) = n - \max\{|\mathbf{z}| : \mathbf{z}\text{ is a common subsequence of }\mathbf{x}\text{ and }\mathbf{y}\}.
\)
Since \(F\) preserves subsequence relations (a subsequence of \(F(\mathbf{x})\) corresponds to the image under \(F\) of a subsequence of \(\mathbf{x}\)), the length of the longest common subsequence is unchanged. Therefore \(d_L(\mathbf{x},\mathbf{y}) = d_L(F(\mathbf{x}),F(\mathbf{y}))\).
\end{proof}

\subsection{Cyclic operation on prefixes or suffixes}
Let \(\mathbf{x} = (\mathbf{u},\mathbf{c},\mathbf{v})\) and \(\mathbf{y} = (\mathbf{u},\mathbf{d},\mathbf{v})\) be as in Lemma~\ref{Pro:lm3}. 
For fixed \(\mathbf{c}\) and \(\mathbf{d}\), the size of the intersection \(D_t(\mathbf{x}) \cap D_t(\mathbf{y})\) can be increased by an appropriate choice of \(\mathbf{u}\) and \(\mathbf{v}\). To explain this phenomenon, we first introduce several basic operations on sequences. To maximize the size of the intersection \(D_t(\mathbf{x}) \cap D_t(\mathbf{y})\), we introduce a cyclic operation on prefixes or suffixes of $q$-ary sequences. 

\begin{definition}[Cyclic operation on prefixes and suffixes] 
For any sequence \(\mathbf{u} = (u_1,\dots,u_s) \in \Sigma_q^s\) with $q\ge 2$ and $s\geq 1$, the \emph{cyclic operation on prefixes} transforms it into a periodic sequence \(\mathbf{u}^c = (c_1,\dots,c_s) \in \mathbf{C}_q(s)\) through the following procedure. 
Define a chain of vectors
\[
\mathbf{u} = \mathbf{u}^{(1)} \to \mathbf{u}^{(2)} \to \cdots \to \mathbf{u}^{(s)} = \mathbf{u}^c,
\]
where each transition \(\mathbf{u}^{(r)} \to \mathbf{u}^{(r+1)}\) ($1\leq r\leq s-1$) is defined according to the order of operations listed below. For each \(i\in[1,s]\), write \(\mathbf{u}^{(i)}=(u_1^{(i)},\dots,u_s^{(i)})\). (Note that \(\mathbf{u}^{(1)}=\mathbf{u}\) by definition.)

\begin{enumerate}
    \item \textbf{Distinctness of the initial segment.}
    For \(r=1,\dots,\min\{q,s\}-1\), set \(p = \min\{q,s\}-r\) (processing the first \(\min\{q,s\}\) positions from right to left).  Define \(\mathbf{u}^{(r+1)}\) from \(\mathbf{u}^{(r)}\) by modifying only position \(p\) as follows:
    \begin{itemize}
      \item if \(u_p^{(r)}\) already appears among \(\{u_{p+1}^{(r)}, \dots, u_{\min\{q,s\}}^{(r)}\}\), set \(u_p^{(r+1)}\) to an arbitrary symbol in 
            \(\Sigma_q \setminus \{u_{i}^{(r)}|1\leq i\leq \min\{q,s\}\}\) (such a symbol exists because the prefix contains at most \(q-1\) distinct symbols);
      \item otherwise, set \(u_p^{(r+1)} = u_p^{(r)}\).
    \end{itemize}
    All other positions remain unchanged. That is, \(u_k^{(r+1)} = u_k^{(r)}\) for \(k\in [p-1]\cup [p+1,s]\).
    \item \textbf{Periodicity (only when \(s>q\)).}
    For \(r=q,\dots,s-1\), set \(p=r+1\) (processing indices \(r+1\in [q+1,s]\) from left to right). Define \(\mathbf{u}^{(r+1)}\) from \(\mathbf{u}^{(r)}\) as follows:
    \begin{itemize}
      \item if \(u_{p}^{(r)} \neq u_{p-q}^{(r)}\), then \(\mathbf{u}_{[1,p-1]}^{(r+1)}\) is obtained by swapping the elements \(u_{p}^{(r)}\) and \(u_{p-q}^{(r)}\) within the prefix \(\mathbf{u}_{[1,p-1]}^{(r)}\). For all \(k\in [p,s]\), set \(u_k^{(r+1)}=u_k^{(r)}\);
      \item otherwise, \(\mathbf{u}^{(r+1)}=\mathbf{u}^{(r)}\).
    \end{itemize}
\end{enumerate}

Finally, set \(\mathbf{u}^c = \mathbf{u}^{(s)}\).

Moreover, the \emph{cyclic operation on suffixes} is defined analogously. Applying the cyclic operation on prefixes to \(\mathbf{u}^R\) yields a sequence \((\mathbf{u}^R)^c=(c'_1,\dots,c'_s)\in\mathbf{C}_q(s)\), via the chain
\( \mathbf{u}^R = (\mathbf{u}^R)^{(1)} \to \cdots \to (\mathbf{u}^R)^{(s)} = (\mathbf{u}^{R})^c.\) Define $\mathbf{u}^{cR}=((\mathbf{u}^R)^c)^R$.
Reversing this chain yields
\(
\mathbf{u} = \big((\mathbf{u}^R)^{(1)}\big)^R \to \cdots \to \big((\mathbf{u}^R)^{(s)}\big)^R = \big((\mathbf{u}^R)^{c}\big)^R=\mathbf{u}^{cR}.
\)
\label{Pro:def1}
\end{definition}

%$(\mathbf{u}^{cR})^R=$ where \(\mathbf{u}^{cR}=((\mathbf{u}^R)^c)^R=((\mathbf{u}^R)^{(s)})^R\),

We now illustrate the cyclic operation on prefixes and suffixes (Definition~\ref{Pro:def1}) with the following example: a non‑periodic prefix sequence \(\mathbf{u}\) is transformed into \(\mathbf{u}^{c}\), and a non‑periodic suffix sequence \(\mathbf{v}\) into \(\mathbf{v}^{cR}\).

\begin{example}
\label{Pro:ex1}
Let $q=4$, $s=8$, and take the initial vector
\(
\mathbf{v}=\mathbf{u}=\mathbf{u}^{(1)}=(1,1,3,3,3,2,3,2).
\)
The cyclic operation on prefixes transforms $\mathbf{u}^{(1)}$ into the periodic vector $\mathbf{u}^{c}=\mathbf{u}^{(8)}=(1,0,3,2,1,0,3,2)$ through the following chain:
\(
\mathbf{u}^{(1)} \to \mathbf{u}^{(2)} \to \cdots \to \mathbf{u}^{(8)}.
\)
The cyclic operation on suffixes transforms $\mathbf{v}= \big((\mathbf{v}^R)^{(1)}\big)^R$ into the periodic vector $\mathbf{v}^{cR}=\big((\mathbf{v}^R)^{(8)}\big)^R=(1,0,3,2,1,0,3,2)$ through the following chain:
\(
 \big((\mathbf{v}^R)^{(1)}\big)^R \to \big((\mathbf{v}^R)^{(2)}\big)^R \to \cdots \to \big((\mathbf{v}^R)^{(8)}\big)^R.
\)

For reference, we also define two fixed subsequences
\(
\mathbf{c}=(0,1,3,0,1,2,3,1,2,0,1,2)\) and 
\(\mathbf{d}=(1,0,3,1,0,2,3,1,2,0,2,1).
\)
We set $\mathcal{X}_{(i)} = D_7(\mathbf{u}^{(i)},\mathbf{c}) \cap D_7(\mathbf{u}^{(i)},\mathbf{d})$ and $\mathcal{Y}_{(i)} = D_7\big(\mathbf{c},\big((\mathbf{v}^R)^{(i)}\big)^R\big) \cap D_7\big(\mathbf{d},\bigl((\mathbf{v}^R)^{(i)}\bigr)^R\big)$ for each $i\in [8]$. The transforms of $\mathbf{u}^{(i)}$ and $\bigl((\mathbf{v}^R)^{(i)}\bigr)^R$ are detailed in Table~\ref{Pro:tab1}, which also records the sizes $|\mathcal{X}_{(i)}|$ and $|\mathcal{Y}_{(i)}|$ at each stage.

\begin{table}[htbp]
  \centering
  \caption{Transformation of \(\mathbf{u}^{(1)}\) into the periodic vector \(\mathbf{u}^{(8)}\) and of \(\bigl((\mathbf{v}^R)^{(1)}\bigr)^R\) into the periodic vector \(\bigl((\mathbf{v}^R)^{(8)}\bigr)^R\)}
  \label{Pro:tab1}
  \begin{tabular}{lcccr}
    \toprule
    $i$ & $\mathbf{u}^{(i)}$ &  $|\mathcal{X}_{(i)}|$  & $\big((\mathbf{v}^R)^{(i)}\big)^R$ & $|\mathcal{Y}_{(i)}|$ \\
    \midrule
    $1$ & $(1,1,3,3,3,2,3,2)$ & $5891$ & $(1,1,3,3,3,2,3,2)$ & $4842$ \\
    $2$ & $(1,1,2,3,3,2,3,2)$ & $8219$ & $(1,1,3,3,3,2,3,2)$ & $4842$ \\
    $3$ & $(1,1,2,3,3,2,3,2)$ & $8219$ & $(1,1,3,3,3,2,1,2)$ & $5336$ \\
    $4$ & $(0,1,2,3,3,2,3,2)$ & $10614$& $(1,1,3,3,3,2,1,0)$ & $5730$ \\
    $5$ & $(3,1,2,0,3,2,3,2)$ & $13813$& $(1,1,3,3,0,2,1,3)$ & $8023$ \\
    $6$ & $(3,2,1,0,3,2,3,2)$ & $13931$& $(1,1,3,1,0,2,3,1)$ & $9064$ \\
    $7$ & $(1,2,3,0,1,2,3,2)$ & $14687$& $(1,1,3,2,0,1,3,2)$ & $10298$ \\
    $8$ & $(1,0,3,2,1,0,3,2)$ & $15380$& $(1,0,3,2,1,0,3,2)$ & $13524$ \\
    \bottomrule\end{tabular}
\end{table}
\end{example}

\begin{remark}
By Definition~\ref{Pro:def1}, the cyclic operation on prefixes maps \(\mathbf{u}\) to a periodic sequence \(\mathbf{u}^c \in \mathbf{C}_q(s)\) via the chain
\[
\mathbf{u} = \mathbf{u}^{(1)} \to \mathbf{u}^{(2)} \to \cdots \to \mathbf{u}^{(s)} = \mathbf{u}^c .
\]
Let \(p =\min\{s,q\}\). The operation ensures that the prefix \(\mathbf{u}^{(p)}_{[1,p]}\) consists of distinct symbols (the first step towards periodicity).  

For \(i \ge q+1\) (which requires \(s \ge q+1\)), the prefix \(\mathbf{u}^{(i-1)}_{[1,i-1]}\) is already a periodic sequence.  
If \(u_i^{(i-1)} \neq u_{i-q}^{(i-1)}\), swapping the two elements \(u_i^{(i-1)}\) and \(u_{i-q}^{(i-1)}\) makes \(\mathbf{u}^{(i)}_{[1,i]}\) periodic as well, because after the swap we have \(u_{i-q}^{(i)} = u_i^{(i-1)} = u_i^{(i)}\).  
If \(u_i^{(i-1)} = u_{i-q}^{(i-1)}\), then \(\mathbf{u}^{(i)} = \mathbf{u}^{(i-1)}\) and the prefix \(\mathbf{u}^{(i)}_{[1,i]}\) is already periodic.  
Thus for every \(i \ge q+1\), the prefix \(\mathbf{u}^{(i)}_{[1,i]}\) becomes a periodic sequence. Consequently, for \(i = s\) the whole sequence \(\mathbf{u}^{(s)} = \mathbf{u}^c\) is periodic, i.e., \(\mathbf{u}^c \in \mathbf{C}_q(s)\).

By Definition~\ref{Pro:def1}, the cyclic operation on suffixes maps \(\mathbf{u}\) to a periodic sequence \(\mathbf{u}^{cR}\in\mathbf{C}_q(s)\). Equivalently, this operation is obtained by applying the cyclic operation on prefixes to the reversal \(\mathbf{u}^R\): it produces \((\mathbf{u}^R)^c\) via the chain
\[
\mathbf{u}^{R}=(\mathbf{u}^{R})^{(1)}\to(\mathbf{u}^{R})^{(2)}\to\cdots\to(\mathbf{u}^{R})^{(s)}=(\mathbf{u}^{R})^{c}.
\]
\label{Pro:rmk1}
\end{remark}

The preceding remark shows that the cyclic operation on prefixes eventually yields a periodic sequence. The following lemma establishes that along this chain the size of the deletion ball intersection does not decrease.

\begin{lemma}[Cyclic operation on prefixes increases the intersection size]
\label{Pro:lm5}
Let \(n > t \ge 1\), \(q\geq 2, s\geq 1\), and \(k \ge 0\) be integers.
Consider two sequences
\(
(\mathbf{u}, \mathbf{c}) \in \Sigma_q^{n+k}, 
(\mathbf{u}, \mathbf{d}) \in \Sigma_q^{n},
\)
where
\(
\mathbf{u} = (u_1,\dots,u_s) \in \Sigma_q^s, 
\mathbf{c}=(c_1,\dots,c_{n+k-s}) \in \Sigma_q^{n+k-s}, 
\mathbf{d}=(d_1,\cdots,d_{n-s})\in \Sigma_q^{n-s}.
\)
Let \(\mathbf{u}^c \in \mathbf{C}_q(s)\) be the periodic sequence obtained from \(\mathbf{u}\) via the cyclic operation (Definition~\ref{Pro:def1}), with the intermediate chain
\[
\mathbf{u} = \mathbf{u}^{(1)} \to \mathbf{u}^{(2)} \to \cdots \to \mathbf{u}^{(s)} = \mathbf{u}^c .
\]
Then for every $i\in [s-1]$,
\begin{equation}
\bigl| D_{t+k}(\mathbf{u}^{(i)},\mathbf{c}) \cap D_t(\mathbf{u}^{(i)},\mathbf{d}) \bigr|
\le
\bigl| D_{t+k}(\mathbf{u}^{(i+1)}, \mathbf{c}) \cap D_t(\mathbf{u}^{(i+1)}, \mathbf{d}) \bigr|.\nonumber
\end{equation}
In particular,
\begin{equation}
\bigl| D_{t+k}(\mathbf{u},\mathbf{c}) \cap D_t(\mathbf{u},\mathbf{d}) \bigr|
\le
\bigl| D_{t+k}(\mathbf{u}^c, \mathbf{c}) \cap D_t(\mathbf{u}^{c}, \mathbf{d}) \bigr|.\nonumber
\end{equation}
\end{lemma}

\begin{proof}
For convenience, for $i \in [s]$ set $\mathbf{u}^{(i)} = (u_1^{(i)},\dots,u_s^{(i)})$ and
$\mathcal{X}_{(i)} = D_{t+k}(\mathbf{u}^{(i)}, \mathbf{c}) \cap D_t(\mathbf{u}^{(i)}, \mathbf{d})$.

\emph{Case $1\leq s\leq q$:} 

If $s = 1$ then $\mathbf{u}^c = \mathbf{u}$ and the claim is trivial; hence we assume $s \ge 2$ in the following.

We prove by induction on $2 \le s \le q$ that for all $i \in [s-1]$,
\begin{equation}
|\mathcal{X}_{(i)}| \le |\mathcal{X}_{(i+1)}|. \label{Pro:eq1}
\end{equation}
For $i \in [s-1]$, we define $A^{(i)} = \{ u_j^{(i)} \mid j \in [s-1-i] \}$,
$B^{(i)} = \{ u_j^{(i)} \mid j \in [s+1-i,\,s] \}$,
$a_i = u^{(i)}_{s-i}$, and $a'_i = u^{(i+1)}_{s-i}$.
By the cyclic operation, if $a_i \notin B^{(i)}$ then $a'_i = a_i$ and $\mathbf{u}^{(i+1)} = \mathbf{u}^{(i)}$;
otherwise $a_i\in B^{(i)}$, $a'_i \in \Sigma_q \setminus \big(B^{(i)}\cup A^{(i)}\big)$, and $u_j^{(i+1)} = u_j^{(i)}$ for all $j \in [s] \setminus \{s-i\}$.

\emph{Base case $s = 2$.}

For $s = 2$, $B^{(1)} = \{u_2^{(1)}\}$ and $A^{(1)}=\emptyset$.

If $a_1 \notin B^{(1)}$ (i.e., $u_1^{(1)} \neq u_2^{(1)}$), then $\mathbf{u}^{(2)} = \mathbf{u}^{(1)}$ and Eq.~\eqref{Pro:eq1} holds trivially.

If $a_1 \in B^{(1)}$ (i.e., $u_1^{(1)} = u_2^{(1)}$), then $a'_1 \neq a_1$ and $u_2^{(2)} = u_2^{(1)}$.
Decompose $\mathcal{X}_{(1)}$ and $\mathcal{X}_{(2)}$ according to the first symbol:
\begin{align*}
|\mathcal{X}_{(1)}| = |\mathcal{X}_{(1)}^{a_1}| + \sum_{j \in \Sigma_q \setminus \{a_1\}} |\mathcal{X}_{(1)}^{j}|,\qquad
|\mathcal{X}_{(2)}| = |\mathcal{X}_{(2)}^{a'_1}| + |\mathcal{X}_{(2)}^{a_1}| + \sum_{j \in \Sigma_q \setminus \{a_1,a'_1\}} |\mathcal{X}_{(2)}^{j}|.
\end{align*}
By Lemmas~\ref{Pro:lm1} and \ref{Pro:lm2}, for every $j \in \Sigma_q \setminus \{a_1,a'_1\}$ we have $|\mathcal{X}_{(1)}^{j}| = |\mathcal{X}_{(2)}^{j}|$.
Moreover,
\begin{align*}
|\mathcal{X}_{(1)}^{a'_1}| &= \bigl| D_{t+k-1-\ell_0}(\mathbf{c}_{[\ell_0+2,\,n+k-2]}) \cap D_{t-1-\ell_1}(\mathbf{d}_{[\ell_1+2,\,n-2]}) \bigr|,\qquad
|\mathcal{X}_{(1)}^{a_1}| = \bigl| D_{t+k}(u_2^{(1)},\mathbf{c}) \cap D_t(u_2^{(1)},\mathbf{d}) \bigr|,\\
|\mathcal{X}_{(2)}^{a'_1}| &= \bigl| D_{t+k}(u_2^{(2)},\mathbf{c}) \cap D_t(u_2^{(2)},\mathbf{d}) \bigr|,\qquad
|\mathcal{X}_{(2)}^{a_1}| = \bigl| D_{t+k-1}(\mathbf{c}) \cap D_{t-1}(\mathbf{d}) \bigr|,
\end{align*}
where $\ell_0,\ell_1 \ge 0$ are the smallest indices with $c_{1+\ell_0} = a'_1$ and $d_{1+\ell_1} = a'_1$, and $u_2^{(1)} = u_2^{(2)}$.

Thus $|\mathcal{X}_{(1)}^{a'_1}| \le |\mathcal{X}_{(2)}^{a_1}|$ and $|\mathcal{X}_{(1)}^{a_1}| = |\mathcal{X}_{(2)}^{a'_1}|$.
Combining these relations with the equality of the common terms yields
$|\mathcal{X}_{(1)}| \le |\mathcal{X}_{(2)}|$, which is exactly Eq.~\eqref{Pro:eq1} for $s = 2$. This establishes the base case.

\emph{Induction step}

Now assume that Eq.~\eqref{Pro:eq1} holds for all smaller sequence lengths, i.e., for every $l$ with $2 \le l \le s-1$ and for all $i \in [l-1]$, where $3 \le s \le q$. We prove it for $l = s$ and $i \in [s-1]$. Let $\mathbf{u} = (u_1,\dots,u_s)$.

Fix $i \in [s-1]$. If $a_i \notin B^{(i)}$, then $\mathbf{u}^{(i+1)} = \mathbf{u}^{(i)}$ and Eq.~\eqref{Pro:eq1} holds trivially.  

Assume now $a_i \in B^{(i)}$. By the cyclic operation, $a'_i \in \Sigma_q \setminus (A^{(i)} \cup B^{(i)})$ and $u_j^{(i+1)} = u_j^{(i)}$ for all $j \neq s-i$. Decompose $\mathcal{X}_{(i)}$ and $\mathcal{X}_{(i+1)}$ according to the first symbol:
\begin{align*}
|\mathcal{X}_{(i)}|   &= \sum_{j \in A^{(i)}\setminus\{a_i\}}\!\!|\mathcal{X}_{(i)}^{j}| \;+\; |\mathcal{X}_{(i)}^{a_i}| \;+\; |\mathcal{X}_{(i)}^{a'_i}| \;+\!\! \sum_{j \in \Sigma_q \setminus (A^{(i)}\cup\{a_i,a'_i\})}\!\!|\mathcal{X}_{(i)}^{j}|,\\
|\mathcal{X}_{(i+1)}| &= \sum_{j \in A^{(i)}\setminus\{a_i\}}\!\!|\mathcal{X}_{(i+1)}^{j}| \;+\; |\mathcal{X}_{(i+1)}^{a_i}| \;+\; |\mathcal{X}_{(i+1)}^{a'_i}| \;+\!\! \sum_{j \in \Sigma_q \setminus (A^{(i)}\cup\{a_i,a'_i\})}\!\!|\mathcal{X}_{(i+1)}^{j}|.
\end{align*}
We compare the corresponding terms in four groups.

\begin{itemize}
\item \emph{Symbols not in $A^{(i)}\cup\{a_i,a'_i\}$.}
For every $j \in \Sigma_q \setminus (A^{(i)}\cup\{a_i,a'_i\})$, the suffix $\mathbf{u}^{(i+1)}_{[s+1-i,s]}$ equals $\mathbf{u}^{(i)}_{[s+1-i,s]}$, and the first occurrence of $j$ is the same in both sequences (if it exists). Hence
\[
|\mathcal{X}_{(i)}^{j}| = |\mathcal{X}_{(i+1)}^{j}|.
\]

\item \emph{Symbols in $A^{(i)}$.}
Take $j \in A^{(i)}\setminus\{a_i\}$ and let $\ell_j \ge 0$ be the smallest index such that $u^{(i)}_{1+\ell_j}=u^{(i+1)}_{1+\ell_j}=j$. Then
\begin{align*}
\mathcal{X}_{(i)}^{j} &=
D_{t+k-\ell_j}\bigl(\mathbf{u}^{(i)}_{[\ell_j+2,\,s-1-i]},\, a_i,\, \mathbf{u}^{(i)}_{[s+1-i,s]},\, \mathbf{c}\bigr)
\cap
D_{t-\ell_j}\bigl(\mathbf{u}^{(i)}_{[\ell_j+2,\,s-1-i]},\, a_i,\, \mathbf{u}^{(i)}_{[s+1-i,s]},\, \mathbf{d}\bigr),\\[4pt]
\mathcal{X}_{(i+1)}^{j} &=
D_{t+k-\ell_j}\bigl(\mathbf{u}^{(i+1)}_{[\ell_j+2,\,s-1-i]},\, a'_i,\, \mathbf{u}^{(i+1)}_{[s+1-i,s]},\, \mathbf{c}\bigr)
\cap
D_{t-\ell_j}\bigl(\mathbf{u}^{(i+1)}_{[\ell_j+2,\,s-1-i]},\, a'_i,\, \mathbf{u}^{(i+1)}_{[s+1-i,s]},\, \mathbf{d}\bigr).
\end{align*}
The middle parts $\mathbf{u}^{(i)}_{[\ell_j+2,\,s-1-i]}$ and $\mathbf{u}^{(i+1)}_{[\ell_j+2,\,s-1-i]}$ are identical, and the suffixes $\mathbf{u}^{(i)}_{[s+1-i,s]}$, $\mathbf{u}^{(i+1)}_{[s+1-i,s]}$ are also identical. The only difference is that the symbol following that middle part is $a_i$ in $\mathbf{u}^{(i)}$ and $a'_i$ in $\mathbf{u}^{(i+1)}$.  
After removing the common prefix of length $\ell_j+1$, we reduce the problem to a pair of sequences of length $s-\ell_j-1 < s$. Applying the induction hypothesis to these shorter sequences (with the two candidates $a_i$ and $a'_i$ playing the role of the distinguished symbols) yields
\[
|\mathcal{X}_{(i)}^{j}| \le |\mathcal{X}_{(i+1)}^{j}|.
\]

\item \emph{The new symbol $a'_i$.}
Since $a'_i$ does not appear in $\mathbf{u}^{(i)}$, its first occurrence in $\mathbf{u}^{(i)}\circ\mathbf{c}$ and $\mathbf{u}^{(i)}\circ\mathbf{d}$ must be inside $\mathbf{c}$ and $\mathbf{d}$ respectively. Let $\ell_0,\ell_1 \ge 0$ be the smallest indices with $c_{1+\ell_0}=a'_i$ and $d_{1+\ell_1}=a'_i$. Then
\[
\mathcal{X}_{(i)}^{a'_i} =
D_{t+k-s-\ell_0}\bigl(\mathbf{c}_{[\ell_0+2,\,n+k-s]}\bigr)
\cap
D_{t-s-\ell_1}\bigl(\mathbf{d}_{[\ell_1+2,\,n-s]}\bigr).
\]
In $\mathbf{u}^{(i+1)}$, the symbol $a'_i$ appears at position $s-i$, so after extracting it we are left with the suffix $\mathbf{u}^{(i+1)}_{[s+1-i,s]}$:
\[
\mathcal{X}_{(i+1)}^{a'_i} =
D_{t+k-s+i+1}\bigl(\mathbf{u}^{(i+1)}_{[s+1-i,s]},\, \mathbf{c}\bigr)
\cap
D_{t-s+i+1}\bigl(\mathbf{u}^{(i+1)}_{[s+1-i,s]},\, \mathbf{d}\bigr).
\]

\item \emph{The original symbol $a_i$.}
If $a_i \in A^{(i)}$, then the same induction argument as for $j \in A^{(i)}$ already gives $|\mathcal{X}_{(i)}^{a_i}| \le |\mathcal{X}_{(i+1)}^{a_i}|$.  
Otherwise $a_i \in B^{(i)}$, which means $a_i$ occurs in the suffix of length $i$ of $\mathbf{u}^{(i)}$. By definition,
\[
\mathcal{X}_{(i)}^{a_i} =
D_{t+k-s+i+1}\bigl(\mathbf{u}^{(i)}_{[s+1-i,s]},\, \mathbf{c}\bigr)
\cap
D_{t-s+i+1}\bigl(\mathbf{u}^{(i)}_{[s+1-i,s]},\, \mathbf{d}\bigr).
\]
In $\mathbf{u}^{(i+1)}$, since $a_i\in B^{(i)}$, we let $\ell_2 \ge 0$ be the smallest index such that $u^{(i+1)}_{s+1+\ell_2-i}=a_i$. Then
\[
\mathcal{X}_{(i+1)}^{a_i} =
D_{t+k-s+i-\ell_2}\bigl(\mathbf{u}^{(i+1)}_{[s+2+\ell_2-i,\,s]},\, \mathbf{c}\bigr)
\cap
D_{t-s+i-\ell_2}\bigl(\mathbf{u}^{(i+1)}_{[s+2+\ell_2-i,\,s]},\, \mathbf{d}\bigr).
\]
From the above estimates we obtain
\[
|\mathcal{X}_{(i)}^{a'_i}| + |\mathcal{X}_{(i)}^{a_i}|
\;\le\;
|\mathcal{X}_{(i+1)}^{a'_i}| + |\mathcal{X}_{(i+1)}^{a_i}|.
\]
\end{itemize}

Summing the inequalities for all four groups, we conclude
\[
|\mathcal{X}_{(i)}| \le |\mathcal{X}_{(i+1)}|,
\]
which proves Eq.~\eqref{Pro:eq1} for $l = s$ and $i \in [s-1]$.

\emph{Case $s\geq q+1$}:

For $1 \le i \le q-1$, the cyclic operation only involves the first $q$ coordinates of $\mathbf{u}^{(i)}$.
We may view the pair $\bigl(\mathbf{u}_{[q+1,\,s]},\mathbf{c}\bigr)$ as a new suffix $\mathbf{c}'$ and $\bigl(\mathbf{u}_{[q+1,\,s]},\mathbf{d}\bigr)$ as $\mathbf{d}'$.
With this identification the situation reduces exactly to the case $s = q$ (already proved), so we obtain
\[
|\mathcal{X}_{(i)}| \le |\mathcal{X}_{(i+1)}|,
\]
for $1\leq i\leq q-1$.

We prove by induction on \(s\) that for all \(i \in [q,s-1]\),
\begin{equation}
|\mathcal{X}_{(i)}| \le |\mathcal{X}_{(i+1)}|. \label{Pro:eq2}
\end{equation}

Assume that Eq. (\ref{Pro:eq2}) holds for all \(l\) with \(q+1 \le l \le s-1\) and for all \(i \in [q,l-1]\). We prove it for \(l = s\) and \(i \in [q,s-1]\) (the base case of the induction can be obtained when $s=q+1$ in the following).  To this end, write $\mathbf{u} = (u_1,\dots,u_{s})$.  

For \(i\in[q,s-1]\) and \(s\geq q+1\), the construction of the chain allows us to write  
\[
\mathbf{u}^{(i)} = \bigl(\mathbf{e}_i,\; u_{i+1}^{(i)},\; u_{[i+2,s]}^{(i)}\bigr),
\]
where \(\mathbf{e}_i=(e_1,\dots,e_i)\in \mathbf{C}_q(i)\) is a periodic sequence. Consequently, \(\{e_1,\dots,e_q\}=\Sigma_q\).

If \(u_{i+1}^{(i)} = u_{i+1}^{(i+1-q)} = e_{i+1-q}\), then \(\mathbf{u}^{(i)}=\mathbf{u}^{(i+1)}\), and hence Eq.~(\ref{Pro:eq2}) holds trivially.

Now suppose \(u_{i+1}^{(i)} \neq e_{i+1-q}\). Then \(\mathbf{u}^{(i+1)}\) is obtained from \(\mathbf{u}^{(i)}\) by swapping the two elements \(e_{i-q+1}\) and \(u_{i+1}^{(i)}\) within the prefix \(\mathbf{u}_{[1,i]}^{(i)}\), so that \(\mathbf{u}_{[1,i+1]}^{(i+1)}\) is periodic. Moreover, there exist indices \(p,r\in[q]\) such that \(e_p=u_{i+1}^{(i)}\), \(e_r=e_{i+1-q}\), and \(r\ne p\). Under this swap, within the prefix \(\mathbf{u}_{[1,i]}^{(i)}\), the element \(e_r\) is replaced by \(e_p\), and \(e_p\) is replaced by \(e_r\); the suffix \(\mathbf{u}_{[i+1,s]}^{(i)}\) remains unchanged, i.e., \(\mathbf{u}_{[i+1,s]}^{(i)}=\mathbf{u}_{[i+1,s]}^{(i+1)}\).

Define
\(
m_0=\min\{p,r\}, m_1=\max\{p,r\}, r_1=i+1-q,
\)
and let \(p_1=i+1-q+\hat p\) for some \(\hat p\in[q]\) such that \(p_1\equiv p\pmod q\). In other words, \(r_1\) and \(p_1\) are the rightmost indices such that \(u_{r_1}=e_r\) and \(u_{p_1}=e_p\) within the prefix \(\mathbf{u}_{[1,i]}^{(i)}\). Furthermore, $r_1<p_1$.

Decomposing \(\mathcal{X}_{(i)}\) and \(\mathcal{X}_{(i+1)}\) according to the first symbol yields
\(
|\mathcal{X}_{(i)}| = \sum_{j \in [q]} |\mathcal{X}_{(i)}^{e_j}|~\text{and}~
|\mathcal{X}_{(i+1)}| = \sum_{j \in [q]} |\mathcal{X}_{(i+1)}^{e_j}|,
\)
where
\[
\begin{aligned}
\mathcal{X}_{(i)}^{e_j} &= D_{t+k-j+1}\bigl(\mathbf{u}_{[j+1,s]}^{(i)},\; \mathbf{c}\bigr)
\cap D_{t-j+1}\bigl(\mathbf{u}_{[j+1,s]}^{(i)},\; \mathbf{d}\bigr), \\[4pt]
\mathcal{X}_{(i+1)}^{e_{m_0}} &= D_{t+k-m_1+1}\bigl(\mathbf{u}_{[m_1+1,s]}^{(i+1)},\; \mathbf{c}\bigr)
\cap D_{t-m_1+1}\bigl(\mathbf{u}_{[m_1+1,s]}^{(i+1)},\; \mathbf{d}\bigr), \\[4pt]
\mathcal{X}_{(i+1)}^{e_{m_1}} &= D_{t+k-m_0+1}\bigl(\mathbf{u}_{[m_0+1,s]}^{(i+1)},\; \mathbf{c}\bigr)
\cap D_{t-m_0+1}\bigl(\mathbf{u}_{[m_0+1,s]}^{(i+1)},\; \mathbf{d}\bigr), \\[4pt]
\mathcal{X}_{(i+1)}^{e_h} &= D_{t+k-h+1}\bigl(\mathbf{u}_{[h+1,s]}^{(i+1)},\; \mathbf{c}\bigr)
\cap D_{t-h+1}\bigl(\mathbf{u}_{[h+1,s]}^{(i+1)},\; \mathbf{d}\bigr), 
\end{aligned}
\]
for each \(j \in [q]\), \(h \in [q]\setminus\{m_0,m_1\}\). 

We compare the corresponding terms.
\begin{enumerate}
  \item For \(h\in[q]\setminus\{m_0,m_1\}\), we distinguish three subcases according to the position of \(h+1\) relative to \(r_1\) and \(p_1\).
\begin{itemize}
  \item If \(h+1\le r_1\), then \(\mathbf{u}_{[h+1,i]}^{(i)}\) contains both \(e_p\) and \(e_r\), and \(\mathbf{u}_{[h+1,i]}^{(i+1)}\) is obtained by swapping them. Hence, by the induction hypothesis on length \(s-h\) (\(<s\)), we have  
  \(
  |\mathcal{X}_{(i)}^{e_h}|\le |\mathcal{X}_{(i+1)}^{e_h}|.
  \)
  \item If \(r_1 < h+1 \le p_1\), then \(\mathbf{u}_{[h+1,i+1]}^{(i)}\) contains only \(e_p\) among \(\{e_r,e_p\}\). That is, we have \(u_{i+1}^{(i)}=u_{p_1}^{(i)}=e_p\) in the sequence \(\mathbf{u}_{[h+1,i+1]}^{(i)}\), and its length is \(i+1-(h+1)+1 \le i-r_1+1 = q\), since \(r_1=i+1-q\) and \(r_1<h+1<p_1<i+1\). In this sequence, \(\mathbf{u}_{[h+1,i+1]}^{(i+1)}\) is obtained by changing the element at position \(p_1\) in \(\mathbf{u}_{[h+1,i+1]}^{(i)}\)  (which is the same as the element at position \(i+1\)) from \(e_{p_1}\) to \(e_r\), so that all elements in \(\mathbf{u}_{[h+1,i+1]}^{(i+1)}\) are distinct (this is precisely the cyclic operation on prefixes for \emph{Case $s\leq q$}). Thus, by the result of \emph{Case \(s\le q\)}, we get
\[
|\mathcal{X}_{(i)}^{e_h}|\le |\mathcal{X}_{(i+1)}^{e_h}|.
\]
  \item If \(h+1>p_1\), then \(\mathbf{u}_{[h+1,s]}^{(i)}=\mathbf{u}_{[h+1,s]}^{(i+1)}\), and consequently  
  \(
  |\mathcal{X}_{(i)}^{e_h}|=|\mathcal{X}_{(i+1)}^{e_h}|.
  \)
\end{itemize}
  \item For \(j\in\{m_0,m_1\}\), a similar analysis based on the relative order of \(m_0,m_1\) with respect to \(r_1,p_1\) yields  
\[
|\mathcal{X}_{(i)}^{e_{m_1}}|\le |\mathcal{X}_{(i+1)}^{e_{m_0}}|,\qquad
|\mathcal{X}_{(i)}^{e_{m_0}}|\le |\mathcal{X}_{(i+1)}^{e_{m_1}}|.
\]
\end{enumerate}

Combining all these inequalities, we obtain
\[
|\mathcal{X}_{(i)}| \le |\mathcal{X}_{(i+1)}|.
\]
This establishes Eq.~(\ref{Pro:eq2}) for \(l=s\) and every \(i\in[q,s-1]\). Hence the lemma follows.
\end{proof}

The following corollary is an immediate application of Lemma~\ref{Pro:lm5} to the cyclic operation on suffixes via reversal.

\begin{corollary}
\label{Pro:cor1}
Let \(r,s\ge 1, q\geq 2,\) and \(m,k,t \ge 0\) be integers.
Consider sequences
\(
\mathbf{u} = (u_1,\dots,u_s) \in \Sigma_q^s, 
\mathbf{v} = (v_1,\dots,v_r) \in \Sigma_q^r,
\mathbf{c}=(c_1,\dots,c_{m+k}) \in \Sigma_q^{m+k}, 
\mathbf{d}=(d_1,\cdots,d_{m})\in \Sigma_q^{m}.
\)
Then
\begin{align*}
\bigl| D_{t+k}(\mathbf{c},\mathbf{v}) \cap D_t(\mathbf{d},\mathbf{v}) \bigr|
&\le
\bigl| D_{t+k}(\mathbf{c},\mathbf{v}^{cR}) \cap D_t(\mathbf{d},\mathbf{v}^{cR}) \bigr|,\\
\bigl| D_{t+k}(\mathbf{u},\mathbf{c},\mathbf{v}) \cap D_t(\mathbf{u},\mathbf{d},\mathbf{v}) \bigr|
&\le
\bigl| D_{t+k}(\mathbf{u}^c, \mathbf{c},\mathbf{v}^{cR}) \cap D_t(\mathbf{u}^{c}, \mathbf{d},\mathbf{v}^{cR}) \bigr|.
\end{align*}
\end{corollary}

\subsection{Properties of deletion balls}

We begin with two elementary inequalities for $D_q(n,t)$ with $q\geq 2$,
\begin{align}
D_q(n,t) &\le D_q(n+1,t+1), \label{Pro:eq3}\\
D_q(n,t) &\le D_q(n+1,t). \label{Pro:eq4}
\end{align}

The following corollary is a direct consequence of Eqs. \eqref{def:eq1}, \eqref{Pro:eq3}, and \eqref{Pro:eq4}.

\begin{corollary}
Let $t\geq 1, q \geq 2$, and $m\geq 1$. Define 
\(
N(n,t) = \sum_{i=1}^{m} c_i D_q(n-a_i, t-b_i),
\)
where $a_i, b_i, c_i$ are integers for $i \in [m]$.
For all $n>\max\{t+a_i-b_i\mid i\in[m]\}$, the following recurrence holds:
\[
N(n,t) = N(n-1,t) + N(n-2,t-1) + \cdots + N(n-q, t+1-q).
\]
Moreover, for $n\geq t+1$ we have 
\(
D_q(n,t) \le q \cdot D_q(n-1,t).
\)
\label{Pro:cor2}
\end{corollary}

We next establish a difference estimate that will be useful later.
\begin{lemma}
\label{Pro:lm6}
Let $t,s,i\ge 0$ be integers, and let $q\geq 2$.
For \(n \ge t+s\) and $t\ge i$,
\begin{align}
D_q(n,t)-D_q(n-s,t) &\ge D_q(n-i,t-i)-D_q(n-s-i,t-i). \nonumber
\end{align}
\end{lemma}
\begin{proof}
By the recurrence relation for \(D_q(n,t)\) and Eq. \((\ref{Pro:eq3})\), we have
\[
D_q(n,t) - D_q(n-1,t) \ge D_q(n-i,t-i) - D_q(n-1-i,t-i).
\]
For each \(j = 0, 1, \dots, s-1\), shifting the indices gives
\[
D_q(n-j,t) - D_q(n-j-1,t) \ge D_q(n-j-i,t-i) - D_q(n-j-1-i,t-i).
\]
Summing these inequalities for \(j = 0, 1, \dots, s-1\) yields 
\[
D_q(n,t) - D_q(n-s,t) \ge D_q(n-i,t-i) - D_q(n-s-i,t-i).
\]
Hence the lemma follows.
\end{proof}

We will need several known bounds for non-periodic sequences, as well as some auxiliary ones. It is well known from \cite{Hirschberg} that $D_q(n,t)=|D_t(\mathbf{x})|$ for any periodic sequence $\mathbf{x}\in\mathbf{C}_q(n)$ with $q\ge2$ and $n\ge t$. The next three lemmas give upper bounds when the sequence is not periodic.

\begin{lemma}[Levenshtein \cite{L0}]
\label{Pro:lm7}
Let $q\geq 2$. For any $\mathbf{x}\in \Sigma_q^n$ and $n\geq t$,
\begin{align*}
\binom{r(\mathbf{x})-t+1}{t}\leq \left|D_t(\mathbf{x})\right|\leq \binom{r(\mathbf{x})+t-1}{t}.
\end{align*}
\end{lemma}

\begin{lemma}[Gabrys and Yaakobi \cite{Gabrys}]
\label{Pro:lm8}
Let $\mathbf{x}\in \Sigma_2^n$. For all $n\geq t$, if $r(\mathbf{x})\leq n-1$, then
\begin{align*}
|D_t(\mathbf{x})|\leq D_2(n-2,t)+D_2(n-2,t-1)+D_2(n-4,t-2).
\end{align*}
\end{lemma}

\begin{lemma}[Wang, Li, and Fu \cite{Wang}]
\label{Pro:lm9}
Let $\mathbf{x}\in \Sigma_3^n$.  For $n\geq t+2$ and $\mathbf{x}\notin \mathbf{C}_3(n)$,
\begin{align*}
 |D_t(\mathbf{x})|\leq D_3(n-2,t)+D_3(n-2,t-1)+D_3(n-3,t-1)+D_3(n-3,t-2)+D_3(n-5,t-3).
\end{align*}
\end{lemma}

We now derive the asymptotic expansions of \(D_q(n,t)\), which will be used repeatedly in the characterizations below. %For convenience, we adopt the convention that \(1/k! = 0\) for any integer \(k<0\) in the asymptotic expansion.

\begin{lemma}
\label{Pro:lm10}
Let $q\geq 4$. For $t\geq 3$ and sufficiently large $n$,
\begin{align*}
D_q(n,t)&=\frac{n^t}{t!}-\frac{1}{2}\cdot\frac{n^{t-1}}{(t-2)!}+\frac{(3t-1)}{24}\cdot\frac{n^{t-2}}{(t-3)!}+O(n^{t-3}),\\
D_3(n,t)&=\frac{n^t}{t!}-\frac{1}{2}\cdot\frac{n^{t-1}}{(t-2)!}+\frac{(3t-25)}{24}\cdot\frac{n^{t-2}}{(t-3)!}+O(n^{t-3}),\\
D_2(n,t)&=\frac{n^t}{t!}-\frac{3}{2}\cdot\frac{n^{t-1}}{(t-2)!}+\frac{(27t^2-55t+50)}{24}\cdot\frac{n^{t-2}}{(t-2)!}+O(n^{t-3}).
\end{align*}
When $t=2$ and $n\geq 2$,
\begin{align*}
D_q(n,2)=\frac{n^2}{2}-\frac{n}{2},\qquad D_3(n,2)=\frac{n^2}{2}-\frac{n}{2}, \qquad D_2(n,2)= \frac{n^2}{2}-\frac{3n}{2}+2.
\end{align*}
Moreover, for $q\geq t+2$ and $n\geq t$, we have $D_q(n,t)=D_{q-1}(n,t)$. For \(q\le t+1\), we have
\begin{align*}
D_q(n,t)-D_{q-1}(n,t)=\frac{n^{t-q+2}}{(t-q+1)!}+O(n^{t-q+1}),
\end{align*}
for sufficiently large $n$.
\end{lemma}

\begin{proof}
Since $D_q(n,t)=\sum\limits_{i=0}^{t}\binom{n-t}{i}D_{q-1}(t,t-i)$, we have 
\begin{align*}
 D_q(n,t)&=\binom{n-t}{t}+D_{q-1}(t,1)\cdot\binom{n-t}{t-1}+D_{q-1}(t,2)\cdot\binom{n-t}{t-2}+\sum\limits_{i=0}^{t-3}\binom{n-t}{i}D_{q-1}(t,t-i)\\
 &=\frac{\prod_{i=0}^{t-1}(n-t-i)}{t!}+\frac{t\cdot\prod_{i=0}^{t-2}(n-t-i)}{(t-1)!}+\frac{\binom{t}{2}\cdot\prod_{i=0}^{t-3}(n-t-i)}{(t-2)!}+O(n^{t-3})\\
 &=\frac{n^t-\frac{(3t-1)t}{2}\cdot n^{t-1}+\frac{t(t-1)(27t^2-19t+2)}{24}\cdot n^{t-2}+O(n^{t-3})}{t!}+\frac{t\cdot n^{t-1}-\frac{t(t-1)(3t-2)}{2}\cdot n^{t-2}+O(n^{t-3})}{(t-1)!}\\
 &~~~~~+\frac{\frac{t(t-1)}{2}\cdot n^{t-2}+O(n^{t-3})}{(t-2)!}+O(n^{t-3})\\
 &=\frac{n^t}{t!}-\frac{1}{2}\cdot\frac{n^{t-1}}{(t-2)!}+\frac{(3t-1)}{24}\cdot\frac{n^{t-2}}{(t-3)!}+O(n^{t-3}),
\end{align*}
where $D_{q-1}(t,2)=\binom{t}{2}$ for $t\geq3$ and $q\geq 4$. When $t=2$, $q\geq 4$, and $n\geq 2$, we have
$D_q(n,t)=\frac{n^2}{2}-\frac{n}{2}.$

Similarly, since $D_3(n,t)=\sum\limits_{i=0}^{t}\binom{n-t}{i}\sum\limits_{j=0}^{t-i}\binom{i}{j}$, it follows that
\begin{equation*}
D_3(n,t)=\binom{n-t}{t}+t\binom{n-t}{t-1}+D_2(t,2)\binom{n-t}{t-2}+O(n^{t-3})=\frac{n^t}{t!}-\frac{1}{2}\cdot\frac{n^{t-1}}{(t-2)!}+\frac{(3t-25)}{24}\cdot\frac{n^{t-2}}{(t-3)!}+O(n^{t-3}),
\end{equation*}
where $D_2(t,2)=\frac{t^2-3t+4}{2}$ for $t\geq3$.
When $t=2$ and $n\geq 2$, we have
$D_3(n,t)=\frac{n^2}{2}-\frac{n}{2}.$

Furthermore, since $D_2(n,t)=\sum\limits_{i=0}^{t}\binom{n-t}{i}$, we have 
\begin{equation*}
D_2(n,t)=\binom{n-t}{t}+\binom{n-t}{t-1}+\binom{n-t}{t-2}+O(n^{t-3})=\frac{n^t}{t!}-\frac{3}{2}\cdot\frac{n^{t-1}}{(t-2)!}+\frac{(27t^2-55t+50)}{24}\cdot\frac{n^{t-2}}{(t-2)!}+O(n^{t-3}),
\end{equation*}
for $t\geq2$.

Now suppose $q\geq t+1$, $k\geq 1$ and $n\geq t$. Since $t\le q-1\le q+k-1$, for any $0\le i\le t$ we have
$D_{q+k-1}(t,t-i)=D_{q-1}(t,t-i)=|D_{t-i}(\mathbf{x})|$,
where $\mathbf{x}=(0,1,\dots,t-1)$ is a sequence of distinct symbols. Because 
$D_q(n,t)=\sum\limits_{i=0}^{t}\binom{n-t}{i}D_{q-1}(t,t-i)$ and $D_{q+k}(n,t)=\sum\limits_{i=0}^{t}\binom{n-t}{i}D_{q+k-1}(t,t-i)$, it follows that $D_{q+k}(n,t)=D_q(n,t)$
for all $q\geq t+1$ and $k\geq 1$. Consequently, $D_q(n,t)-D_{q-1}(n,t)=0$ for $q\geq t+2$ and $n\geq t$.

Moreover, for $q\leq t+1$, we obtain
\begin{align*}
D_q(n,t)&-D_{q-1}(n,t)=\sum\limits_{i=0}^{t}\binom{n-t}{i}D_{q-1}(t,t-i)-\sum\limits_{i=0}^{t}\binom{n-t}{i}D_{q-2}(t,t-i)\\
&=\sum\limits_{i=t-q+3}^{t}\big(D_{q-1}(t,t-i)-D_{q-2}(t,t-i)\big)\cdot\binom{n-t}{i}+\big(D_{q-1}(t,q-2)-D_{q-2}(t,q-2)\big)\cdot\binom{n-t}{t-q+2}\\
&~~~~+\sum\limits_{i=0}^{t-q+1}\big(D_{q-1}(t,t-i)-D_{q-2}(t,t-i)\big)\cdot\binom{n-t}{i}\\
&\overset{(a)}{=}\big(D_{q-1}(t,q-2)-D_{q-2}(t,q-2)\big)\frac{n^{t-q+2}}{(t-q+2)!}+O(n^{t-q+1})\\
&\overset{(b)}{=}\frac{n^{t-q+2}}{(t-q+1)!}+O(n^{t-q+1}),
\end{align*}
where $(a)$ holds for sufficiently large $n$ since $D_{q-1}(t,t-i)-D_{q-2}(t,t-i)=0$ for $i\in [t-q+3,t]$, and $(b)$ follows from the fact that $D_{q-1}(t,q-2) - D_{q-2}(t,q-2) = t-q+2$, which can be seen from the following expressions:
\[
D_{q-1}(t,q-2) = \sum_{i=0}^{q-2} \binom{t-q+2}{i} D_{q-2}(q-2, q-2-i),\quad
D_{q-2}(t,q-2) = \sum_{i=0}^{q-2} \binom{t-q+2}{i} D_{q-3}(q-2, q-2-i),
\]
together with the boundary values $D_{q-2}(q-2,q-2)=D_{q-3}(q-2,q-2)=1$, $D_{q-2}(q-2,q-3)=q-2$, $D_{q-3}(q-2,q-3)=q-3$, and $D_{q-2}(q-2,q-2-i)=D_{q-3}(q-2,q-2-i)$ for $i\in[2,q-2]$.
\end{proof}

With the expansions and auxiliary bounds established above, we now characterize sequences whose deletion ball size attains the asymptotic value of $D_q(n,t)$. The following lemma summarizes the results.

\begin{lemma}
\label{Pro:lm11}
Let $q\geq 2$, $t\geq 1$, and $\mathbf{x}\in \Sigma_q^n$, and let $n$ be sufficiently large. 
\begin{enumerate}
\item If $|D_t(\mathbf{x})|$ has the same leading term as $D_q(n,t)$, i.e., $|D_t(\mathbf{x})| = \frac{n^t}{t!} - O(n^{t-1})$, then $r(\mathbf{x}) = n -O(1)$.
\item If $\mathbf{x}\in \Sigma_2^n$ and $|D_t(\mathbf{x})|$ matches the first two  terms of $D_2(n,t)$, i.e., $|D_t(\mathbf{x})| = \frac{n^t}{t!} - \frac{3}{2}\frac{n^{t-1}}{(t-2)!}+O(n^{t-2})$, then $\mathbf{x}\in \mathbf{C}_2(n)$.
\item If $\mathbf{x}\in \Sigma_3^n$ and $|D_t(\mathbf{x})|$ matches the first three  terms of $D_3(n,t)$, i.e., $|D_t(\mathbf{x})| = \frac{n^t}{t!} - \frac{1}{2}\frac{n^{t-1}}{(t-2)!} +\frac{(3t-25)}{24}\cdot\frac{n^{t-2}}{(t-3)!}+O(n^{t-3})$, then $\mathbf{x}\in \mathbf{C}_3(n)$.
\item If $\mathbf{x}\notin \mathbf{C}_q(n)$, then $D_q(n,t)-|D_t(\mathbf{x})|\geq n^{t-q+1}+O(n^{t-q})$. That is, if  $|D_t(\mathbf{x})|$ matches the first $q$  terms of $D_q(n,t)$, then $\mathbf{x}\in \mathbf{C}_q(n)$.
\end{enumerate}
\end{lemma}

\begin{proof}
We prove the four statements in order.

\emph{1)} By Lemma \ref{Pro:lm7}, it follows that $\binom{r(\mathbf{x})-t+1}{t}\leq \left|D_t(\mathbf{x})\right|\leq \binom{r(\mathbf{x})+t-1}{t},$ 
which yields the asymptotic expansion
\begin{align}
|D_t(\mathbf{x})|=\frac{r(\mathbf{x})^t}{t!}-O(r(\mathbf{x})^{t-1}). \label{Pro:eq5}
\end{align}
If $|D_t(\mathbf{x})| = \frac{n^t}{t!} - O(n^{t-1})$, then from Eq. \eqref{Pro:eq5} we obtain $r(\mathbf{x})=n-O(1)$.

\emph{2)} If $\mathbf{x}\in \Sigma_2^n$ and $r(\mathbf{x})\leq n-1$, then by Lemmas \ref{Pro:lm8} and \ref{Pro:lm10} we have
\begin{align*}
|D_t(\mathbf{x})|&\leq D_2(n-2,t)+D_2(n-2,t-1)+D_2(n-4,t-2)=\frac{(n-2)^t}{t!}-\frac{3}{2}\cdot\frac{(n-2)^{t-1}}{(t-2)!}+\frac{(n-2)^{t-1}}{(t-1)!}+O(n^{t-2})\\
&=\frac{n^t}{t!}-2t\cdot \frac{n^{t-1}}{t!}-\frac{3}{2}\cdot\frac{n^{t-1}}{(t-2)!}+\frac{n^{t-1}}{(t-1)!}+O(n^{t-2})=\frac{n^t}{t!}-\frac{3t-1}{2}\cdot\frac{n^{t-1}}{(t-1)!}+O(n^{t-2})\\
&\leq \frac{n^t}{t!}-\frac{3}{2}\cdot\frac{n^{t-1}}{(t-1)!}+O(n^{t-2}),
\end{align*}
for $t\geq 2$. Thus, if $\mathbf{x}\in \Sigma_2^n$ and $|D_t(\mathbf{x})|$ matches the first two  terms of $D_2(n,t)$, then $r(\mathbf{x})=n$, which implies that $\mathbf{x} \in \mathbf{C}_2(n)$.

\emph{3)} If $\mathbf{x}\in \Sigma_3^n$ and $\mathbf{x}\notin \mathbf{C}_3(n)$, then by Lemmas \ref{Pro:lm9} and \ref{Pro:lm10} we obtain that
\begin{align*}
|D_t(\mathbf{x})|&\leq D_3(n-2,t)+D_3(n-2,t-1)+D_3(n-3,t-1)+D_3(n-3,t-2)+D_3(n-5,t-3)\\
&=\frac{(n-2)^t}{t!}-\frac{1}{2}\cdot\frac{(n-2)^{t-1}}{(t-2)!}+\frac{3t-25}{24}\cdot\frac{(n-2)^{t-2}}{(t-3)!}+\frac{(n-2)^{t-1}}{(t-1)!}-\frac{1}{2}\cdot\frac{(n-2)^{t-2}}{(t-3)!}+\frac{(n-3)^{t-1}}{(t-1)!}-\frac{1}{2}\cdot\frac{(n-3)^{t-2}}{(t-3)!}\\
&~~+\frac{(n-3)^{t-2}}{(t-2)!}+O(n^{t-3})\\
&=\frac{n^t}{t!}-\frac{t-1}{2}\cdot\frac{n^{t-1}}{(t-1)!}+\frac{3t^2-31t+26}{24}\cdot\frac{n^{t-2}}{(t-2)!}+O(n^{t-3})\\
&\leq \frac{n^t}{t!}-\frac{t-1}{2}\cdot\frac{n^{t-1}}{(t-1)!}+\frac{3t^2-31t+50}{24}\cdot\frac{n^{t-2}}{(t-2)!}+O(n^{t-3})\\
&= \frac{n^t}{t!}-\frac{1}{2}\cdot\frac{n^{t-1}}{(t-2)!}+\frac{(3t-25)}{24}\cdot\frac{n^{t-2}}{(t-3)!}+O(n^{t-3}),
\end{align*}
for $t\geq 2$. Therefore, if $|D_t(\mathbf{x})|$ matches the first three  terms of $D_3(n,t)$, then $\mathbf{x} \in \mathbf{C}_3(n)$.

\emph{4)}  Consider $\mathbf{x}\notin \mathbf{C}_q(n)$. Then there exists some index $i\in [1,n-1]$ at which the sequence breaks periodicity, i.e., $\mathbf{x}_{[i,n]}\notin \mathbf{C}_{q}(n-i+1)$. By Corollary \ref{Pro:cor1}, we have
\begin{align*}
|D_t(\mathbf{x})|\leq |D_t(\mathbf{x}_{[1,i]}^{c},\mathbf{x}_{[i+1,n]}^{cR})|,
\end{align*}
where \(\mathbf{x}_{[1,i]}^{c}=\mathbf{c}_q(i,\sigma)\) and \(\mathbf{x}_{[i+1,n]}^{cR}=\mathbf{c}_q(n-i,\hat{\sigma})\), with \(\sigma=(\sigma_0,\dots,\sigma_{q-1})\) and \(\hat{\sigma}=(\hat{\sigma}_0,\dots,\hat{\sigma}_{q-1})\) being orderings of \(\Sigma_q\) such that the last element of \(\mathbf{c}_q(i,\sigma)\) is \(x_i\) and the first element of \(\mathbf{c}_q(n-i,\hat{\sigma})\) is \(x_{i+1}\). For convenience, define \(\mathbf{y}_{(i)}=(\mathbf{c}_q(i,\sigma),\mathbf{c}_q(n-i,\hat{\sigma}))\) for \(i\in[1,n-1]\).

We prove by induction on \(i\) that for all \(i \in [1,n-1]\),
\begin{align}
D_q(n,t)-|D_t(\mathbf{y}_{(i)})|\geq n^{t-q+1}+O(n^{t-q}).\label{Pro:eq6}
\end{align}

\emph{Base case.} 
For \(i = 1\), let \(\mathbf{y}_{(1)} = (\mathbf{c}_q(1,\sigma),\mathbf{c}_q(n-1,\hat{\sigma})) = (a,\mathbf{c}_q(n-1,\hat{\sigma}))\) with $\sigma_0=\hat{\sigma}_a$ for some $a\in[q-2]\cup\{0\}$. Then, by Lemmas~\ref{Pro:lm1} and~\ref{Pro:lm2},
\[
\begin{aligned}
|D_t(\mathbf{y}_{(1)})| &= |D_t(\mathbf{y}_{(1)})^{\hat{\sigma}_a}|+\sum_{j\in \Sigma_q\setminus\{\hat{\sigma}_a\}} |D_t(\mathbf{y}_{(1)})^{j}|\\
&= D_q(n-1,t) +D_q(n-2,t-1)+\cdots + D_q(n-q-1,t-q)-D_q(n-2-a,t-1-a) \\
&\overset{(a)}{\leq} D_q(n,t)-n^{t-1-a}+O(n^{t-2-a}) \overset{(b)}{\leq} D_q(n,t)-n^{t-q+1}+O(n^{t-q}),
\end{aligned}
\]
where $(a)$ follows from Eqs. \eqref{intr:eq2} and \eqref{def:eq1} with $0\leq a\leq q-2$,  $(b)$ follows from $0\leq a\leq q-2$. Thus, the base case \(i = 1\) satisfies Eq. \eqref{Pro:eq6}.

\emph{Inductive step.} Assume that Eq. \eqref{Pro:eq6} holds for every \(i\) with \(1\le i\le s\), where \(s\le n-2\). We prove it for \(i = s+1\). For convenience, write
\(
\mathbf{y}_{(s+1)} = (\mathbf{c}_q(s+1,\sigma),\mathbf{c}_q(n-s-1,\hat{\sigma})) = (a_1,a_2,a_3,\dots,a_{s+1},\mathbf{c}_q(n-s-1,\hat{\sigma})).
\)

If $s+1\leq q-1$, then applying Lemmas~\ref{Pro:lm1} and~\ref{Pro:lm2},
\begin{align*}
|D_t(\mathbf{y}_{(s+1)})| &= \sum_{j\in[s+1]} |D_t(\mathbf{y}_{(s+1)})^{a_j}|+\sum_{i\in \Sigma_q\setminus\{a_k|k\in[s+1]\}} |D_t(\mathbf{y}_{(s+1)})^{i}| \\
&\overset{(a)}{\leq} \sum_{j=0}^{s-1} \big(D_q(n-1-j,t-j)-(n-1-j)^{t-j-q+1}+O((n-1-j)^{t-j-q})\big)+D_q(n-s-1,t-s)\\
&+\sum_{i\in \Sigma_q\setminus\{a_k|k\in[s+1]\}} |D_t(\mathbf{y}_{(s+1)})^{i}| \\
&\overset{(b)}{\leq} \sum_{j=0}^{s-1} \big(D_q(n-1-j,t-j)-(n-1-j)^{t-j-q+1}+O((n-1-j)^{t-j-q})\big)+D_q(n-s-1,t-s)\\
&+\sum_{i\in [s+2,q]} D_q(n-i,t-i+1) \\
& \overset{(c)}{\leq} D_q(n,t)-n^{t-q+1}+O(n^{t-q}),
\end{align*}
where $(a)$ follows from the induction hypothesis, $(b)$ follows from $i\notin \{a_k|k\in[s+1]\}$, $(c)$ follows from Eq. \eqref{def:eq1}. Similarly, if $s+1\geq q$, then we also have the above result.  

By induction, for every \(i \in [n-1]\), we have
\(
|D_t(\mathbf{y}_{(i)})|\leq D_q(n,t)-n^{t-q+1}+O(n^{t-q}).
\)

Now, if \(\mathbf{x} \notin \mathbf{C}_q(n)\), then
\(
|D_t(\mathbf{x})| \le \max_{1 \le i \le n-1} |D_t(\mathbf{y}_{(i)})|\leq D_q(n,t)-n^{t-q+1}+O(n^{t-q}).
\)
Therefore, if \(|D_t(\mathbf{x})|\) has the first $q$  terms of \(D_q(n,t)\), then we have $\mathbf{x}\in \mathbf{C}_q(n)$.
\end{proof}

\section{Lower Bounds on $N_q(n,2,t)$}
\label{sec4}

In this section, we explicitly construct pairs of $q$-ary sequences whose deletion balls have large intersection, thereby establishing a lower bound on $N_q(n,2,t)$. 
Specifically, for $q \ge 3$ and $n\geq q+3$, we consider the sequences
\begin{equation}\label{Low:eq1}
\mathbf{x}^{(0)} = (0,1,q-1,0,1,\mathbf{w}^{(0)}), \qquad
\mathbf{y}^{(0)} = (1,0,q-1,1,0,\mathbf{w}^{(0)}),
\end{equation}
where $\mathbf{w}^{(0)} = (w_1,\dots,w_{n-5})$ has entries satisfying $w_j \equiv j+1 \pmod q$ for $j \in [n-5]$. For $i \in [1, q-3]$,  we also consider the sequences
\begin{equation}\label{Low:eq2}
\mathbf{x}^{(i)} = (q-2-i+1,\dots,q-2,0,1,q-1,0,1,\mathbf{w}^{(i)}), \qquad
\mathbf{y}^{(i)} = (q-2-i+1,\dots,q-2,1,0,q-1,1,0,\mathbf{w}^{(i)}),
\end{equation}
where $\mathbf{w}^{(i)} = (w_1,\dots,w_{n-5-i})$ has entries satisfying $w_j \equiv j+1 \pmod q$ for $j \in [n-5-i]$. Here, the index $i$ (with $0 \le i \le q-3$) denotes the length of the common prefix of $\mathbf{x}^{(i)}$ and $\mathbf{y}^{(i)}$. When \(q=3\), the construction reduces to \(\mathbf{x}^{(0)}\) and \(\mathbf{y}^{(0)}\), which is precisely the construction in \cite{Wang}, with the convention that the set \([1,q-3]\) is empty. We will show that $d_L(\mathbf{x}^{(i)},\mathbf{y}^{(i)}) = 2$ and compute $|D_t(\mathbf{x}^{(i)}) \cap D_t(\mathbf{y}^{(i)})|$ for $0\leq i\leq q-3$, which yields a lower bound on \(N_q(n,2,t)\).

To compute the size of this intersection and verify the distance condition, we rely on combinatorial decompositions in Lemmas \ref{Pro:lm1} and \ref{Pro:lm2}. We first confirm that the constructed sequences have Levenshtein distance exactly $2$.

\begin{lemma}
Let $q \ge 3$ and $n\geq q+3$. For every $0\leq i\leq q-3$, the sequences \(\mathbf{x}^{(i)},\mathbf{y}^{(i)}\) defined in Eqs.~(\ref{Low:eq1}) and~(\ref{Low:eq2}) satisfy \(d_L(\mathbf{x}^{(i)},\mathbf{y}^{(i)})=2\).
\label{Low:lm1}
\end{lemma}
\begin{proof}
For \(i=0\), the prefixes \((0,1,q-1,0,1)\) and \((1,0,q-1,1,0)\) occur at the beginning of \(\mathbf{x}^{(0)}\) and \(\mathbf{y}^{(0)}\), respectively. A direct computation shows that their Levenshtein distance is \(2\). Since the suffixes \(\mathbf{w}^{(0)}\) are identical, the distance between \(\mathbf{x}^{(0)}\) and \(\mathbf{y}^{(0)}\) is exactly \(2\).

For \(i>0\), the common prefix \((q-i-1,\dots,q-2)\) is identical in both sequences and therefore does not affect the Levenshtein distance. The parts immediately following this prefix are precisely the blocks \((0,1,q-1,0,1)\) and \((1,0,q-1,1,0)\), respectively, followed by identical suffixes \(\mathbf{w}^{(i)}\). Thus the distance is again \(2\). This completes the lemma.
\end{proof}

Having established that $d_L(\mathbf{x}^{(i)},\mathbf{y}^{(i)}) = 2$ for all $0\leq i \leq q-3$, we now compute the size of the intersection of their deletion balls of radius $t$. First, we need some auxiliary lemmas.

\begin{lemma}
Let $q\geq 3$. Let $\mathbf{w} = (w_1,\dots,w_{n-5})$ with $w_i \equiv i+1 \pmod q$.
For $n\geq q+t+3$,
\begin{align*}
|D_t(1,q-1,0,1,\mathbf{w})\cap D_{t-1}(q-1,1,0,\mathbf{w})|
&= D_q(n-2,t-1)+D_q(n-5,t-2)+D_q(n-5-q,t-2-q)\\
&\quad -D_q(n-4,t-1)-D_q(n-4-q,t-1-q),\\
|D_t(0,q-1,1,0,\mathbf{w})\cap D_{t-1}(q-1,0,1,\mathbf{w})|=&D_q(n-2,t-1)+D_q(n-5,t-2)-D_q(n-4,t-1).
\end{align*}
\label{Low:lm2}
\end{lemma}
\begin{proof}
Set $\mathcal{X}=D_t(\mathbf{x})\cap D_{t-1}(\mathbf{y})$ with $\mathbf{x}=(1,q-1,0,1,\mathbf{w})$ and $\mathbf{y}=(q-1,1,0,\mathbf{w})$. For $n\geq t+3$, by Lemmas \ref{Pro:lm1} and \ref{Pro:lm2}, we have $|\mathcal{X}| = \sum\limits_{i=0}^{q-1}|\mathcal{X}^{i}|$ with $|\mathcal{X}^{0}|= |D_{t-3}(\mathbf{w})| = D_q(n-5,t-3),$ and
\begin{align*}
|\mathcal{X}^{1}| &= |D_{t-2}(0,\mathbf{w})|
= D_q(n-4,t-2)-D_q(n-4-q,t-1-q)+D_q(n-5-q,t-2-q), \\
|\mathcal{X}^{q-1}| &= |D_{t-1}(0,1,\mathbf{w})\cap D_{t-1}(1,0,\mathbf{w})| \overset{(a)}{=} D_q(n-3,t-1)-D_q(n-4,t-1)+D_q(n-5,t-2), \\
|\mathcal{X}^{i}| &= |D_{t-2-i}(\mathbf{w}_{[i,n-5]})| = D_q(n-4-i,t-2-i) \quad \text{for } 2\le i\le q-2, 
\end{align*}
where $(a)$ follows from Theorem \ref{def:thm1}. Thus, we have 
\begin{align*}
|\mathcal{X}| &= \Bigl(\sum_{i=1}^{q} D_q(n-2-i,t-i)\Bigr) + D_q(n-5,t-2) + D_q(n-5-q,t-2-q)- D_q(n-4,t-1) - D_q(n-4-q,t-1-q).
\end{align*}
By Eq. \eqref{def:eq1}, for $n\ge q+t+3$ and $q\geq 3$, we have
\[
|\mathcal{X}| = D_q(n-2,t-1) + D_q(n-5,t-2) + D_q(n-5-q,t-2-q)
- D_q(n-4,t-1) - D_q(n-4-q,t-1-q).
\]
A similar argument yields $|D_t(0,q-1,1,0,\mathbf{w})\cap D_{t-1}(q-1,0,1,\mathbf{w})|=D_q(n-2,t-1)+D_q(n-5,t-2)-D_q(n-4,t-1)$, which completes the proof.
\end{proof}

By Lemmas~\ref{Pro:lm1},~\ref{Pro:lm2}, and~\ref{Low:lm2}, we obtain the following results, which give explicit expressions for $|D_t(\mathbf{x}^{(i)}) \cap D_t(\mathbf{y}^{(i)})|$ for $t \ge 2$, $q \ge 3$, and $0\leq i \leq q-3$. For convenience, define $N_q^{(i)}(n,t) = |D_t(\mathbf{x}^{(i)}) \cap D_t(\mathbf{y}^{(i)})|$ for $0\leq i \leq q-3$ and $n\geq q+t+3$.

\begin{lemma}
\label{Low:lm3}
Let $t \ge 2$, $q \ge 3$, and $n\geq q+t+3$. Then,
\begin{align*}
N_q&^{(0)}(n,t)\\
=&D_q(n-4,t-2)+5D_q(n-5,t-2)+4D_q(n-5,t-3)+3D_q(n-6,t-3)+D_q(n-6,t-4)+D_q(n-5-q,t-2-q)\\
&+\sum\limits_{i=3}^{q-1}\big(2D_q(n-3-i,t-1-i)+3D_q(n-4-i,t-1-i)+D_q(n-3-i,t-2-i)+D_q(n-4-i,t-2-i)\big).
\end{align*}
Moreover, we have 
\begin{align}
N_q&^{(0)}(n,t)\geq D_q(n-2,t-2).\label{Low:eq3}
\end{align}
For all $t \ge 2$ and sufficiently large $n$,
\begin{align}
N_3^{(0)}(n,t)&=\frac{6}{(t-2)!}n^{t-2}-\frac{3t+13}{(t-3)!}n^{t-3}+\frac{3t^2+t+140}{4(t-4)!}n^{t-4}+O(n^{t-5}),\nonumber\\
N_q^{(0)}(n,t)&=\frac{6}{(t-2)!}n^{t-2}-\frac{3t+13}{(t-3)!}n^{t-3}+\frac{3t^2+25t+64}{4(t-4)!}n^{t-4}+O(n^{t-5}),\qquad \text{for $q\geq 4$}.\nonumber
 %\label{Low:eq4}
\end{align}
%where we adopt the convention that \(1/k! = 0\) for any integer \(k<0\) in the asymptotic expansion.
\end{lemma}

\begin{proof}
From Eq. \eqref{Low:eq1}, we have
\(\mathbf{x}^{(0)} = (0,1,q-1,0,1,\mathbf{w}^{(0)})\) and 
\(\mathbf{y}^{(0)} = (1,0,q-1,1,0,\mathbf{w}^{(0)}),\)
where $\mathbf{w}^{(0)} = (w_1,\dots,w_{n-5})$ satisfies $w_j \equiv j+1 \pmod q$ for $j\in [n-5]$. For brevity, write \(\mathbf{w}:=\mathbf{w}^{(0)}\). Define \(\mathcal{X}=D_t(\mathbf{x}^{(0)})\cap D_t(\mathbf{y}^{(0)})\), so \(N_q^{(0)}(n,t)=|\mathcal{X}|\).

By Lemmas \ref{Pro:lm1} and \ref{Pro:lm2}, we have
\(
|\mathcal{X}| = \sum_{i=0}^{q-1} |\mathcal{X}^i|
\)
with $|\mathcal{X}^0|= |D_t(1,q-1,0,1,\mathbf{w})\cap D_{t-1}(q-1,1,0,\mathbf{w})|$, and 
\begin{align*}
|\mathcal{X}^1| &= |D_{t-1}(q-1,0,1,\mathbf{w})\cap D_t(0,q-1,1,0,\mathbf{w})|,\\
|\mathcal{X}^{q-1}| &= |D_{t-2}(0,1,\mathbf{w})\cap D_{t-2}(1,0,\mathbf{w})|
\overset{(a)}{=} D_q(n-3,t-2)-D_q(n-4,t-2)+D_q(n-5,t-3),\\
%|\mathcal{X}^2| &= |D_{t-5}(\mathbf{w}_{[2,n-5]})| = D_q(n-6,t-5),\\
%&\;\;\vdots\\
|\mathcal{X}^i| &= |D_{t-3-i}(\mathbf{w}_{[i,n-5]})| = D_q(n-4-i,t-3-i) \quad \text{for } 2\le i\le q-2,
%&\;\;\vdots\\
%|\mathcal{X}^{q-2}| &= |D_{t-q}(\mathbf{w}_{[q-2,n-5]})| = D_q(n-2-q,t-q-1).
\end{align*}
where $(a)$ follows from Theorem \ref{def:thm1}. Applying Lemma \ref{Low:lm2} to $|\mathcal{X}^0|$ and $|\mathcal{X}^1|$ yields
\begin{align*}
|\mathcal{X}^0| &= D_q(n-2,t-1)+D_q(n-5,t-2)+D_q(n-5-q,t-2-q) -D_q(n-4,t-1)-D_q(n-4-q,t-1-q),\\[2mm]
|\mathcal{X}^1| &= D_q(n-2,t-1)+D_q(n-5,t-2)-D_q(n-4,t-1).
\end{align*}

Thus, we obtain
\begin{align}
|\mathcal{X}| &=2D_q(n-2,t-1)+D_q(n-3,t-2)+2D_q(n-5,t-2)+D_q(n-5,t-3)+D_q(n-5-q,t-2-q)\nonumber\\
&-2D_q(n-4,t-1)-D_q(n-4,t-2)-D_q(n-4-q,t-1-q)+\sum\limits_{i=3}^{q-1}D_q(n-3-i,t-2-i)\nonumber\\
&=2D_q(n-2,t-1)+D_q(n-2,t-2)+2D_q(n-5,t-2)+D_q(n-5,t-3)+D_q(n-5-q,t-2-q)\nonumber\\
&-2D_q(n-4,t-1)-D_q(n-4,t-2)-D_q(n-4,t-3)-D_q(n-5,t-4)-D_q(n-4-q,t-1-q)\label{Low:eq5}\\
&=D_q(n-4,t-2)+5D_q(n-5,t-2)+4D_q(n-5,t-3)+3D_q(n-6,t-3)+D_q(n-6,t-4)\nonumber\\
&~~+D_q(n-5-q,t-2-q)+\sum\limits_{i=3}^{q-1}\big(2D_q(n-3-i,t-1-i)+3D_q(n-4-i,t-1-i)\nonumber\\
&~~~~~~~~~~~~~~~~~~~~~~~~~~~~~~~~~~~~~~~~~~~~~~~~~+D_q(n-3-i,t-2-i)+D_q(n-4-i,t-2-i)\big).\nonumber
\end{align}

By Eq. \eqref{def:eq1}, for $n\geq t+3$, we have 
\begin{align*}
D_q&(n-2,t-2)=D_q(n-3,t-2)+D_q(n-4,t-3)+D_q(n-5,t-4)+\sum\limits_{i=3}^{q-1}D_q(n-3-i,t-2-i)\\
&=D_q(n-4,t-2)+2D_q(n-5,t-3)+D_q(n-6,t-4)+2D_q(n-3-q,t-q-1)+2D_q(n-4-q,t-q-2)\\
&~~~+D_q(n-5-q,t-q-3)+\sum\limits_{i=3}^{q-1}\big(D_q(n-3-i,t-2-i)+2D_q(n-3-i,t-1-i)+D_q(n-4-i,t-2-i)\big)\\
&\overset{(a)}{\leq} D_q(n-4,t-2)+2D_q(n-5,t-3)+D_q(n-6,t-4)+2D_q(n-3-q,t-q-1)+2qD_q(n-5-q,t-q-2)\\
&~~~+qD_q(n-6-q,t-q-3)+\sum\limits_{i=3}^{q-1}\big(D_q(n-3-i,t-2-i)+2D_q(n-3-i,t-1-i)+D_q(n-4-i,t-2-i)\big)\\
&\overset{(b)}{\leq} D_q(n-4,t-2)+4D_q(n-5,t-3)+D_q(n-6,t-4)+5D_q(n-5,t-2)+3D_q(n-6,t-3)\\
&~~~+D_q(n-5-q,t-2-q)+\sum\limits_{i=3}^{q-1}\big(2D_q(n-3-i,t-1-i)+3D_q(n-4-i,t-1-i)\\
&~~~~~~~~~~~~~~~~~~~~~~~~~~~~~~~~~~~~~~~~~~~~~~~~~+D_q(n-3-i,t-2-i)+D_q(n-4-i,t-2-i)\big)=N_q^{(0)}(n,t).
\end{align*}
Here, $(a)$ follows from Corollary \ref{Pro:cor2}. For $(b)$, we use Eq. \eqref{Pro:eq3} to obtain that $2D_q(n-3-q,t-q-1)\leq 2D_q(n-5,t-3)$ and $2qD_q(n-5-q,t-q-2)+qD_q(n-6-q,t-q-3)\leq 5D_q(n-5,t-2)+3D_q(n-6,t-3)+D_q(n-5-q,t-2-q)+\sum\limits_{i=3}^{q-1}3D_q(n-4-i,t-1-i)$.

Since $D_3(n,t)=\frac{n^t}{t!}-\frac{1}{2}\cdot\frac{n^{t-1}}{(t-2)!}+\frac{(3t-25)}{24}\cdot\frac{n^{t-2}}{(t-3)!}+O(n^{t-3})$,  we have 
$$N_3^{(0)}(n,t)=\frac{6}{(t-2)!}n^{t-2}-\frac{3t+13}{(t-3)!}n^{t-3}+\frac{3t^2+t+140}{4(t-4)!}n^{t-4}+O(n^{t-5}).$$  
Since $D_q(n,t)=\frac{n^t}{t!}-\frac{1}{2}\cdot\frac{n^{t-1}}{(t-2)!}+\frac{(3t-1)}{24}\cdot\frac{n^{t-2}}{(t-3)!}+O(n^{t-3})$ for $q\geq 4$,  we conclude that
$$N_q^{(0)}(n,t)=\frac{6}{(t-2)!}n^{t-2}-\frac{3t+13}{(t-3)!}n^{t-3}+\frac{3t^2+25t+64}{4(t-4)!}n^{t-4}+O(n^{t-5}).$$  
\end{proof}

We now analyze the difference between $N_q^{(0)}(n,t)$ and $N_{q-1}^{(0)}(n,t)$.
Recall that both $N_q^{(0)}(n,t)$ and $N_{q-1}^{(0)}(n,t)$ can be written in the unified form $\sum_{a,b} c_{a,b} D_k(n-a,t-b)$ for $k = q$ and $k = q-1$, respectively. 
A direct comparison shows that $N_q^{(0)}(n,t)$ differs from $N_{q-1}^{(0)}(n,t)$ by exactly four additional terms:
\[
2D_q(n-2-q,t-q) + 3D_q(n-3-q,t-q) + D_q(n-2-q,t-1-q) + D_q(n-3-q,t-1-q).
\]

Let $q\ge3$, $t\ge2$, and $n\geq q+t+3$. We consider the following three cases.
\begin{enumerate}
\item If $q\ge t+1$, then all four additional terms vanish (since their $t$-indices are negative).
By Lemma \ref{Pro:lm10}, $D_q(n-a,t-b)=D_{q-1}(n-a,t-b)$ for all $b\ge2$.
Thus, $N_q^{(0)}(n,t)-N_{q-1}^{(0)}(n,t)=0$.

\item If $q=t$, then the four additional terms reduce to constants with total sum $5$.
By Lemma \ref{Pro:lm10}, $D_q(n-a,t-b)=D_{q-1}(n-a,t-b)$ for all $b\ge2$.
Therefore, $N_q^{(0)}(n,t)-N_{q-1}^{(0)}(n,t)=5$.

\item If $q\le t-1$, then the leading term of the four additional terms is $\frac{5n^{t-q}}{(t-q)!}$ for sufficiently large $n$.
By Lemma \ref{Pro:lm10}, the difference between the common terms satisfies
\[
D_q(n-a,t-b)-D_{q-1}(n-a,t-b)=\frac{n^{t-b-q+2}}{(t-b-q+1)!} + O(n^{t-b-q+1})
\]
for all $b$ satisfying $q \le t - b + 1$. Combining these contributions yields
\[
N_q^{(0)}(n,t)-N_{q-1}^{(0)}(n,t)=\frac{(6t-6q+5)n^{t-q}}{(t-q)!} + O(n^{t-q-1}),
\]
for sufficiently large $n$.
\end{enumerate}

\begin{remark}\label{Low:rmk1}
Let $q\ge4$, $t\ge2$, and $n\geq q+t+3$. Then
\[
N_q^{(0)}(n,t)-N_{q-1}^{(0)}(n,t) =
\begin{cases}
0, & \text{if } q\ge t+1, \\
5, & \text{if } q=t, \\
\dfrac{(6t-6q+5)n^{t-q}}{(t-q)!} + O(n^{t-q-1}), & \text{if } q\le t-1 \text{ and } n \text{ is sufficiently large}.
\end{cases}
\]
\end{remark}

\begin{lemma}
\label{Low:lm4}
Let $t\ge2$ and $q\ge4$.
\begin{itemize}
  \item For each $i \in [q-3]$ and $n\geq t+2q$,
\begin{align*}
  N_q^{(i)}(n,t)=\sum\limits_{j=0}^{i-1}N_q^{(j)}(n-i+j,t-i+j+1)+N_q^{(0)}(n-i,t-i)-\sum\limits_{j=1}^{i}D_q(n-3-i-q+j,t-2-i-q+j).
\end{align*}
  Moreover, for \(n\ge t+2q\), the explicit closed-form expression for \(N_q^{(q-3)}(n,t)\) is given by
\begin{align*}
N_q^{(q-3)}(n,t) &=
\sum_{j=0}^{q-3} \binom{q-3}{j} \Bigl[ 2D_q(n-q+1,\,t-j-1) + D_q(n-q+1,\,t-j-2)+ 2D_q(n-q-2,\,t-j-2) \\
&\quad  + D_q(n-q-2,\,t-j-3)+ D_q(n-2q-2,\,t-j-2-q) - 2D_q(n-q-1,\,t-j-1) \\
&\quad - D_q(n-q-1,\,t-j-2) - D_q(n-q-1,\,t-j-3) - D_q(n-q-2,\,t-j-4) \\
&\quad - D_q(n-2q-1,\,t-j-1-q) \Bigr] - \sum_{m=0}^{q-4} \sum_{j=0}^{m} \binom{m}{j} D_q(n-3-q-m,\,t-2-q-j).
\end{align*}
  \item  For \(i\in[q-3]\) and \(n\ge t+3q\), we have \(N_q^{(i)}(n,t)\ge N_q^{(i-1)}(n,t)\) when $t\geq q+2$; moreover, equality holds when \(t\le q+1\).
  \item For \(i\in[q-3]\), \(t\ge q+2\), and sufficiently large \(n\), the difference between $N_q^{(i)}(n,t)$ and $N_q^{(i-1)}(n,t)$ satisfies
\[
N_q^{(i)}(n,t)-N_q^{(i-1)}(n,t)=\dfrac{5}{(t-q-2)!}n^{t-q-2}+O(n^{t-q-3}).
\]
Moreover, the  asymptotic expansion of $N_q^{(i)}(n,t)$ is 
\begin{equation}
N_q^{(i)}(n,t)=\frac{6}{(t-2)!}n^{t-2}-\frac{3t+13}{(t-3)!}n^{t-3}+\frac{3t^2+25t+64}{4(t-4)!}n^{t-4}+O(n^{t-5}),
 \nonumber%\label{Low:eq6}
\end{equation}
for $i\in [q-3]$
\end{itemize}
\end{lemma}

\begin{proof}
We proceed in several steps.

\emph{1) Derivation of the recurrence relation for $N_q^{(i)}(n,t)$.}

For convenience, let $\mathbf{u}=(1,2,\cdots,q-1,q)$, where we adopt the convention that \(\mathbf{u}_{[a,b]}\) denotes the empty sequence for all positive integers \(a,b\) satisfying \(q\ge a > b \ge 1\). 

From Eqs. \eqref{Low:eq1} and  \eqref{Low:eq2}, for \(0\le i\le q-3\),
\(
\mathbf{x}^{(i)} = \mathbf{u}_{[q-1-i,q-2]}\circ(0,1,q-1,0,1,\mathbf{w}^{(i)}), 
\mathbf{y}^{(i)} = \mathbf{u}_{[q-1-i,q-2]}\circ(1,0,q-1,1,0,\mathbf{w}^{(i)}),
\)
where $\mathbf{w}^{(i)} = (w_1,\dots,w_{n-i-5})$ satisfies $w_j \equiv j+1 \pmod q$ for all $j\in [n-i-5]$.
Then, for $i\in [q-3]$, we have $\mathbf{x}^{(i)}=\mathbf{u}_{[q-1-i,q-2]}\circ \mathbf{x}^{(0)}_{[1,n-i]}$ and  $\mathbf{y}^{(i)}=\mathbf{u}_{[q-1-i,q-2]}\circ \mathbf{y}^{(0)}_{[1,n-i]}$. For convenience, set $\hat{\mathcal{X}}_{(i)}=D_{t-i}(\mathbf{x}^{(0)}_{[1,n-i]})\cap D_{t-i}(\mathbf{y}^{(0)}_{[1,n-i]})$ for $i\in [q-3]$. 

For each \(0\le i\le q-3\) and $n\geq t+q+3$, define $\mathcal{X}_{(i)} = D_t(\mathbf{x}^{(i)}) \cap D_t(\mathbf{y}^{(i)})$, so that $N_q^{(i)}(n,t) = |\mathcal{X}_{(i)}|$. For $i\in [q-3]$, by Lemmas \ref{Pro:lm1} and \ref{Pro:lm2}, we have $|\mathcal{X}_{(i)}|=\sum\limits_{k=0}^{q-1}|\mathcal{X}_{(i)}^{k}|$ with
\begin{align}
|\mathcal{X}_{(i)}^k| &= |D_{t-k+q-1-i}(\mathbf{u}_{[k+1,q-2]}\circ \mathbf{x}^{(0)}_{[1,n-i]})\cap D_{t-k+q-1-i}(\mathbf{u}_{[k+1,q-2]}\circ \mathbf{y}^{(0)}_{[1,n-i]})|, \qquad\text{for $q-1-i \leq k\leq q-2$,}\nonumber\\
|\mathcal{X}_{(i)}^k| &= |\big(D_{t-i}(\mathbf{x}^{(0)}_{[1,n-i]})\cap D_{t-i}(\mathbf{y}^{(0)}_{[1,n-i]})\big)^{k}|, \qquad\text{for $k\in \Sigma_q\setminus [q-1-i,q-2]$,}\label{Low:eq7}
\end{align}
where $|\mathcal{X}_{(i)}^k|=N_q^{(q-2-k)}(n+q-2-i-k,t+q-1-i-k)$ for $k\in [q-1-i,q-2]$ and $n\geq t+2q$, and $|\mathcal{X}_{(i)}^k|=|\hat{\mathcal{X}}_{(i)}^{k}|$ for $k\in \Sigma_q\setminus [q-1-i,q-2]$.

For $n\geq t+2q$ and $i\in [q-3]$, we have $n-i\geq t-i+q+3$. By Lemma \ref{Low:lm3}, we have  $|\hat{\mathcal{X}}_{(i)}|=N_q^{(0)}(n-i,t-i)$. Moreover, for $q-1-i\le j\le q-2$, it follows that 
\begin{align}
|\hat{\mathcal{X}}_{(i)}^j|= |D_{t-3-i-j}(\mathbf{w}_{[j,n-5-i]}^{(i)})| = D_q(n-4-i-j,t-3-i-j).\label{Low:eq8}
\end{align}

Therefore, by Eqs. (\ref{Low:eq7}) and (\ref{Low:eq8}), for $i\in [q-3]$, we have 
\begin{align}
N_q^{(i)}(n,t)&=|\mathcal{X}_{(i)}|=\sum\limits_{k=q-1-i}^{q-2}|\mathcal{X}_{(i)}^{k}|+\sum\limits_{k\in \Sigma_q\setminus[q-1-i,q-2]}|\mathcal{X}_{(i)}^{k}|=\sum\limits_{k=q-1-i}^{q-2}|\mathcal{X}_{(i)}^{k}|+\sum\limits_{k\in \Sigma_q\setminus[q-1-i,q-2]}|\hat{\mathcal{X}}_{(i)}^{k}|\nonumber\\
&=\sum\limits_{k=q-1-i}^{q-2}|\mathcal{X}_{(i)}^{k}|+|\hat{\mathcal{X}}_{(i)}|-\sum\limits_{k\in [q-1-i,q-2]}|\hat{\mathcal{X}}_{(i)}^{k}|\nonumber\\
&=\sum\limits_{j=0}^{i-1}N_q^{(j)}(n-i+j,t-i+j+1)+N_q^{(0)}(n-i,t-i)-\sum\limits_{j=1}^{i}D_q(n-3-i-q+j,t-2-i-q+j). \label{Low:eq9}
\end{align}

\emph{2) Explicit closed-form expression of $N_q^{(i)}(n,t)$.}

From the recurrence, we derive a simpler relation for consecutive terms
\[
N_q^{(i+1)}(n,t)=N_q^{(i)}(n-1,t)+N_q^{(i)}(n-1,t-1)-D_q(n-3-q,t-2-q).
\]
To solve this, we introduce shift operators $S_1,S_2$ defined by $S_1D_q(n,t)=D_q(n-1,t)$ and $S_2D_q(n,t)=D_q(n-1,t-1)$.
By Lemma~\ref{Low:lm3}, \(N_q^{(0)}(n,t)\) is a linear combination of terms \(D_q(n-a_k,t-b_k)\) for some integers \(a_k,b_k\). Combining this with Eq.~\eqref{Low:eq9}, it follows that for every \(i\in[q-3]\), \(N_q^{(i)}(n,t)\) is also a linear combination of terms \(D_q(n-a_k,t-b_k)\) for some integers \(a_k,b_k\). Hence, by the definitions of $S_1$ and $S_2$, we have \(S_1N_q^{(i)}(n,t)=N_q^{(i)}(n-1,t)\) and  \( S_2N_q^{(i)}(n,t)=N_q^{(i)}(n-1,t-1)\).
The recurrence becomes
\(
N_q^{(i+1)}(n,t)=(S_1+S_2)N_q^{(i)}(n,t)-F(n,t),
\)
where $F(n,t)=D_q(n-3-q,t-2-q)$ and $0\leq i\leq q-4$. Iterating this recurrence \(i\) times gives
\[
N_q^{(i)}(n,t)=(S_1+S_2)^i N_q^{(0)}(n,t)-\sum_{k=0}^{i-1}(S_1+S_2)^{k}F(n,t).
\]

Note that \(S_1\) and \(S_2\) commute (they act on different variables, shifting \(n\) and \(t\) respectively), so we have the binomial expansion
\(
(S_1+S_2)^m f(n,t)=\sum_{j=0}^{m}\binom{m}{j}f(n-m,t-j).
\)
Thus we obtain the general expression
\begin{align}\nonumber%\label{Low:eq10}
N_q^{(i)}(n,t)=\sum_{j=0}^{i}\binom{i}{j}N_q^{(0)}(n-i,t-j)-\sum_{m=0}^{i-1}\sum_{j=0}^{m}\binom{m}{j}D_q(n-3-q-m,t-2-q-j),
\end{align}
for \(1\le i\le q-3\), where \(n\ge t+2q\) is chosen so that \(n-i\ge t-j+q+3\) for all \(i,j\), ensuring the applicability of Lemma~\ref{Low:lm3}.

By the expression for \(N_q^{(0)}(n,t)\) from Eq. (\ref{Low:eq5}), we obtain the following closed-form expression for \(N_q^{(q-3)}(n,t)\)
\begin{align*}
N_q^{(q-3)}(n,t)
&=\sum_{j=0}^{q-3}\binom{q-3}{j}\Big\{
2D_q(n-q+1,t-j-1)+D_q(n-q+1,t-j-2)+2D_q(n-q-2,t-j-2)\\
&\quad+D_q(n-q-2,t-j-3)+D_q(n-2q-2,t-j-2-q)-2D_q(n-q-1,t-j-1)\\
&\quad-D_q(n-q-1,t-j-2)-D_q(n-q-1,t-j-3)-D_q(n-q-2,t-j-4)\\
&\quad-D_q(n-2q-1,t-j-1-q)\Big\}-\sum_{m=0}^{q-4}\sum_{j=0}^{m}\binom{m}{j}D_q(n-3-q-m,t-2-q-j).
\end{align*}

\emph{3) Monotonicity proof.}

We show $N_q^{(i)}(n,t)\ge N_q^{(i-1)}(n,t)$ for $i\in[q-3]$ and $n\geq t+3q$. For general \(i\in [q-3]\), the expression for \(N_q^{(i)}(n,t)\) can be written as a finite linear combination
\(
N_q^{(i)}(n,t) = \sum_{a,b} c_{a,b}^{(i)} \, D_q(n-a,\,t-b),
\)
where the coefficients \(c_{a,b}^{(i)}\) are obtained by matching the shifted arguments of \(D_q(\cdot,\cdot)\).  
Since \(D_q(m,s)\) satisfies the recurrence
\(
D_q(m,s) = \sum_{j=0}^{q-1} D_q(m-1-j,\,s-j)
\)
for $m>s$, and \(N_q^{(i)}(n,t)\) is a linear combination of such terms, it follows that \(N_q^{(i)}(n,t)\) obeys the same recurrence
\[
N_q^{(i)}(n,t) = \sum_{j=0}^{q-1} N_q^{(i)}(n-1-j,\,t-j),
\]
for all \(i\in [q-3]\) and \(n\geq t+2q+1\).

From the recurrence relation, for all \(i\in[q-3]\) and \(n\ge t+2q+1\), we have
\[
N_q^{(i)}(n,t)-N_q^{(i-1)}(n,t)=N_q^{(i-1)}(n-1,t)+N_q^{(i-1)}(n-1,t-1)-N_q^{(i-1)}(n,t)-D_q(n-3-q,t-2-q),
\]
Note that $N_q^{(i-1)}(n,t)$ satisfies the following recurrence
\(
N_q^{(i-1)}(n,t)=\sum_{k=0}^{q-1}N_q^{(i-1)}(n-1-k,t-k)
\)
for $1\leq i\leq q-3$ and $n\geq t+2q+1$.
Substituting this into the difference gives
\[
N_q^{(i)}(n,t)-N_q^{(i-1)}(n,t)=N_q^{(i-1)}(n-1-q,t-q)-D_q(n-3-q,t-2-q),
\]
for $i\in[q-3]$ and $n\geq t+2q+1$. 

It remains to prove that \(N_q^{(i-1)}(n-1-q,t-q) \ge D_q(n-3-q,t-2-q)\) for all $i\in [q-3]$ and $n\geq t+2q+i$.
\begin{itemize}
  \item If \(t\le q+1\), then $N_q^{(i-1)}(n-1-q,t-q)=D_q(n-3-q,t-2-q)=0$, so \(N_q^{(i)}(n,t)=N_q^{(i-1)}(n,t)\) for all \(i\in[q-3]\) and $n\geq t+3q$.

  \item Assume \(t\ge q+2\). By Eq. \eqref{Low:eq3} in Lemma \ref{Low:lm3}, for any \(m\ge s+2q\) we have \(N_q^{(0)}(m,s)\ge D_q(m-2,s-2)\).
  Consequently,
  \(
  N_q^{(0)}(n-1-q,t-q) \ge D_q(n-3-q,t-2-q)
  \)
  and
  \(
  N_q^{(1)}(n,t)\ge N_q^{(0)}(n,t)\ge D_q(n-2,t-2)
  \)
  for \(n\ge t+2q+1\).

  By induction, assume that \(N_q^{(i-1)}(n,t)\ge D_q(n-2,t-2)\) for \(n\ge t+2q+i-1\). Then the same argument gives
  \[
  N_q^{(i)}(n,t)\ge N_q^{(i-1)}(n,t)\ge D_q(n-2,t-2)
  \]
  for \(n\ge t+2q+i\). Thus, \(N_q^{(i)}(n,t)\ge N_q^{(i-1)}(n,t)\) for all \(i\in[q-3]\) and \(n\ge t+3q\).
\end{itemize}

\emph{4) Asymptotic expansion for $N_q^{(i)}(n,t)$.}

For sufficiently large $n$, we have $N_q^{(i)}(n,t)-N_q^{(i-1)}(n,t)=N_q^{(i-1)}(n-1-q,t-q)-D_q(n-3-q,t-2-q)$ for $i\in[q-3]$. By Lemma \ref{Low:lm3}, when $i=1$, for $t\geq q+2$,
\begin{align*}
N_q^{(1)}(n,t)-N_q^{(0)}(n,t)=\frac{6}{(t-q-2)!}n^{t-q-2}-\frac{1}{(t-q-2)!}n^{t-q-2}+O(n^{t-q-3})=\frac{5}{(t-q-2)!}n^{t-q-2}+O(n^{t-q-3}),
\end{align*}
and $N_q^{(1)}(n,t)=\frac{6}{(t-2)!}n^{t-2}+O(n^{t-3})$.
Similarly, for $i\in [2, q-3]$ and $t\geq q+2$,
\begin{align*}
N_q^{(i)}(n,t)-N_q^{(i-1)}(n,t)=\frac{6}{(t-q-2)!}n^{t-q-2}-\frac{1}{(t-q-2)!}n^{t-q-2}+O(n^{t-q-3})=\frac{5}{(t-q-2)!}n^{t-q-2}+O(n^{t-q-3}),
\end{align*}
and $N_q^{(i)}(n,t)=\frac{6}{(t-2)!}n^{t-2}+O(n^{t-3})$. Moreover, we have
\begin{align}
N_q^{(i)}(n,t)-N_q^{(0)}(n,t)=\frac{5i}{(t-q-2)!}n^{t-q-2}+O(n^{t-q-3}), \label{Low:eq11}
\end{align}
for $i\in [q-3]$ and $t\geq q+2$. 

By Eq. (\ref{Low:eq11}) and Lemma \ref{Low:lm3}, we have $N_q^{(i)}(n,t)=\frac{6}{(t-2)!}n^{t-2}-\frac{3t+13}{(t-3)!}n^{t-3}+\frac{3t^2+25t+64}{4(t-4)!}n^{t-4}+O(n^{t-5})$
for $q\geq 4$ and $i\in [q-3]$.

\end{proof}

By the monotonicity property established in Lemma~\ref{Low:lm4}, namely $N_q^{(i+1)}(n,t)\ge N_q^{(i)}(n,t)$, $N_q^{(q-3)}(n,t)$ is the maximum value among all $N_q^{(i)}(n,t)$ for $0\leq i\leq q-3$ and $n\geq t+3q$. For notational convenience, we define
\(
N_q(n,t) = N_q^{(q-3)}(n,t)
\)
for all $q\ge3$ and $n\geq t+3q$. 
Using the formula for the difference between \(N_q^{(0)}(n,t)\) and \(N_{q-1}^{(0)}(n,t)\) from Remark~\ref{Low:rmk1}, together with the recurrence relation in Lemma~\ref{Low:lm4}, we obtain the corresponding result for \(N_q(n,t)\).
\begin{remark}\label{Low:rmk2}
Let $q\ge4$, $t\ge2$, and $n\geq t+3q$. Then
\[
N_q(n,t)-N_{q-1}(n,t) =
\begin{cases}
0, & \text{if } q\ge t+1, \\
5, & \text{if } q=t, \\
\dfrac{(6t-6q+5)n^{t-q}}{(t-q)!} + O(n^{t-q-1}), & \text{if } q\le t-1~\text{and $n$ is sufficiently large}.
\end{cases}
\]
\end{remark}

Furthermore, by the explicit asymptotic expansion of $N_q^{(q-3)}(n,t)$ in Lemma \ref{Low:lm4}, a lower bound on $N_q(n,2,t)$ follows immediately, as summarized in the theorem below. 

\begin{theorem}
\label{Low:thm1}
Let $q\ge3$, $t\ge2$, and $n\geq t+3q$. Then
\begin{equation}
N_q(n,2,t)\geq N_q(n,t).\nonumber
\end{equation}
Moreover, for $q\geq 4$ and sufficiently large $n$, 
\begin{align*}
N_q(n,2,t)\geq \frac{6}{(t-2)!}n^{t-2}-\frac{3t+13}{(t-3)!}n^{t-3}+\frac{3t^2+25t+64}{4(t-4)!}n^{t-4}+O(n^{t-5}).
\end{align*}
\end{theorem}

\section{The upper bound}
\label{sec5}

In this section, we prove that the lower bound on $N_q(n,2,t)$ derived in the previous section is tight for all $t\ge 2$, $q\ge 3$, and sufficiently large $n$. Specifically, we establish the equality $N_q(n,2,t)=N_q(n,t)$ under these conditions, thereby showing that the lower bound is also an upper bound. The proof is based on an exhaustive case analysis of all pairs of vectors $\mathbf{x},\mathbf{y}\in\Sigma_q^n$.

To carry out this analysis, we first need several auxiliary results on $N_q(n,1,t+1,t)$ for certain values of $n$ and $t$. These are stated in Lemmas \ref{Upp:lm1} and \ref{Upp:lm3}, whose proofs are deferred to Appendix~\ref{APP-A}.

\begin{lemma}
\label{Upp:lm1}
For $n\geq 4$ and $q\geq 3$, we have
$N_q(n,1,2,1)\le 3.$
\end{lemma}

For convenience, we define
\(
M_q(n,t):= D_q(n,t) - D_q(n-2,t) + D_q(n-3,t-1)
\)
for all $t\ge1$ and $n\ge t+2$. Using the definition of $M_q(n,t)$ and the recurrence relation in Eq. \eqref{def:eq1}, we derive the following lemma.

\begin{lemma}
\label{Upp:lm2}
Let $n\ge t+2, t\geq 1$, and $q\geq 3$. Then $M_q(n,t) \le D_q(n,t)$. Moreover, for every $l\in[q-1]$,
\(
N_q(n-1,1,t) + \sum_{i=1}^{l} D_q(n-1-i,t-i)
\le
\sum_{i=0}^{l} M_q(n-1-i,t-i).
\)
\end{lemma}

\begin{proof}
The first inequality follows directly from 
\(
M_q(n,t) = D_q(n,t) - \bigl[D_q(n-2,t) - D_q(n-3,t-1)\bigr],
\)
and the fact that $D_q(n-2,t) \ge D_q(n-3,t-1)$ for all $n\geq t+2$.

We define
\(
\Delta:= \sum_{i=0}^{l} M_q(n-1-i,t-i) - N_q(n-1,1,t) - \sum_{i=1}^{l} D_q(n-1-i,t-i).
\)
Since \(M_q(n,t) = D_q(n,t) - D_q(n-2,t) + D_q(n-3,t-1)\) and  \(N_q(n-1,1,t) = D_q(n-1,t) - D_q(n-2,t) + D_q(n-3,t-1)\),
we obtain
\(
\Delta = D_q(n-2,t) - D_q(n-3,t) - D_q(n-3,t-1) + D_q(n-4-l,t-1-l)
\)
for $l\in [q-1]$.

We now consider two cases:
\begin{itemize}
  \item For $l\in [q-1]$, if $n \ge t+3$, then, by Eqs. \eqref{def:eq1} and \eqref{Pro:eq3}, we have
  \(
  \Delta=D_q(n-4-l,t-1-l)-D_q(n-3-q,t-q) \ge 0.
  \)
  \item For $l\in [q-1]$, if $n = t+2$, then 
  \(
  \Delta = D_q(t,t) - D_q(t-1,t-1) = 1 - 1 = 0 \ge 0.
  \)
\end{itemize}

In both cases, $\Delta \ge 0$, which implies that
\(
\sum_{i=0}^{l} M_q(n-1-i,t-i) \ge N_q(n-1,1,t) + \sum_{i=1}^{l} D_q(n-1-i,t-i).
\)
This completes the proof.
\end{proof}

\begin{lemma}
\label{Upp:lm3}
Let $q\ge 3$, $t\ge1$, and $n\ge t+3$. Then $N_q(n,1,t+1,t) = M_q(n,t)$.
\end{lemma}

We now consider pairs of vectors $\mathbf{x}=(x_1,\dots,x_n),\mathbf{y}=(y_1,\dots,y_n)\in\Sigma_q^n$ with $d_L(\mathbf{x},\mathbf{y})\ge2$ and $x_1\neq y_1$. Our goal is to prove that $|D_q(\mathbf{x},\mathbf{y};t,t)|\le N_q(n,t)$ for all sufficiently large $n$. To this end, we first prove several auxiliary lemmas.

\begin{lemma}
\label{Upp:lm4}
Let $q\ge3$, $n\geq t+1$, and let $\mathbf{x}=(x_1,\dots,x_n),\mathbf{y}=(y_1,\dots,y_n)\in\Sigma_q^n$ be sequences with $x_1\neq y_1$ and $\mathbf{x}_{[2,n]}=\mathbf{y}_{[2,n]}$. Then
\[
|D_t(\mathbf{x})| + |D_t(\mathbf{y})| \le 2D_q(n,t) + D_q(n-q-1,t-q) - D_q(n-q,t-q+1).
\]
Moreover, if $\mathbf{x}_{[2,n]}\notin \mathbf{C}_q(n-1)$, then
\(
|D_t(\mathbf{x})| + |D_t(\mathbf{y})| \le 2D_q(n,t)-2n^{t-q+1}+O(n^{t-q}).
\) 
\end{lemma}

\begin{proof}
For convenience, let $\mathbf{u}=(u_1,\cdots,u_{n-1})=\mathbf{x}_{[2,n]}=\mathbf{y}_{[2,n]}$, and set $\mathcal{X}=D_{t-1}(\mathbf{u})$. Without loss of generality, assume $x_1=0$ and $y_1=1$. By Lemmas~\ref{Pro:lm1} and~\ref{Pro:lm2}, we have \(|\mathcal{X}| = \sum_{i=0}^{q-1} |\mathcal{X}^i|\), where
\[
\begin{aligned}
\mathcal{X}^i &= i \circ  D_{t-1-\ell_i}(\mathbf{u}_{[2+\ell_i,\,n-1]}) , \quad i = 0,1,\dots,q-1,
\end{aligned}
\]
and $\ell_i$ is non‑negative integers defined as the smallest indices such that $u_{1+\ell_i} = i$ (if a required symbol does not appear, the corresponding term vanishes). Moreover, the indices $\ell_i$ are distinct.

Since $\mathbf{x}=(0,\mathbf{u})$ and $\mathbf{y}=(1,\mathbf{u})$, it follows that
\begin{align*}
|D_t(\mathbf{x})^{i}|=|D_{t-1}(\mathbf{u})^i|\leq D_q(n-2-\ell_i,t-1-\ell_i)~\text{and}~  |D_t(\mathbf{y})^{j}|=|D_{t-1}(\mathbf{u})^j|\leq D_q(n-2-\ell_j,t-1-\ell_j),
\end{align*}
for $i\in [q-1]$ and $j\in \{0\}\cup [2,q-1]$. Moreover, $|D_t(\mathbf{x})^{0}|=|D_t(\mathbf{y})^{1}|=|D_{t}(\mathbf{u})|\leq D_q(n-1,t).$ Thus,
\begin{align*}
|D_t(\mathbf{x})|+|D_t(\mathbf{y})|&=\sum\limits_{i=0}^{q-1}\big(|D_t(\mathbf{x})^{i}|+|D_t(\mathbf{y})^i|\big)
=2|D_t(\mathbf{u})|+|D_{t-1}(\mathbf{u})^1|+|D_{t-1}(\mathbf{u})^0|+2\sum\limits_{i=2}^{q-1}|D_{t-1}(\mathbf{u})^i|\\
&\leq 2D_q(n-1,t)+D_q(n-2-\ell_0,t-1-\ell_0)+D_q(n-2-\ell_1,t-1-\ell_1)+2\sum\limits_{i=2}^{q-1}D_q(n-2-\ell_i
,t-1-\ell_i)\\
&\overset{(a)}{\leq} D_q(n-q-1,t-q)+D_q(n-q,t-q+1)+2\sum\limits_{i=0}^{q-2}D_q(n-1-i
,t-i)\\
&\overset{(b)}{=}2D_q(n,t)+D_q(n-q-1,t-q)-D_q(n-q, t-q+1),
\end{align*}
where $(a)$ follows from Eq. \eqref{Pro:eq3} together with the distinctness of the $\ell_i$, and $(b)$ follows from Eq. \eqref{def:eq1} for $n\geq t+1$.

Moreover, if $\mathbf{x}_{[2,n]}\notin \mathbf{C}_q(n-1)$ then $\mathbf{x},\mathbf{y}\notin \mathbf{C}_q(n)$. Combining this with Lemma \ref{Pro:lm11}, we have 
$|D_t(\mathbf{x})|\leq D_q(n,t)-n^{t-q+1}+O(n^{t-q})$ and $|D_t(\mathbf{y})|\leq D_q(n,t)-n^{t-q+1}+O(n^{t-q})$.
Thus, 
\(
|D_t(\mathbf{x})| + |D_t(\mathbf{y})| \le 2D_q(n,t)-2n^{t-q+1}+O(n^{t-q}).
\) 
\end{proof}

\begin{lemma}
\label{Upp:lm5}
Let $q\ge3$ and $n$ be sufficiently large. For sequences $\mathbf{x}=(x_1,\dots,x_n),\mathbf{y}=(y_1,\dots,y_n)\in\Sigma_q^n$ satisfying
$x_1\neq y_1$, $x_1=y_2$, $x_2=y_1$, and $\mathbf{x}_{[3,n]}\neq\mathbf{y}_{[3,n]}$, we have
\[
|D_t(\mathbf{x})\cap D_t(\mathbf{y})| \le 2D_q(n-2,t-1) + \sum_{i=2}^{q-1}D_q(n-1-i,t-i) + D_q(n-5,t-2)-D_q(n-3,t-2).
\]
\end{lemma}

\begin{proof}
Let \(\mathcal{X} = D_t(\mathbf{x}) \cap D_t(\mathbf{y})\). Without loss of generality, assume \(x_1 = 1\) and \(y_1 = 0\). Write \(\mathbf{x} = (1,0,\mathbf{u})\) and \(\mathbf{y} = (0,1,\mathbf{w})\), where \(\mathbf{u} = (u_1,\dots,u_{n-2})\) and \(\mathbf{w} = (w_1,\dots,w_{n-2})\) with \(\mathbf{u} \neq \mathbf{w}\).

 By Lemmas~\ref{Pro:lm1} and~\ref{Pro:lm2}, we have \(|\mathcal{X}| = \sum_{i=0}^{q-1} |\mathcal{X}^i|\), where
\(\mathcal{X}^0 = 0 \circ \bigl( D_{t-1}(\mathbf{u}) \cap D_t(1,\mathbf{w}) \bigr),\)
\(\mathcal{X}^1 = 1 \circ \bigl( D_{t}(0,\mathbf{u}) \cap D_{t-1}(\mathbf{w}) \bigr),\)
\[\mathcal{X}^i = i \circ \bigl( D_{t-2-\ell_i}(\mathbf{x}_{[4+\ell_i,\,n]}) \cap D_{t-2-\ell_i^*}(\mathbf{y}_{[4+\ell_i^*,\,n]}) \bigr), \quad i = 2,\dots,q-1,
\]
with non‑negative integers \(\ell_i,\ell_i^*\) defined as the smallest indices such that $x_{3+\ell_i}=i$ and $y_{3+\ell_i^*}=i$ for $i\ge2$. Moreover, the indices $\ell_i$ for $i\geq 2$ are  pairwise distinct.

Since the indices $\ell_i$ for $i\geq 2$ are pairwise distinct, it follows that
\(
\sum\limits_{i=2}^{q-1}|\mathcal{X}^i|\leq \sum\limits_{i=2}^{q-1}|D_{t-2-\ell_i}(\mathbf{x}_{[4+\ell_i,n]})|\leq \sum\limits_{i=2}^{q-1}D_q(n-1-i,t-i).
\)

If $u_1 \neq 1$, then $\mathbf{u} \notin D_1(1,\mathbf{w})$ because $\mathbf{u} \neq \mathbf{w}$. Consequently, by Lemma \ref{Upp:lm3},
\begin{align*}
|\mathcal{X}^{0}|\leq N_q(n-2,1,t,t-1)=D_q(n-2,t-1)-D_q(n-4,t-1)+D_q(n-5,t-2).
\end{align*}
Since $|\mathcal{X}^1| \le |D_{t-1}(\mathbf{w})| \le D_q(n-2,t-1)$, for $q\ge 3$ and $n\geq t+3$,
\[
|\mathcal{X}| = |\mathcal{X}^{0}| + |\mathcal{X}^{1}| + \sum_{i=2}^{q-1} |\mathcal{X}^i| 
\le 2D_q(n-2,t-1) + \sum_{i=2}^{q-1} D_q(n-1-i,t-i) + D_q(n-5,t-2) - D_q(n-4,t-1).
\]
A symmetric argument holds for the case $w_1 \neq 0$.

Now consider the case $u_1 = 1$ and $w_1 = 0$. For all $i \ge 2$, we have $\ell_i \ge 1$ and the indices $\ell_i$ are pairwise distinct. Hence
\[
\sum_{i=2}^{q-1} |\mathcal{X}^i|
\le \sum_{i=2}^{q-1} \bigl| D_{t-2-\ell_i}(\mathbf{x}_{[4+\ell_i,\,n]}) \bigr|
\le \sum_{i=2}^{q-1} D_q(n-2-i,t-1-i).
\]
Since $|\mathcal{X}^i| \le D_q(n-2,t-1)$ for $i \in \{0,1\}$, we get
\(
|\mathcal{X}|
\le 2D_q(n-2,t-1) + \sum_{i=2}^{q-1} D_q(n-1-i,t-i)+D_q(n-1-q,t-q)-D_q(n-3,t-2).
\)

For sufficiently large $n$, we have  $D_q(n-1-q,t-q)\leq D_q(n-5,t-2)$ and  $D_q(n-3,t-2) \leq D_q(n-4,t-1)$.
Thus, for sufficiently large $n$, we have $|\mathcal{X}|\le 2D_q(n-2,t-1) + \sum_{i=2}^{q-1} D_q(n-1-i,t-i) + D_q(n-5,t-2) - D_q(n-3,t-2)$.
\end{proof}

Combining Lemmas \ref{Upp:lm4} and \ref{Upp:lm5}, we now focus on pairs of sequences $\mathbf{x}=(x_1,\dots,x_n),\mathbf{y}=(y_1,\dots,y_n)\in\Sigma_q^n$ that differ in their first coordinate, i.e., $x_1\neq y_1$. Building on the recurrence relations established in Lemma \ref{Low:lm4}, we prove the following upper bound for $t\ge2$ and sufficiently large $n$.

\begin{lemma}
\label{Upp:lm6}
Let $q\ge 3, t\ge2$, and $n$ be sufficiently large. For all $\mathbf{x}=(x_1,...,x_n),\mathbf{y}=(y_1,...,y_n)\in\Sigma_q^n$ satisfying $d_L(\mathbf{x},\mathbf{y})\ge2$ and $x_1\neq y_1$,
\(
|D_q(\mathbf{x},\mathbf{y};t,t)| \le N_q^{(0)}(n,t).
\)
\end{lemma}

\begin{proof}
Let $\mathcal{X} = D_t(\mathbf{x}) \cap D_t(\mathbf{y})$. Without loss of generality, assume $x_1 = 1$ and $y_1 = 0$. By Lemmas \ref{Pro:lm1} and \ref{Pro:lm2}, we have
\(
|\mathcal{X}| = \sum_{i=0}^{q-1} |\mathcal{X}^i|,
\)
where $\mathcal{X}^0 = 0 \circ \bigl( D_{t-1-\ell_0}(\mathbf{x}_{[3+\ell_0,n]}) \cap D_t(\mathbf{y}_{[2,n]}) \bigr),$ 
$\mathcal{X}^1 = 1 \circ \bigl( D_t(\mathbf{x}_{[2,n]}) \cap D_{t-1-\ell_1^*}(\mathbf{y}_{[3+\ell_1^*,n]}) \bigr),$ and 
\begin{align*}
\mathcal{X}^i = i \circ \bigl( D_{t-1-\ell_i}(\mathbf{x}_{[3+\ell_i,n]}) \cap D_{t-1-\ell_i^*}(\mathbf{y}_{[3+\ell_i^*,n]}) \bigr), \quad i = 2,\dots,q-1.
\end{align*}
Here $\ell_0, \ell_1^*, \ell_i, \ell_i^*$ are the smallest non-negative integers such that
\(
x_{2+\ell_0}=0, y_{2+\ell_1^*}=1, x_{2+\ell_i}=i, y_{2+\ell_i^*}=i\ \ (2\leq i\leq q-1).
\)
Note that $\ell_0=0$ if and only if  $x_2=0$, $\ell_1^*=0$ if and only if  $y_2=1$, and for $2\leq i\le q-1$, $\ell_i=0$ if and only if  $x_2=i$, and $\ell_i^*=0$ if and only if  $y_2=i$. 

We analyze $|\mathcal{X}|$ in the following four cases according to the values of $\ell_0$ and $\ell^*_1$. \textbf{Case A:} \(\ell_0 = \ell_1^* = 0\); \textbf{Case B:} \(\ell_0 \ge 1\) and \(\ell_1^* \ge 1\); \textbf{Case C:} \(\ell_0 = 0\) and \(\ell_1^* \ge 1\); \textbf{Case D:} \(\ell_0 \ge 1\) and \(\ell_1^* = 0\).

\textbf{Case A:} \(\ell_0 = \ell_1^* = 0\). 

Then \(x_2 = 0\), \(y_2 = 1\), and $\ell_i, \ell_i^* \ge 1$ for all $i\in[2,q-1]$. Since $d_L(\mathbf{x},\mathbf{y})\ge2$, we have $\mathbf{x}_{[3,n]}\notin D_1(\mathbf{y}_{[2,n]})$ and $\mathbf{y}_{[3,n]}\notin D_1(\mathbf{x}_{[2,n]})$. By Lemma \ref{Upp:lm3}, for $n\geq t+4$ and $t\geq 2$,
\begin{align}
|\mathcal{X}^0|&\leq N_q(n-2,1,t,t-1)=D_q(n-2,t-1)-D_q(n-4,t-1)+D_q(n-5,t-2),\nonumber\\%\label{Upp::eq1}\\
|\mathcal{X}^1|&\leq N_q(n-2,1,t,t-1)=D_q(n-2,t-1)-D_q(n-4,t-1)+D_q(n-5,t-2).\label{Upp:eq2}
\end{align}

We further distinguish two subcases based on the values of \(\ell_i\) and \(\ell_i^*\).
\textbf{Case A1:}  \(\min\{\ell_i\mid i\in [2,q-1]\}\geq 2\) or \(\min\{\ell_i^*\mid i\in [2,q-1]\}\geq 2\);
\textbf{Case A2:}  \(\min\{\ell_i\mid i\in [2,q-1]\}=\min\{\ell_i^*\mid i\in [2,q-1]\}= 1\).  

\textbf{Case A1:}   \(\min\{\ell_i\mid i\in [2,q-1]\}\geq 2\) or \(\min\{\ell_i^*\mid i\in [2,q-1]\}\geq 2\).

For each $i\in[2,q-1]$, we have
\begin{align}\label{Upp:eq3}
|\mathcal{X}^i| \le \min\bigl\{ |D_{t-1-\ell_i}(\mathbf{x}_{[3+\ell_i,\,n]})|,\; |D_{t-1-\ell_i^*}(\mathbf{y}_{[3+\ell_i^*,\,n]})| \bigr\}\leq D_q(n-4,t-3).
\end{align}

%Since for every \(i \in [2, q-1]\) we have \(\ell_i \ge 2\) or \(\ell_i^* \ge 2\), and the indices \(\ell_i\) (or \(\ell_i^*\)) are distinct, it follows from Equation \eqref{intr::eq3} under the condition \(n \ge t+2\) that
%\[
%\sum_{i=2}^{q-1} |\mathcal{X}^i| \le \sum_{i=2}^{q-1} D_q(n-2-i,\; t-1-i).
%\]
%For convenience, let $x_3=a$ and $y_3=b$. Here, there exists at least one element $0$ or $1$ in $a,b$. Without loss of gen, consider $a=0$ or $1$. We only discuss $a=0$ because of $a=1$ similarly.

Let $x_3=a$ and $y_3=b$. Without loss of generality, assume that \(\min\{\ell_i\mid i\in [2,q-1]\}\geq 2\). Then $a=0$ or $1$. We only discuss $a=0$ (the case $a=1$ is symmetric). We further decompose $\mathcal{X}^0$ and $\mathcal{X}^1$ with $(x_1,x_2,x_3)=(1,0,0)$ and $(y_1,y_2,y_3)=(0,1,b)$:
$|\mathcal{X}^0|=\sum\limits_{i=0}^{q-1}|\mathcal{X}^{(0,i)}|$ and $|\mathcal{X}^1|=\sum\limits_{i=0}^{q-1}|\mathcal{X}^{(1,i)}|$ 
 with $|\mathcal{X}^{(0,0)}|=|D_{t-1}(x_4,x_{5},...,x_n)\cap D_{t-1-\ell_{(0,0)}^*}(y_{\ell_{(0,0)}^*+4},...,y_n)|$, and
\begin{align}
|\mathcal{X}^{(0,i)}|=|D_{t-2-\ell_{(0,i)}}(x_{5+\ell_{(0,i)}},...,x_n)\cap D_{t-\ell_{(0,i)}^*}(y_{3+\ell_{(0,i)}^*},...,y_n)|,\label{Upp:eq4}
\end{align}
where  $\ell_{(0,i)}, \ell_{(0,i)}^* \ge 0$ for $i \in \Sigma_q \setminus \{0\}$, $\ell_{(0,0)}^* \ge 0$, and the indices $\ell_{(0,i)}$ are pairwise distinct. 

We now consider two subcases according to the value of $b$. 
\begin{itemize}
  \item If $b\neq 0$, then $\ell_{(0,0)}^*\geq 1$. Since the indices $\ell_{(0,i)}$ are pairwise distinct for $i\in \Sigma_q\setminus\{0\}$,
\begin{align}
|\mathcal{X}^0|&=\sum\limits_{i=0}^{q-1}|\mathcal{X}^{(0,i)}|\leq \sum\limits_{i\in \Sigma_q\setminus\{0\}} |D_{t-2-\ell_{(0,i)}}(x_{5+\ell_{(0,i)}},...,x_n)|+|D_{t-1-\ell_{(0,0)}^*}(y_{4+\ell_{(0,0)}^*},...,y_n)|\nonumber\\
&\leq  2D_q(n-4,t-2)+\sum\limits_{i=1}^{q-2}D_q(n-4-i,t-2-i).\label{Upp:eq5}
\end{align}
For sufficiently large $n$, by Eqs. \eqref{Upp:eq2}, \eqref{Upp:eq3}, and \eqref{Upp:eq5} we have
\begin{align}
|\mathcal{X}|=\sum\limits_{i=0}^{q-1}|\mathcal{X}^i|&\leq 2D_q(n-4,t-2)+(q-2)D_q(n-4,t-3)+\sum_{i=1}^{q-1} D_q(n-4-i,\; t-2-i)\nonumber\\
&~~~~+D_q(n-2,t-1)-D_q(n-4,t-1)+D_q(n-5,t-2)\nonumber\\
&\overset{(a)}{\leq}\frac{5}{(t-2)!}n^{t-2}+O(n^{t-3})\overset{(b)}{<}N_q^{(0)}(n,t),\nonumber
\end{align}
where $(a)$ follows from Eqs. \eqref{intr:eq2} and \eqref{def:eq1}, and $(b)$ follows from Lemma \ref{Low:lm3}.
  \item If $b=0$, then $\ell_{(0,0)}^*=0$, $\mathbf{x}_{[1,3]}=(1,0,0)$, and $\mathbf{y}_{[1,3]}=(0,1,0)$. Since $d_L(\mathbf{x},\mathbf{y})\geq 2$, we have $\mathbf{x}_{[4,n]}\neq \mathbf{y}_{[4,n]}$. Thus,
\begin{align}
|\mathcal{X}^{(0,0)}|\leq N_q(n-3,1,t-1)=D_q(n-3,t-1)-D_q(n-4,t-1)+D_q(n-5,t-2).\label{Upp:eq6}
\end{align}

We further distinguish two subcases based on the value of $x_4$.
\textbf{Case A1-1:}  $x_4\neq 1$;
\textbf{Case A1-2:}  $x_4=1$.  
\begin{enumerate}
  \item \textbf{Case A1-1:}  $x_4\neq 1$. 

We divide this case into two subcases based on the value of $x_4$.  
\begin{itemize}
  \item \textit{Subcase 1:} $x_4 = 0$. Then $\ell_{(0,i)}\geq 1$ for $i\in [q-1]$. By Eq. \eqref{Upp:eq4}, we have \(|\mathcal{X}^{(0,i)}|\leq D_q(n-5,t-3)\)
for $i\in [q-1]$.
\item \textit{Subcase 2:} $x_4 \notin \{0,1\}$. We further consider two possibilities according to whether $x_4$ equals $y_4$.
\begin{itemize}
  \item[$(1)$] When $x_4=y_4$, we have $|\mathcal{X}^{(0,x_4)}|=|D_{t-2}(x_{5},...,x_n)\cap D_{t-2}(y_{5},...,y_n)|\leq N_q(n-4,1,t-2)=D_q(n-4,t-2)-D_q(n-5,t-2)+D_q(n-6,t-3)$ and $|\mathcal{X}^{(0,i)}|\leq D_q(n-5,t-3)$
for $i\in [q-1]\setminus\{x_4\}$;
  \item[$(2)$] When $x_4\neq y_4$, we have  $\ell_{(0,x_4)}^*\geq 3$ and $\ell_{(0,i)}\geq 1$ for $i\in [q-1]\setminus\{x_4\}$. Thus, by Eq. \eqref{Upp:eq4}, we have 
$|\mathcal{X}^{(0,i)}|\leq D_q(n-5,t-3)$ for $i\in [q-1]$.
\end{itemize}
\end{itemize}

In both subcases, for sufficiently large $n$, combining the above discussion with Eqs.~\eqref{Upp:eq2}, \eqref{Upp:eq3}, and~\eqref{Upp:eq6}, we obtain
\[
|\mathcal{X}| = \sum_{i=0}^{q-1} |\mathcal{X}^i|
= |\mathcal{X}^{(0,0)}| + |\mathcal{X}^1| + \sum_{i=1}^{q-1} |\mathcal{X}^{(0,i)}| +\sum_{i=2}^{q-1} |\mathcal{X}^i| 
\overset{(a)}{\leq}\frac{5}{(t-2)!}n^{t-2}+O(n^{t-3})
\overset{(b)}{<}N_q^{(0)}(n,t),
\]
where $(a)$ follows from Eqs. \eqref{intr:eq2} and \eqref{def:eq1}, and $(b)$ follows from Lemma \ref{Low:lm3}.
  \item \textbf{Case A1-2:}  $x_4=1$.  For convenience, set $y_4=c$. Then $\mathbf{x}_{[1,4]}=(1,0,0,1)$ and $\mathbf{y}_{[1,4]}=(0,1,0,c)$. For $i\in [2,q-1]$, by the definition of \(\mathcal{X}^i\), we have
\(
|\mathcal{X}^i| \le \min\bigl\{ |D_{t-1-\ell_i}(\mathbf{x}_{[3+\ell_i,\,n]})|,\; |D_{t-1-\ell_i^*}(\mathbf{y}_{[3+\ell_i^*,\,n]})| \bigr\}\leq D_q(n-5,t-4).
\)
By Lemmas \ref{Pro:lm1} and \ref{Pro:lm2}, we have $|\mathcal{X}^{0}|=\sum\limits_{i=0}^{q-1}|\mathcal{X}^{(0,i)}|$ and $|\mathcal{X}^{1}|=\sum\limits_{i=0}^{q-1}|\mathcal{X}^{(1,i)}|$ with 
\begin{align*}
|\mathcal{X}^{(0,0)}|&=|D_{t-1}(1,x_{5},...,x_n)\cap D_{t-1}(c,y_{5},...,y_n)|\overset{(a)}{\leq} N_q(n-3,1,t-1)\\
&=D_q(n-3,t-1)-D_q(n-4,t-1)+D_q(n-5,t-2),\\
|\mathcal{X}^{(0,1)}|&=|D_{t-2}(x_{5},...,x_n)\cap D_{t}(0,c,y_{5},...,y_n)|\leq D_q(n-4,t-2),\\
\sum\limits_{i=2}^{q-1}|\mathcal{X}^{(0,i)}|&=\sum\limits_{i=2}^{q-1}|D_{t-3-\ell_{(0,i)}}(x_{6+\ell_{(0,i)}},...,x_n)\cap D_{t-2-\ell_{(0,i)}^*}(y_{5+\ell_{(0,i)}^*},...,y_n)|\overset{(b)}{\leq} \sum\limits_{i=1}^{q-2} D_q(n-4-i,\; t-2-i),\\
|\mathcal{X}^{(1,0)}|&=|D_{t}(0,1,x_{5},...,x_n)\cap D_{t-1}(c,y_{5},...,y_n)|\overset{(c)}{\leq} N_q(n-3,1,t,t-1)\\
&=D_q(n-3,t-1)-D_q(n-5,t-1)+D_q(n-6,t-2),\\
|\mathcal{X}^{(1,1)}|&=|D_{t-2}(x_{5},...,x_n)\cap D_{t-2-\ell_{(1,1)}^*}(y_{5+\ell_{(1,1)}^*},...,y_n)|\overset{(d)}{\leq} \max\{N_q(n-4,1,t-2),D_q(n-5,t-3)\}\\
&=D_q(n-4,t-2)-D_q(n-5,t-2)+D_q(n-6,t-3),\\
\sum\limits_{i=2}^{q-1}|\mathcal{X}^{(1,i)}|&=\sum\limits_{i=2}^{q-1}|D_{t-3-\ell_{(1,i)}}(x_{6+\ell_{(1,i)}},...,x_n)\cap D_{t-2-\ell_{(1,i)}^*}(y_{5+\ell_{(1,i)}^*},...,y_n)|\overset{(e)}{\leq} \sum\limits_{i=1}^{q-2} D_q(n-4-i,\; t-2-i).
\end{align*}
Here $\ell_{(1,1)}^*, \ell_{(0,i)},\ell_{(0,i)}^*,\ell_{(1,i)},\ell_{(1,i)}^* \ge 0$ for $i \in \Sigma_q\setminus\{0,1\}$, and the indices $\ell_{(0,i)}$ (respectively $\ell_{(0,i)}^*$, $\ell_{(1,i)}$, $\ell_{(1,i)}^*$) are pairwise distinct. Moreover, $(a)$ follows from $(1,x_{5},...,x_n)\neq (c,y_{5},...,y_n)$; $(b)$ and $(e)$ follow from  the indices $\ell_{(0,i)}$ (or $\ell_{(1,i)}$) are pairwise distinct; $(c)$ follows from $(c,y_{5},...,y_n)\notin D_1(0,1,x_{5},...,x_n)$; $(d)$ follows from the fact that if $\ell_{(1,1)}^* = 0$, then $c = 1$ and $(x_5,\dots,x_n) \neq (y_5,\dots,y_n)$; otherwise, $|\mathcal{X}^{(1,1)}|\leq D_q(n-5,t-3)$. Thus, for sufficiently large $n$,
\begin{align}
|\mathcal{X}|&=\sum\limits_{i=0}^{q-1}|\mathcal{X}^{(0,i)}|+\sum\limits_{i=0}^{q-1}|\mathcal{X}^{(1,i)}|+\sum\limits_{i=2}^{q-1}|\mathcal{X}^i|\nonumber\\
&\leq \big(D_q(n-3,t-1)-D_q(n-4,t-1)+D_q(n-5,t-2)\big)+\big(D_q(n-3,t-1)-D_q(n-5,t-1)\nonumber\\
&~~~~+D_q(n-6,t-2)\big)+2D_q(n-4,t-2)-D_q(n-5,t-2)+D_q(n-6,t-3)\nonumber\\
&~~~~+2\sum_{i=1}^{q-2} D_q(n-4-i,\; t-2-i)+(q-2)D_q(n-5,t-4)\nonumber\\
&\overset{(a)}{=}\frac{6}{(t-2)!}n^{t-2}-\frac{3t+15}{(t-3)!}n^{t-3}+O(n^{t-4})\overset{(b)}{<}N_q^{(0)}(n,t),\label{Upp:eq7}
\end{align}
where $(a)$ follows from Lemma \ref{Pro:lm10}, and $(b)$ follows from Lemma \ref{Low:lm3}.
\end{enumerate}
 
\end{itemize}

\textbf{Case A2:}  \(\min\{\ell_i\mid i\in [2,q-1]\}=\min\{\ell_i^*\mid i\in [2,q-1]\}= 1\). That is, there exist indices \(i, j \in [2, q-1]\) (not necessarily distinct) such that \(\ell_i = 1\) and \(\ell_j^* = 1\).

\textit{Subcase 1:} $i \neq j$. Without loss of generality, we let \(i=2\) and \(j=3\). Then \((x_1,x_2,x_3)=(1,0,2)\), \((y_1,y_2,y_3)=(0,1,3)\), \(\ell_2^*\geq 2\), \(\ell_3\geq 2\), and \(\ell_i\geq 2\) for every \(i\in [4,q-1]\). Thus, $|\mathcal{X}^i|\leq D_q(n-4,t-3)$
for all \(i\in [2,q-1]\).

By Lemmas \ref{Pro:lm1} and \ref{Pro:lm2}, we have 
$|\mathcal{X}^0|=\sum\limits_{i=0}^{q-1}|\mathcal{X}^{(0,i)}|$ and $|\mathcal{X}^1|=\sum\limits_{i=0}^{q-1}|\mathcal{X}^{(1,i)}|$ with
\begin{align}
|\mathcal{X}^{(0,i)}|&=|D_{t-2-\ell_{(0,i)}}(x_{5+\ell_{(0,i)}},...,x_n)\cap D_{t-\ell_{(0,i)}^*}(y_{3+\ell_{(0,i)}^*},...,y_n)|,\nonumber\\
|\mathcal{X}^{(0,2)}|&=|D_{t-1}(x_4,x_{5},...,x_n)\cap D_{t-2-\ell_{(0,2)}^*}(y_{\ell_{(0,2)}^*+5},...,y_n)|,\nonumber\\
|\mathcal{X}^{(0,3)}|&=|D_{t-2-\ell_{(0,3)}}(x_{\ell_{(0,3)}+5},...,x_n)\cap D_{t-1}(y_4,y_{5},...,y_n)|,\nonumber\\
|\mathcal{X}^{(1,i)}|&=|D_{t-\ell_{(1,i)}}(x_{3+\ell_{(1,i)}},...,x_n)\cap D_{t-2-\ell_{(1,i)}^*}(y_{5+\ell_{(1,i)}^*},...,y_n)|,\nonumber\\
|\mathcal{X}^{(1,2)}|&=|D_{t-1}(x_4,x_{5},...,x_n)\cap D_{t-2-\ell_{(1,2)}^*}(y_{\ell_{(1,2)}^*+5},...,y_n)|,\nonumber\\
|\mathcal{X}^{(1,3)}|&=|D_{t-2-\ell_{(1,3)}}(x_{\ell_{(1,3)}+5},...,x_n)\cap D_{t-1}(y_4,y_{5},...,y_n)|,\nonumber
\end{align}
where
$\ell_{(0,i)},\ell_{(1,i)},\ell_{(0,i)}^*,\ell_{(1,i)}^*\geq 0$  for $i\in \Sigma_q\setminus\{2,3\}$, $\ell_{(0,3)},\ell_{(0,2)}^*,\ell_{(1,3)},\ell_{(1,2)}^*\geq 0$, and the indices $\ell_{(0,i)}$ (respectively $\ell_{(1,i)}$, $\ell_{(0,i)}^*$, $\ell_{(1,i)}^*$) are pairwise distinct. Thus, 
\[|\mathcal{X}^0|\leq \sum\limits_{i\in \Sigma_q\setminus\{2\}} |D_{t-2-\ell_{(0,i)}}(x_{5+\ell_{(0,i)}},...,x_n)|+|D_{t-2-\ell_{(0,2)}^*}(y_{5+\ell_{(0,2)}^*},...,y_n)|\leq  2D_q(n-4,t-2)+\sum\limits_{i=1}^{q-2}D_q(n-4-i,t-2-i),\] 
\[|\mathcal{X}^1|\leq \sum\limits_{i\in \Sigma_q\setminus\{3\}} |D_{t-2-\ell_{(1,i)}^*}(y_{5+\ell_{(1,i)}^*},...,y_n)|+|D_{t-2-\ell_{(1,3)}}(x_{5+\ell_{(1,3)}},...,x_n)|\leq   2D_q(n-4,t-2)+\sum\limits_{i=1}^{q-2}D_q(n-4-i,t-2-i).\] 
For sufficiently large $n$, when $i\neq j$, we have 
\begin{align}
|\mathcal{X}|&=\sum\limits_{i=0}^{q-1}|\mathcal{X}^i|
\leq 4D_q(n-4,t-2)+(q-2)D_q(n-4,t-3)+2\sum_{i=1}^{q-1} D_q(n-4-i,\; t-2-i)\nonumber\\
&\overset{(a)}{\leq}\frac{4}{(t-2)!}n^{t-2}+O(n^{t-3})
\overset{(b)}{\leq}N_q^{(0)}(n,t),\nonumber
\end{align}
where $(a)$ follows from Eq. \eqref{intr:eq2}, and $(b)$ follows from Lemma \ref{Low:lm3}.

\textit{Subcase 2:} $i = j$. Without loss of generality, set $i = j = q-1$. Then $x_3 = y_3 = q-1$. That is,
\(
\mathbf{x}_{[1,3]} = (1,0,q-1)\) and \(\mathbf{y}_{[1,3]} = (0,1,q-1).
\)
Hence,
\begin{align}
|\mathcal{X}^{(0,q-1)}|&=|D_{t-1}(\mathbf{x}_{[4,n]})\cap D_{t-1}(\mathbf{y}_{[4,n]})|\overset{(a)}{\leq} N_q(n-3,1,t-1)=D_q(n-3,t-1)-D_q(n-4,t-1)+D_q(n-5,t-2),\nonumber\\
|\mathcal{X}^{(0,1)}|&=|D_{t-2-\ell_{(0,1)}}(\mathbf{x}_{[5+\ell_{(0,1)},n]})\cap D_{t-\ell_{(0,1)}^*}(\mathbf{y}_{[3+\ell_{(0,1)}^*,n]})|,\nonumber\\
|\mathcal{X}^{(0,j)}|&=|D_{t-2-\ell_{(0,j)}}(\mathbf{x}_{[5+\ell_{(0,j)},n]})\cap D_{t-2-\ell_{(0,j)}^*}(\mathbf{y}_{[5+\ell_{(0,j)}^*,n]})|,\qquad \text{for $j\in [0,q-2]\setminus\{1\}$},\nonumber\\
|\mathcal{X}^{(1,q-1)}|&=|D_{t-1}(\mathbf{x}_{[4,n]})\cap D_{t-1}(\mathbf{y}_{[4,n]})|\overset{(b)}{\leq} N_q(n-3,1,t-1)=D_q(n-3,t-1)-D_q(n-4,t-1)+D_q(n-5,t-2),\nonumber\\
|\mathcal{X}^{(1,0)}|&=|D_{t-\ell_{(1,0)}}(\mathbf{x}_{[3+\ell_{(1,0)},n]})\cap D_{t-2-\ell_{(1,0)}^*}(\mathbf{y}_{[5+\ell_{(1,0)}^*,n]})|,\nonumber\\
|\mathcal{X}^{(1,k)}|&=|D_{t-2-\ell_{(1,k)}}(\mathbf{x}_{[5+\ell_{(1,k)},n]})\cap D_{t-2-\ell_{(1,k)}^*}(\mathbf{y}_{[5+\ell_{(1,k)}^*,n]})|,\qquad \text{for $k\in [q-2]$},\nonumber\\
|\mathcal{X}^{q-1}|&=|D_{t-2}(\mathbf{x}_{[4,n]})\cap D_{t-2}(\mathbf{y}_{[4,n]})|\overset{(c)}{\leq} N_q(n-3,1,t-2)=D_q(n-3,t-2)-D_q(n-4,t-2)+D_q(n-5,t-3),\nonumber\\
|\mathcal{X}^i|& \le \min\bigl\{ |D_{t-1-\ell_i}(\mathbf{x}_{[3+\ell_i,\,n]})|,\; |D_{t-1-\ell_i^*}(\mathbf{y}_{[3+\ell_i^*,\,n]})| \bigr\}\leq D_q(n-4,t-3),\qquad \text{for $i\in [2,q-2]$}\nonumber%\label{Upp:eq8}
\end{align}
where $\ell_i, \ell_i^{*}\geq 2$ for $i\in [2,q-2]$, and $\ell_{(0,j)},\ell_{(0,j)}^*\geq 0$ for $j\in [0,q-2]$, and $\ell_{(1,k)},\ell_{(1,k)}^*\geq 0$ for $k\in [0,q-2]$. Here, $\ell_{(0,1)}^*=\ell_{(1,0)}=0$, $\ell_{(0,j)}=\ell_{(1,j)}, \ell_{(0,k)}^*=\ell_{(1,k)}^*$ for any $j\in [q-2]$ and $k\in [0,q-2]\setminus\{1\}$. Moreover, $(a)$, $(b)$, and $(c)$ follow from $\mathbf{x}_{[4,n]}\neq \mathbf{y}_{[4,n]}$.

Thus, for sufficiently large $n$, we have
\begin{align}
|\mathcal{X}|&=\sum\limits_{i=0}^{q-1}|\mathcal{X}^{(0,i)}|+\sum\limits_{i=0}^{q-1}|\mathcal{X}^{(1,i)}|+\sum\limits_{i=2}^{q-1}|\mathcal{X}^i|\nonumber\\
&\leq \sum\limits_{i=0}^{q-2}(|\mathcal{X}^{(0,i)}|+|\mathcal{X}^{(1,i)}|)+|\mathcal{X}^{(0,q-1)}|+|\mathcal{X}^{(1,q-1)}|+2D_q(n-5,t-3)+(q-3)D_q(n-4,t-3)\nonumber\\
&~~~+\sum_{i=2}^{q-1} D_q(n-4-i,\; t-2-i)
\overset{(a)}{=}A_1+A_2+B_1+B_2+O(n^{t-3}),\label{Upp:eq9}
\end{align}
where $(a)$ follows from Eq. \eqref{intr:eq2}, and we set $A_1=\sum\limits_{i=0}^{q-2}|\mathcal{X}^{(0,i)}|$, $A_2=\sum\limits_{i=0}^{q-2}|\mathcal{X}^{(1,i)}|$, $B_1=|\mathcal{X}^{(0,q-1)}|$, $B_2=|\mathcal{X}^{(1,q-1)}|$. Since the indices $\ell_{(0,i)}$ (respectively, $\ell_{(1,i)}^*$) for $i\in[0,q-2]$ are each pairwise distinct nonnegative integers, we have
\begin{align}
A_1&=\sum\limits_{i=0}^{q-2}|\mathcal{X}^{(0,i)}|\leq \sum\limits_{i=0}^{q-2}|D_{t-2-\ell_{(0,i)}}(\mathbf{x}_{[5+\ell_{(0,i)},n]})|\leq \sum\limits_{i=0}^{q-2}D_q(n-4-i,t-2-i)\overset{(a)}{=}\frac{1}{(t-2)!}n^{t-2}+O(n^{t-3}),\label{Upp:eq10}\\
A_2&=\sum\limits_{i=0}^{q-2}|\mathcal{X}^{(1,i)}|\leq \sum\limits_{i=0}^{q-2}|D_{t-2-\ell_{(1,i)}^*}(\mathbf{y}_{[5+\ell_{(1,i)}^*,n]})|\leq \sum\limits_{i=0}^{q-2}D_q(n-4-i,t-2-i)\overset{(b)}{=}\frac{1}{(t-2)!}n^{t-2}+O(n^{t-3}),\label{Upp:eq11}
\end{align}
where $(a)$ and $(b)$ follow from Eq. \eqref{intr:eq2}. Moreover, 
\begin{align}
B_1=|\mathcal{X}^{(0,q-1)}|&\leq D_q(n-3,t-1)-D_q(n-4,t-1)+D_q(n-5,t-2)\overset{(a)}{=} 2D_q(n-5,t-2)+\sum\limits_{i=3}^{q}D_q(n-3-i,t-i)\nonumber\\
&\overset{(b)}{=}\frac{2}{(t-2)!}n^{t-2}+O(n^{t-3}),\label{Upp:eq12}\\
B_2=|\mathcal{X}^{(1,q-1)}|&\leq D_q(n-3,t-1)-D_q(n-4,t-1)+D_q(n-5,t-2)\overset{(c)}{=} 2D_q(n-5,t-2)+\sum\limits_{i=3}^{q}D_q(n-3-i,t-i)\nonumber\\
&\overset{(d)}{=}\frac{2}{(t-2)!}n^{t-2}+O(n^{t-3}),\label{Upp:eq13}
\end{align}
where $(a)$ and $(c)$ follow from Eq. \eqref{def:eq1}, and $(b)$ and $(d)$ follow from Eq. \eqref{intr:eq2}.

In this case, based on the values of \(\ell_{(0,0)},\ell_{(0,1)},\ell_{(1,0)}^*\), and \(\ell_{(1,1)}^*\), we discuss the value of \(A_1+A_2+B_1+B_2\) in the following three cases:
\textbf{Case A2-1:} \(\ell_{(0,1)}\geq 1\) or \(\ell_{(1,0)}^*\geq 1\);
\textbf{Case A2-2:} \(\ell_{(0,1)}=\ell_{(1,0)}^*=0,\) and either \(\ell_{(0,0)}\geq 2\) or \(\ell_{(1,1)}^*\geq 2\);
\textbf{Case A2-3:} \(\ell_{(0,1)}=\ell_{(1,0)}^*=0\) and \(\ell_{(0,0)}=\ell_{(1,1)}^*=1\).

\textbf{Case A2-1:} \(\ell_{(0,1)}\geq 1\) or \(\ell_{(1,0)}^*\geq 1\). Without loss of generality, we assume $\ell_{(0,1)} \ge 1$. Note that $\{\ell_{(0,i)} \mid i\in[0,q-2]\}$ is a set of pairwise distinct nonnegative integers, so there exists at most one index $a\in[0,q-2]\setminus\{1\}$ such that $\ell_{(0,a)} = 0$.
It follows that $x_4 = a$, and 
\[
\big|\mathcal{X}^{(0,a)}\big|
= \Big|D_{t-2-\ell_{(0,a)}}\big(\mathbf{x}_{[5+\ell_{(0,a)},n]}\big)
\cap D_{t-2-\ell_{(0,a)}^*}\big(\mathbf{y}_{[5+\ell_{(0,a)}^*,n]}\big)\Big|.
\]
For all $i\in[0,q-2]\setminus\{a\}$, we have $\ell_{(0,i)} \ge 1$, hence
\[
\sum_{i\in[0,q-2]\setminus\{a\}} \big|\mathcal{X}^{(0,i)}\big|
\le \sum_{i\in[0,q-2]\setminus\{a\}} \big|D_{t-2-\ell_{(0,i)}}\big(\mathbf{x}_{[5+\ell_{(0,i)},n]}\big)\big|
\le (q-2)\,D_q(n-5,t-3).
\]

We further consider two subcases according to the value of $\ell^*_{(0,a)}$.

\begin{itemize}
  \item If $\ell_{(0,a)}^* = 0$, then $y_4 = a = x_4$. Since $d_L(\mathbf{x},\mathbf{y})\ge 2$, we have $\mathbf{x}_{[5,n]} \ne \mathbf{y}_{[5,n]}$. It follows that
\[
\big|\mathcal{X}^{(0,a)}\big|
= \Big|D_{t-2}(\mathbf{x}_{[5,n]}) \cap D_{t-2}(\mathbf{y}_{[5,n]})\Big|
\le N_q(n-4,1,t-2)
= D_q(n-4,t-2) - D_q(n-5,t-2) + D_q(n-6,t-3).
\]
  \item If $\ell_{(0,a)}^* \ge 1$, then
\[
\big|\mathcal{X}^{(0,a)}\big|
\le \big|D_{t-2-\ell_{(0,a)}^*}\big(\mathbf{y}_{[5+\ell_{(0,a)}^*,n]}\big)\big|
\le D_q(n-5,t-3)
\overset{(a)}{\le} D_q(n-4,t-2) - D_q(n-5,t-2) + D_q(n-6,t-3),
\]
where $(a)$ follows from Eq. (\ref{def:eq1}).
\end{itemize}

Combining these bounds, we obtain
\begin{align}
A_1&=\sum\limits_{i=0}^{q-2}|\mathcal{X}^{(0,i)}|\leq  D_q(n-4,t-2)-D_q(n-5,t-2)+D_q(n-6,t-3)+(q-2)D_q(n-5,t-3)\nonumber\\
&\overset{(a)}{=}2D_q(n-6,t-3)+\sum\limits_{i=3}^{q}D_q(n-4-i,t-i-1)+(q-2)D_q(n-5,t-3)\overset{(b)}{=}O(n^{t-3}),\label{Upp:eq14}
\end{align}
where $(a)$ follows from Eq. \eqref{def:eq1}, and $(b)$ follows from Eq. \eqref{intr:eq2}.

For sufficiently large \(n\), by Eqs. \eqref{Upp:eq9}, \eqref{Upp:eq11}, \eqref{Upp:eq12}, \eqref{Upp:eq13}, and \eqref{Upp:eq14}, we have 
\begin{align}
|\mathcal{X}|\leq\frac{5}{(t-2)!}n^{t-2}+O(n^{t-3})\overset{(a)}{<}N_q^{(0)}(n,t),\nonumber
\end{align}
where $(a)$ follows from Lemma \ref{Low:lm3}.

\textbf{Case A2-2:} \(\ell_{(0,1)}=\ell_{(1,0)}^*=0,\) and \(\ell_{(0,0)}\geq 2\) or \(\ell_{(1,1)}^*\geq 2\).

Without loss of generality, we assume $\ell_{(0,0)}\geq 2$, which implies $x_5\neq 0$. For simplicity, let $x_5=a$. Since \(\ell_{(0,1)}=\ell_{(1,0)}^*=0,\) we have $\mathbf{x}_{[1,4]}=(1,0,q-1,1)$ and $\mathbf{y}_{[1,4]}=(0,1,q-1,0)$. We decompose $\mathcal{X}^{(0,q-1)}$ as $|\mathcal{X}^{(0,q-1)}|=\sum\limits_{i=0}^{q-1}|\mathcal{X}^{(0,q-1,i)}|$, with
\begin{align*}
|\mathcal{X}^{(0,q-1,0)}|&=|D_{t-2-\ell_{(0,q-1,0)}}(\mathbf{x}_{[6+\ell_{(0,q-1,0)},n]})\cap D_{t-1}(\mathbf{y}_{[5,n]})|\leq D_q(n-6,t-3), \\
|\mathcal{X}^{(0,q-1,1)}|&=|D_{t-1}(\mathbf{x}_{[5,n]})\cap D_{t-2-\ell_{(0,q-1,1)}^*}(\mathbf{y}_{[6+\ell_{(0,q-1,1)}^*,n]})|, \\
|\mathcal{X}^{(0,q-1,a)}|&=|D_{t-2}(\mathbf{x}_{[6,n]})\cap D_{t-2-\ell_{(0,q-1,a)}^*}(\mathbf{y}_{[6+\ell_{(0,q-1,a)}^*,n]})|, \\
|\mathcal{X}^{(0,q-1,i)}|&=|D_{t-2-\ell_{(0,q-1,i)}}(\mathbf{x}_{[6+\ell_{(0,q-1,i)},n]})\cap D_{t-2-\ell_{(0,q-1,i)}^*}(\mathbf{y}_{[6+\ell_{(0,q-1,i)}^*,n]})|\leq D_q(n-6,t-3), ~~\text{for \(i\in[2,q-1]\setminus\{a\}\),}
\end{align*}
where \(\ell_{(0,q-1,j)}\ge 1\) for \(j\in[0,q-1]\setminus\{1,a\}\) and \(\ell_{(0,q-1,k)}^*\ge 0\) for \(k\in[q-1]\).
Since $\ell_{(0,q-1,1)}^*$ and $\ell_{(0,q-1,a)}^*$ are distinct, 
\begin{align*}
|\mathcal{X}^{(0,q-1,1)}|+|\mathcal{X}^{(0,q-1,a)}|\leq D_q(n-5,t-2)+D_q(n-6,t-3).
\end{align*}
Combining these bounds and applying Eq. \eqref{intr:eq2}, we conclude for sufficiently large $n$ that
\begin{align}\label{Upp:eq15}
B_1= |\mathcal{X}^{(0,q-1)}|\leq \frac{1}{(t-2)!}n^{t-2}+O(n^{t-3}).
\end{align}

For sufficiently large $n$,  by Eqs. \eqref{Upp:eq9}, \eqref{Upp:eq10}, \eqref{Upp:eq11}, \eqref{Upp:eq13}, and \eqref{Upp:eq15}, we obtain
\begin{align*}
|\mathcal{X}|\leq A_1+A_2+B_1+B_2+O(n^{t-3})\leq \frac{5}{(t-2)!}n^{t-2}+O(n^{t-3})\overset{(a)}{<} N_q^{(0)}(n,t),
\end{align*}
where $(a)$ follows from Lemma \ref{Low:lm3}.

\textbf{Case A2-3:} \(\ell_{(0,1)}=\ell_{(1,0)}^*=0\) and \(\ell_{(0,0)}=\ell_{(1,1)}^*=1\).

In this case, we have $\mathbf{x}_{[1,5]} = (1,0,q-1,1,0)$ and $\mathbf{y}_{[1,5]} = (0,1,q-1,0,1)$. Then  $\mathbf{x}_{[1,5]}$ and $\mathbf{y}_{[1,5]}$ are precisely the prefixes of $\mathbf{x}^{(0)}$ and $\mathbf{y}^{(0)}$ from Eq. \eqref{Low:eq1}. That is, $\mathbf{x}^{(0)}=\mathbf{x}_{[1,5]}\circ \mathbf{w}^{(0)}$ and $\mathbf{y}^{(0)}=\mathbf{y}_{[1,5]}\circ \mathbf{w}^{(0)}$, with $\mathbf{w}^{(0)} = (w_1,\dots,w_{n-5})$ satisfying $w_j \equiv j+1 \pmod q$ for $j \in [n-5]$. For simplicity, denote $\mathbf{u}=(u_1,\cdots,u_{n-5}) = \mathbf{x}_{[6,n]}$ and $\mathbf{w}=(w'_1,\cdots,w'_{n-5}) = \mathbf{y}_{[6,n]}$. Since $d_L(\mathbf{x}_{[1,5]},\mathbf{y}_{[1,5]}) \ge 2$, we always have $d_L(\mathbf{x},\mathbf{y}) \ge 2$ regardless of whether $\mathbf{u}=\mathbf{w}$.
By Lemmas \ref{Pro:lm1} and \ref{Pro:lm2}, we have
\begin{align*}
|\mathcal{X}^{(0,q-1)}|&=|D_{t-1}(1,0,\mathbf{u})\cap D_{t-1}(0,1,\mathbf{w})|\overset{(a)}{\leq} N_q(n-3,1,t-1)=D_q(n-3,t-1)-D_q(n-4,t-1)+D_q(n-5,t-2),\\
|\mathcal{X}^{(0,1)}|&=|D_{t-2}(0,\mathbf{u})\cap D_{t}(q-1,0,1,\mathbf{w})|\leq |D_{t-2}(0,\mathbf{u})|\leq D_q(n-4,t-2)\\
|\mathcal{X}^{(0,0)}|&=|D_{t-3}(\mathbf{u})\cap D_{t-2}(1,\mathbf{w})|\leq |D_{t-3}(\mathbf{u})|\leq D_q(n-5,t-3)\\
|\mathcal{X}^{(0,j)}|&=|D_{t-2-\ell_{(0,j)}}(\mathbf{x}_{[5+\ell_{(0,j)},n]})\cap D_{t-2-\ell_{(0,j)}^*}(\mathbf{y}_{[5+\ell_{(0,j)}^*,n]})|,\qquad \text{for $j\in [2,q-2]$}\\
|\mathcal{X}^{(1,q-1)}|&=|D_{t-1}(1,0,\mathbf{u})\cap D_{t-1}(0,1,\mathbf{w})|\overset{(b)}{\leq} N_q(n-3,1,t-1)=D_q(n-3,t-1)-D_q(n-4,t-1)+D_q(n-5,t-2),\\
|\mathcal{X}^{(1,1)}|&=|D_{t-2}(0,\mathbf{u})\cap D_{t-3}(\mathbf{w})|\leq |D_{t-3}(\mathbf{w})|\leq D_q(n-5,t-3)\\
|\mathcal{X}^{(1,0)}|&=|D_{t}(q-1,1,0,\mathbf{u})\cap D_{t-2}(1,\mathbf{w})|\leq |D_{t-2}(1,\mathbf{w})|\leq D_q(n-4,t-2)\\
|\mathcal{X}^{(1,j)}|&=|D_{t-2-\ell_{(1,j)}}(\mathbf{x}_{[5+\ell_{(1,j)},n]})\cap D_{t-2-\ell_{(1,j)}^*}(\mathbf{y}_{[5+\ell_{(1,j)}^*,n]})|,\qquad \text{for $j\in [2,q-2]$}\\
|\mathcal{X}^{q-1}|&=|D_{t-2}(1,0,\mathbf{u})\cap D_{t-2}(0,1,\mathbf{w})|\overset{(c)}{\leq} N_q(n-3,1,t-2)=D_q(n-3,t-2)-D_q(n-4,t-2)+D_q(n-5,t-3),\\
|\mathcal{X}^i|& \le \min\bigl\{ |D_{t-1-\ell_i}(\mathbf{x}_{[3+\ell_i,\,n]})|,\; |D_{t-1-\ell_i^*}(\mathbf{y}_{[3+\ell_i^*,\,n]})| \bigr\},\qquad \text{for $i\in [2,q-2]$}
\end{align*}
where $\ell_i, \ell_i^{*}\geq 4$ for $i\in [2,q-2]$, and $\ell_{(0,j)},\ell_{(0,j)}^*,\ell_{(1,j)},\ell_{(1,j)}^*\geq 2$ for $j\in [2,q-2]$.  Moreover, the indices $\ell_i$ are pairwise distinct, and the same holds for $\ell_i^*$, $\ell_{(0,j)}$, $\ell_{(0,j)}^*$, $\ell_{(1,j)}$, and $\ell_{(1,j)}^*$ within their respective index ranges $i,j \in [2,q-2]$. Here, $(a)$, $(b)$, and $(c)$ follow from $(1,0,\mathbf{u})\neq (0,1,\mathbf{w})$ and Theorem \ref{def:thm1}.

Since the indices $\ell_i\geq 4$ and $\ell_{(0,i)}, \ell_{(1,i)}\geq 2$ for $i\in [2,q-2]$ are pairwise distinct, we have
\begin{align}
\sum\limits_{i=2}^{q-2}|\mathcal{X}^{i}|&\leq \sum\limits_{i=2}^{q-2}|D_{t-1-\ell_i}(\mathbf{x}_{[3+\ell_i,n]})|\leq \sum\limits_{i=3}^{q-1}D_q(n-3-i,t-2-i),\label{Upp:eq16}\\
\sum\limits_{i=2}^{q-2}|\mathcal{X}^{(0,i)}|&\leq \sum\limits_{i=2}^{q-2}|D_{t-2-\ell_{(0,i)}}(\mathbf{x}_{[5+\ell_{(0,i)},n]})|\leq \sum\limits_{i=3}^{q-1}D_q(n-3-i,t-1-i),\label{Upp:eq17}\\
\sum\limits_{i=2}^{q-2}|\mathcal{X}^{(1,i)}|&\leq \sum\limits_{i=2}^{q-2}|D_{t-2-\ell_{(1,i)}}(\mathbf{x}_{[5+\ell_{(1,i)},n]})|\leq \sum\limits_{i=3}^{q-1}D_q(n-3-i,t-1-i).\nonumber
\end{align}
Thus, for $j\in \{0,1\}$, we have 
\begin{align}\label{Upp:eq18}
|\mathcal{X}^{j}|\leq D_q(n-2,t-1)+D_q(n-5,t-2)-D_q(n-4,t-1).
\end{align}

Note that, in the proof of Lemma \ref{Low:lm2}, for the decomposition of  \(D_t(\mathbf{x}^{(0)})\cap D_t(\mathbf{y}^{(0)})\) (which satisfies $|D_t(\mathbf{x}^{(0)})\cap D_t(\mathbf{y}^{(0)})|=N_q^{(0)}(n,t)$), all components match their maximal values except for \(|(D_t(\mathbf{x}^{(0)})\cap D_t(\mathbf{y}^{(0)}))^{(0,1)}|=D_q(n-4,t-2)-D_q(n-4-q,t-1-q)+D_q(n-5-q,t-2-q)\).

We next analyze $|\mathcal{X}|$ in the following two subcases, depending on whether $\mathbf{u}=\mathbf{w}$.
\textbf{Case A2-3a:}  $\mathbf{u}=\mathbf{w}$;
\textbf{Case A2-3b:}  $\mathbf{u}\neq \mathbf{w}$.  

\begin{itemize}
  \item \textbf{Case A2-3a:}  $\mathbf{u}=\mathbf{w}$. Then
\begin{align}
|\mathcal{X}^{(0,1)}|+|\mathcal{X}^{(1,0)}|&\leq |D_{t-2}(0,\mathbf{u})|+|D_{t-2}(1,\mathbf{u})|\nonumber\\
&\overset{(a)}{\leq} 2D_q(n-4,t-2)+D_q(n-q-5,t-q-2)-D_q(n-q-4,t-q-1)\nonumber\\
&\overset{(b)}{=}2D_q(n-4,t-2)-n^{t-1-q} + O(n^{t-2-q}),\label{Upp:eq19}
\end{align}
where $(a)$ follows from Lemma \ref{Upp:lm4}, and $(b)$ follows from Eq. \eqref{intr:eq2}. This upper bound from $(a)$ matches exactly \(|(D_t(\mathbf{x}^{(0)})\cap D_t(\mathbf{y}^{(0)}))^{(0,1)}|+|(D_t(\mathbf{x}^{(0)})\cap D_t(\mathbf{y}^{(0)}))^{(1,0)}|\). Thus, we obtain
\begin{align*}
|\mathcal{X}| &=\sum\limits_{i=2}^{q-2}\big(|\mathcal{X}^i|+|\mathcal{X}^{(0,i)}|+|\mathcal{X}^{(1,i)}|\big)+|\mathcal{X}^{q-1}|+|\mathcal{X}^{(0,q-1)}|+|\mathcal{X}^{(1,q-1)}|+\big(|\mathcal{X}^{(0,0)}|+|\mathcal{X}^{(1,1)}|\big)+\big(|\mathcal{X}^{(0,1)}|+|\mathcal{X}^{(1,0)}|\big)\\
&\leq N_q^{(0)}(n,t).
\end{align*}
Although the bound $|\mathcal{X}|\le N_q^{(0)}(n,t)$ is already established, we need to characterize precisely which choices of $\mathbf{u}$ attain equality.
We further refine the bound by distinguishing the following subcases according to the structure of $\mathbf{u}$.

%Although the upper bound \(|\mathcal{X}|\le N_q^{(0)}(n,t)\) has already been obtained, we need to determine precisely which choices of \(\mathbf{u}\) can attain this bound. To this end, we now distinguish the following subcases according to the structure of \(\mathbf{u}\).
%Moreover, we now distinguish the following subcases according to \(\mathbf{u}\).
\begin{enumerate}
 \item If $\mathbf{u} \notin \mathbf{C}_q(n-5)$, applying Lemma \ref{Upp:lm4} yields the improved bound,
\begin{align}\label{Upp:eq20}
|\mathcal{X}^{(0,1)}| + |\mathcal{X}^{(1,0)}|
\le \big|D_{t-2}(0,\mathbf{u})\big| + \big|D_{t-2}(1,\mathbf{u})\big|
\le 2D_q(n-4,t-2) - 2n^{t-1-q} + O(n^{t-2-q}).
\end{align}
By comparing Eqs. \eqref{Upp:eq19} and \eqref{Upp:eq20}, we obtain
\begin{align}\label{Upp:eq21}
|\mathcal{X}| \le N_q^{(0)}(n,t) - n^{t-1-q} + O(n^{t-2-q}).
\end{align}
  \item  If $\mathbf{u} \in \mathbf{C}_q(n-5)$ and $\{u_i|i\in[q-3]\}\neq [2,q-2]$, then there exists some $j\in [q-2]$ such that $\ell_{(0,j)}\geq q-1$ and $|\mathcal{X}^{(0,j)}|\leq D_q(n-3-q,t-1-q)$. Since the indices $\ell_{(0,i)}\geq 2$ for $i\in [2,q-2]$ are pairwise distinct, we have
\begin{align}
\sum\limits_{i=2}^{q-2}|\mathcal{X}^{(0,i)}|\leq \sum\limits_{i=2}^{q-2}|D_{t-2-\ell_{(0,i)}}(\mathbf{x}_{[5+\ell_{(0,i)},n]})|\leq \sum\limits_{i=3}^{q-2}D_q(n-3-i,t-1-i)+D_q(n-3-q,t-1-q).\label{Upp:eq22}
\end{align}
Thus,  by comparing Eqs. \eqref{Upp:eq17} and \eqref{Upp:eq22}, we obtain
\begin{align}\label{Upp:eq23}
|\mathcal{X}| \le N_q^{(0)}(n,t) - n^{t-q} + O(n^{t-1-q}).
\end{align}
 \item If $\mathbf{u} \in \mathbf{C}_q(n-5)$, $\{u_i|i\in[q-3]\}=[2,q-2]$, and $u_{q-2}\neq q-1$, then we consider the value $|\mathcal{X}^{(0,q-1)}|$. From the preceding discussion,
 \begin{align}
 |\mathcal{X}^{(0,q-1)}|&=|D_{t-1}(1,0,\mathbf{u})\cap D_{t-1}(0,1,\mathbf{u})|\leq D_q(n-3,t-1)-D_q(n-4,t-1)+D_q(n-5,t-2)\nonumber\\
&\overset{(a)}{=}2D_q(n-5,t-2)+\sum\limits_{i=2}^{q-1}D_q(n-4-i,t-1-i),\label{Upp:eq24}
 \end{align}
 where $(a)$ follows from Eq. \eqref{def:eq1}. 
 
 By decomposing $\mathcal{X}^{(0,q-1)}$, we have $|\mathcal{X}^{(0,q-1)}|=\sum\limits_{i=0}^{q-1}|\mathcal{X}^{(0,q-1,i)}|$ with 
\begin{align*}
|\mathcal{X}^{(0,q-1,a)}|&=|D_{t-2}(\mathbf{u})|=D_q(n-5,t-2),\quad |\mathcal{X}^{(0,q-1,i)}|=|D_{t-2-\hat{\ell}_i}(\mathbf{u}_{[\hat{\ell}_i+1,n-5]})|=D_q(n-5-\hat{\ell}_i,t-2-\hat{\ell}_i),
\end{align*}
 where $a\in \{0,1\}$, $i\in [2,q-1]$, and \(\hat{\ell}_i\) denotes the position of the first occurrence of symbol \(i\) in \(\mathbf{u}\). Since $\mathbf{u} \in \mathbf{C}_q(n-5)$, $\{u_i|i\in[q-3]\}=[2,q-2]$, and $u_{q-2}\neq q-1$, we have $\{\hat{\ell}_i|i\in[2,q-2]\}=[1,q-3]$ and $\hat{\ell}_{q-1}> q-2$. Thus, in this case, we have 
 \begin{align}
|\mathcal{X}^{(0,q-1)}|=\sum\limits_{i=0}^{q-1}|\mathcal{X}^{(0,q-1,i)}|\leq 2D_q(n-5,t-2)+\sum\limits_{i=2}^{q-2}D_q(n-4-i,t-1-i)+D_q(n-4-q,t-1-q).\label{Upp:eq25}
\end{align}
Comparing Eqs.~\eqref{Upp:eq24} and~\eqref{Upp:eq25} gives
\begin{align}\label{Upp:eq26}
|\mathcal{X}| \le N_q^{(0)}(n,t) - n^{t-q} + O(n^{t-1-q}).
\end{align} 
\end{enumerate}
 
  \item \textbf{Case A2-3b:}  $\mathbf{u}\neq\mathbf{w}$. Then
\begin{align*}
|\mathcal{X}^{(0,q-1)}|&\leq|D_{t-1}(1,0,\mathbf{u})\cap D_{t-1}(0,1,\mathbf{w})|\overset{(a)}{\leq }2D_q(n-5,t-2)+\sum\limits_{i=2}^{q-1}D_q(n-4-i,t-1-i)\\
&~~+D_q(n-8,t-3)-D_q(n-6,t-3),\\
|\mathcal{X}^{(0,1)}|+|\mathcal{X}^{(1,0)}|&\leq |D_{t-2}(0,\mathbf{u})|+|D_{t-2}(1,\mathbf{w})|\leq  2D_q(n-4,t-2),
\end{align*}
where $(a)$ follows from Lemma \ref{Upp:lm5} for sufficiently large $n$. 

The difference between the sum of these three components $|\mathcal{X}^{(0,q-1)}|$, $|\mathcal{X}^{(0,1)}|$, $|\mathcal{X}^{(1,0)}|$ and the corresponding sum from the decomposition of $D_t(\mathbf{x}^{(0)}) \cap D_t(\mathbf{y}^{(0)})$ is given by
\begin{align*}
&\phantom{=} 2D_q(n-4,t-2)+D_q(n-q-5,t-q-2)-D_q(n-q-4,t-q-1)+D_q(n-3,t-1)-D_q(n-4,t-1)\\
&\quad+D_q(n-5,t-2)-\Bigl(2D_q(n-5,t-2)+\sum_{i=2}^{q-1}D_q(n-4-i,t-1-i)\\
&\qquad+D_q(n-8,t-3)-D_q(n-6,t-3)+2D_q(n-4,t-2)\Bigr)\\
&\overset{(a)}{=}D_q(n-6,t-3)-D_q(n-8,t-3)-D_q(n-q-4,t-q-1)+D_q(n-q-5,t-q-2),
\end{align*}
where $(a)$ follows from Eq. \eqref{def:eq1}. Therefore, we obtain
\begin{align}
|\mathcal{X}| &=\sum\limits_{i=2}^{q-2}\big(|\mathcal{X}^i|+|\mathcal{X}^{(0,i)}|+|\mathcal{X}^{(1,i)}|\big)+|\mathcal{X}^{q-1}|+|\mathcal{X}^{(0,q-1)}|+|\mathcal{X}^{(1,q-1)}|+\big(|\mathcal{X}^{(0,0)}|+|\mathcal{X}^{(1,1)}|\big)+\big(|\mathcal{X}^{(0,1)}|+|\mathcal{X}^{(1,0)}|\big)\nonumber\\
&\leq N_q^{(0)}(n,t)-\big(D_q(n-6,t-3)-D_q(n-8,t-3)-D_q(n-q-4,t-q-1)+D_q(n-q-5,t-q-2)\big)\nonumber\\
&\overset{(a)}{\leq} N_q^{(0)}(n,t)-n^{t-4}+O(n^{t-5}), \label{Upp:eq27}
\end{align}
where $(a)$ follows from Eqs. \eqref{intr:eq2} and \eqref{def:eq1}. 
\end{itemize}

\textbf{Case B:} \(\ell_0 \ge 1\) and \(\ell_1^* \ge 1\). 

Recall that 
\(
|\mathcal{X}| = \sum_{i=0}^{q-1} |\mathcal{X}^i|,
\)
where $\mathcal{X}^0 = 0 \circ \bigl( D_{t-1-\ell_0}(\mathbf{x}_{[3+\ell_0,n]}) \cap D_t(\mathbf{y}_{[2,n]}) \bigr),$ 
$\mathcal{X}^1 = 1 \circ \bigl( D_t(\mathbf{x}_{[2,n]}) \cap D_{t-1-\ell_1^*}(\mathbf{y}_{[3+\ell_1^*,n]}) \bigr),$ and 
\(
\mathcal{X}^i = i \circ \bigl( D_{t-1-\ell_i}(\mathbf{x}_{[3+\ell_i,n]}) \cap D_{t-1-\ell_i^*}(\mathbf{y}_{[3+\ell_i^*,n]}) \bigr) 
\)
for $i = 2,\dots,q-1$. Here $\ell_0, \ell_1^*, \ell_i, \ell_i^*$ are the smallest non-negative integers such that
\(
x_{2+\ell_0}=0, y_{2+\ell_1^*}=1, x_{2+\ell_i}=i, y_{2+\ell_i^*}=i\ \ (2\leq i\leq q-1).
\)

Since \(\ell_0, \ell_1^* \ge 1\), it follows that
\begin{align*}
|\mathcal{X}^0| \le \bigl| D_{t-1-\ell_0}(\mathbf{x}_{[3+\ell_0,\,n]}) \bigr| \le D_q(n-3, t-2),\qquad
|\mathcal{X}^1| \le \bigl| D_{t-1-\ell_1^*}(\mathbf{y}_{[3+\ell_1^*,\,n]}) \bigr| \le D_q(n-3, t-2).
\end{align*}
Since the indices \(\ell_i\) (respectively, \(\ell_i^*\)) are pairwise distinct for \(i\in [2,q-1]\), we now distinguish two subcases based on the values of \(\ell_i\) and \(\ell_i^*\) (with \(\ell_i,\ell_i^*\ge 0\)).
\begin{itemize}
    \item  \emph{Subcase 1:} For every \(i \in [2, q-1]\), either \(\ell_i \ge 1\) or \(\ell_i^* \ge 1\).

Then $|\mathcal{X}^i|\leq D_q(n-3,t-2)$ for all $i\in [2,q-1]$. Since the indices \(\ell_i\) are pairwise distinct, at most two of them can be \(\le 1\) (namely the values \(0\) and \(1\)). Consequently,
\begin{align*}
\sum_{i=2}^{q-1} |\mathcal{X}^i| \le 2D_q(n-3,\; t-2)+\sum_{i=3}^{q-2} D_q(n-1-i,\; t-i).
\end{align*}
    \item \emph{Subcase 2:} There exists some \(i \in [2, q-1]\) such that \(\ell_i = \ell_i^* = 0\).
     
Without loss of generality, let \(i=2\) be this index. Then $\mathbf{x}_{[3,n]}\neq \mathbf{y}_{[3,n]}$ because $d_L(\mathbf{x},\mathbf{y})\geq 2$. Hence,
\begin{align*}
|\mathcal{X}^2|=|D_{t-1}(\mathbf{x}_{[3,n]})\cap D_{t-1}(\mathbf{y}_{[3,n]})|&\leq N_q(n-2,1,t-1)=D_q(n-2,t-1)-D_q(n-3,t-1)+D_q(n-4,t-2)\\
&=2D_q(n-4,t-2)+\sum_{i=3}^{q} D_q(n-2-i,\; t-i).
\end{align*} 
Since \(\ell_2^*=0\), we must have \(\ell_i^* \ge 1\) for all \(i \in [3, q-1]\). Moreover, the indices \(\ell_i^*\) for \(i \in \{1\} \cup [3, q-1]\) are pairwise distinct. Thus,
\begin{align*}
\sum_{i\in \{1\}\cup [3,q-1]} |\mathcal{X}^i| \le D_q(n-3,\; t-2)+\sum_{i=3}^{q-1} D_q(n-1-i,\; t-i).
\end{align*}
\end{itemize}

Combining the two cases, we have
\begin{align}
|\mathcal{X}|=\sum\limits_{i=0}^{q-1}|\mathcal{X}^i|\overset{(a)}{\leq}\frac{4}{(t-2)!}n^{t-2}+O(n^{t-3})\overset{(b)}{<}N_q^{(0)}(n,t),\nonumber
\end{align}
where $(a)$ follows from Eq. \eqref{intr:eq2}, and $(b)$ follows from Lemma \ref{Low:lm3} for sufficiently large $n$.

\textbf{Case C:}  \(\ell_0 = 0\) and \(\ell_1^* \ge 1\).

When $\ell_0=0$ and $\ell_1^*\geq 1$, we have $x_2=0$, $y_2\neq 1$, and $\ell_i\geq 1$ for $i\in [2,q-1]$. Since $d_L(\mathbf{x},\mathbf{y}) \ge 2$, it follows that $\mathbf{x}_{[3,n]}\notin D_1(\mathbf{y}_{[2,n]})$. By Lemma~\ref{Upp:lm3}, we have 
\[
|\mathcal{X}^0| = |D_{t-1}(\mathbf{x}_{[3,n]})\cap D_{t}(\mathbf{y}_{[2,n]})| \leq N_q(n-2,1,t,t-1) = D_q(n-2,t-1)-D_q(n-4,t-1)+D_q(n-5,t-2).
\]

Since $\ell_1^*\geq 1$, we have
\(
|\mathcal{X}^1| \leq |D_{t-1-\ell_1^*}(\mathbf{y}_{[3+\ell_1^*,n]})| \leq D_q(n-3,t-2). 
\)
Since all indices $\ell_i$ for $i\in[2,q-1]$ are pairwise distinct and satisfy $\ell_i\ge 1$, we have 
\[
\sum_{i=2}^{q-1} |\mathcal{X}^i| \le D_q(n-3,\; t-2)+\sum_{i=3}^{q-1} D_q(n-1-i,\; t-i).
\]

Thus,
\begin{align*}
|\mathcal{X}| &= \sum_{i=0}^{q-1} |\mathcal{X}^i| \leq D_q(n-2,t-1)-D_q(n-4,t-1)+D_q(n-5,t-2)+D_q(n-3,t-2)+\sum_{i=2}^{q-1} D_q(n-1-i,\; t-i) \\
&\overset{(a)}{=} 2D_q(n-3,t-2)+D_q(n-4,t-2)+2D_q(n-5,t-2)+O(n^{t-3})  \\
&\overset{(b)}{\leq} \frac{5}{(t-2)!}\,n^{t-2}+O(n^{t-3})\overset{(c)}{<} N_q^{(0)}(n,t),
\end{align*}
where $(a)$ and $(b)$ follow from Eqs.~\eqref{intr:eq2} and \eqref{def:eq1}, and $(c)$ follows from Lemma \ref{Low:lm3} for sufficiently large $n$.

\textbf{Case D:} \(\ell_0 \ge 1\) and \( \ell_1^* = 0\).

When $\ell_1^*=0$ and $\ell_0\geq 1$, we have $y_2=1$, $x_2\neq 0$, and $\ell_2^*\geq 1$. By the same argument as in \textbf{Case C}, we obtain
$$|\mathcal{X}|\leq \frac{5}{(t-2)!}\,n^{t-2}+O(n^{t-3})< N_q^{(0)}(n,t),$$
for sufficiently large $n$.

Combining all four cases, we conclude that for sufficiently large \(n\) and all \(\mathbf{x},\mathbf{y}\in\Sigma_q^n\) satisfying \(x_1\neq y_1\) and \(d_L(\mathbf{x},\mathbf{y})\ge 2\),
\[
|D_t(\mathbf{x})\cap D_t(\mathbf{y})| \le N_q^{(0)}(n,t).
\]
\end{proof}

From the proof of Lemma~\ref{Upp:lm6}, we derive the following characterization of the cardinality $|D_t(\mathbf{x}) \cap D_t(\mathbf{y})|$ for arbitrary sequences $\mathbf{x},\mathbf{y}\in\Sigma_q^n$ with $x_1\neq y_1$ and $d_L(\mathbf{x},\mathbf{y})\ge 2$. In particular, in \emph{Case A1-2} and in \emph{Case A2-3}, the intersection size attains the same leading term as \(N_q^{(0)}(n,t)\). The precise situations are summarized in the following remark.

\begin{remark}
Let $\mathbf{x}=(x_1,\dots,x_n), \mathbf{y}=(y_1,\dots,y_n)\in \Sigma_q^n$ with $d_L(\mathbf{x},\mathbf{y})\geq 2$ and $x_1\neq y_1$.
For convenience, let $x_1=a,y_1=b$ with $a,b\in \Sigma_q$, and define $\mathcal{X}=D_t(\mathbf{x})\cap D_t(\mathbf{y})$.
\begin{enumerate}
\item If $|\mathcal{X}|$ has the same leading term as $N_q^{(0)}(n,t)$, i.e., $|\mathcal{X}| = \frac{6}{(t-2)!}n^{t-2}+O(n^{t-3})$, then, by the results of \emph{Case A1-2} and \emph{Case A2-3}, one of the following must hold: $\mathbf{x}_{[1,4]}=(a,b,a,c)$ and $\mathbf{y}_{[1,4]}=(b,a,a,b)$ with $c\in \Sigma_q$, or $\mathbf{x}_{[1,4]}=(a,b,b,a)$ and $\mathbf{y}_{[1,4]}=(b,a,b,c)$ with $c\in \Sigma_q$, or $\mathbf{x}_{[1,5]}=(a,b,c_1,a,b)$ and $\mathbf{y}_{[1,5]}=(b,a,c_1,b,a)$ with $c_1\in \Sigma_q\setminus\{a,b\}$.
\item For $t\geq 3$, if there are no two sequences $\mathbf{x},\mathbf{y}\in \Sigma_q^n$ such that $\mathbf{x}_{[1,5]}=(a,b,c,a,b)$ and $\mathbf{y}_{[1,5]}=(b,a,c,b,a)$ for any $c\in \Sigma_q\setminus\{a,b\}$, then, using Eq. \eqref{Upp:eq7} and item 1 above, $|\mathcal{X}|\leq  \frac{6}{(t-2)!}n^{t-2}-\frac{3t+15}{(t-3)!}n^{t-3} +O(n^{t-4})$.
\item For $t\geq 3$, if $|\mathcal{X}|$ matches the first two  terms of $N_q^{(0)}(n,t)$, i.e., $|\mathcal{X}| = \frac{6}{(t-2)!}n^{t-2}-\frac{3t+13}{(t-3)!}n^{t-3} +O(n^{t-4})$, then, by the results of \emph{Case A2-3}, we have  $\mathbf{x}_{[1,5]}=(a,b,c,a,b)$ and $\mathbf{y}_{[1,5]}=(b,a,c,b,a)$ for some $c\in \Sigma_q\setminus\{a,b\}$.
\item For $t\geq 4$, if $|\mathcal{X}|$ matches the first three  terms of $N_q^{(0)}(n,t)$, then, by \emph{Case A2-3b}, $\mathbf{x}_{[1,5]}=(a,b,c,a,b)$, $\mathbf{y}_{[1,5]}=(b,a,c,b,a)$, and $\mathbf{x}_{[6,n]}=\mathbf{y}_{[6,n]}$ for some $c\in \Sigma_q\setminus\{a,b\}$.
\item For $t\geq q+1$, $|\mathcal{X}|$ matches the first $q$  terms of $N_q^{(0)}(n,t)$ if and only if $\{\mathbf{x}_{[1,5]},\mathbf{y}_{[1,5]}\}=\{(c_q,c_{q-1},c_{q-2},c_q,c_{q-1}),\\(c_{q-1},c_{q},c_{q-2},c_{q-1},c_{q})\}$, and $\mathbf{x}_{[6,n]}=\mathbf{y}_{[6,n]}=\mathbf{c}_q(n-5,\sigma)$, where $\sigma$ is an ordering of $\Sigma_q$ of the form $\sigma= (c_1,\dots,c_{q-2},\\c_{q-1},c_q)$. This characterization follows from \emph{Case A2-3a} and Lemma \ref{Low:lm3}.
\end{enumerate}
\label{Upp:rmk1}
\end{remark}

We now state our main result. Recall that \(N_q(n,t)\) is defined as \(N_q^{(q-3)}(n,t)\) in Lemma~\ref{Low:lm4}. When $q=3$, \(N_q(n,t)=N_q^{(0)}(n,t)\). Moreover, for $i\in [q-3]$ and $q\geq 4$,
\begin{align}
N_q^{(i)}(n,t)=\sum\limits_{j=0}^{i-1}N_q^{(j)}(n-i+j,t-i+j+1)+N_q^{(0)}(n-i,t-i)-\sum\limits_{j=1}^{i}D_q(n-3-i-q+j,t-2-i-q+j).\label{Upp:eq28}
\end{align}  
Let
\begin{align}\label{Upp:eq29}
\mathcal{P}=\left\{ \{\pi(\mathbf{x}^{(q-3)}), \pi(\mathbf{y}^{(q-3)})\} : \pi\in S_q \right\}
\cup
\left\{ \{(\pi(\mathbf{x}^{(q-3)}))^R, (\pi(\mathbf{y}^{(q-3)}))^R\} : \pi\in S_q \right\},
\end{align}
where
\(\mathbf{x}^{(q-3)} = (2,\dots,q-2,0,1,q-1,0,1,\mathbf{w}^{(q-3)})\) and \(\mathbf{y}^{(q-3)} = (2,\dots,q-2,1,0,q-1,1,0,\mathbf{w}^{(q-3)})\)
with $\mathbf{w}^{(q-3)} = \mathbf{c}_q(n-q-2,\sigma)$ and $\sigma=(2,3,\dots,q-1,0,1)$.

\begin{theorem}
Let $q \geq 3$ and $t \geq 2$, and let $\mathcal{P}$ be defined as in Eq.~\eqref{Upp:eq29}. For sufficiently large $n$,  $$N_q(n,2,t)=N_q(n,t).$$ 
Furthermore, suppose $t\ge q+2$ and let $\mathbf{x},\mathbf{y}\in\Sigma_q^n$ with $d_L(\mathbf{x},\mathbf{y})\ge2$. Then, for sufficiently large $n$,
\(|D_t(\mathbf{x})\cap D_t(\mathbf{y})|=N_q(n,t)\) if and only if $\{\mathbf{x},\mathbf{y}\}\in\mathcal{P}$. Moreover, $|\mathcal{P}|=2(q!)$.
\label{Upp:thm1}
\end{theorem}

\begin{proof}
Let $\mathbf{x}=(\mathbf{u},\mathbf{v},\mathbf{w}), \mathbf{y}=(\mathbf{u},\mathbf{v}',\mathbf{w})\in \Sigma_q^n$, where $\mathbf{v}=(v_1,\dots,v_l)$, $\mathbf{v}'=(v_1',\dots,v_l')$, and define $\mathcal{X}=D_t(\mathbf{x})\cap D_t(\mathbf{y})$. Here $|\mathbf{u}|=p$, $|\mathbf{w}|=r$, $|\mathbf{v}|=|\mathbf{v}'|=l$, $v_1\neq v_1'$, $v_l\neq v_l'$, and $d_L(\mathbf{x},\mathbf{y})\geq 2$. A key observation from Lemma~\ref{Pro:lm4} is that analyzing the sequences from their left endpoints is equivalent to analyzing them from their right endpoints. If $p=0$ or $r=0$, then by Lemmas \ref{Low:lm4} and \ref{Upp:lm6}, we have $|\mathcal{X}|\leq N_q^{(0)}(n,t)\leq N_q(n,t)$ for all sufficiently large $n$.

We henceforth restrict our attention to the case $p,r\geq 1$. Without loss of generality, assume $p\leq r$. Since $p+r+l=n$, $n-p$ remains sufficiently large when $n$ is sufficiently large. For notational clarity, we write $\mathbf{x}=(x_1,x_2,\dots,x_n)$ and $\mathbf{y}=(y_1,y_2,\dots,y_n)$. We prove $|\mathcal{X}|\leq N_q(n,t)$ by induction on the length $p$ of the common prefix $\mathbf{u}$. 

We divide the proof into three cases based on the value of $p$: \textbf{Case A:} $1\leq p\leq q-3$; \textbf{Case B:} $p=q-2$ or $p=q-1$; \textbf{Case C:} $p\geq q$. Here, when $q=3$, we only need to consider \textbf{Cases B} and \textbf{C}.

For convenience, define the truncated sequences
\(
\hat{\mathbf{x}} = (\mathbf{v}, \mathbf{w})\) and \( \hat{\mathbf{y}} = (\mathbf{v}', \mathbf{w})
\) 
in \(\Sigma_q^{n-p}\). Let $\hat{\mathcal{X}}=D_{t-p}(\mathbf{v}, \mathbf{w})\cap D_{t-p}(\mathbf{v}', \mathbf{w})$. Note that $d_L(\hat{\mathbf{x}}, \hat{\mathbf{y}}) = d_L(\mathbf{x}, \mathbf{y}) \geq 2$ since the common prefix $\mathbf{u}$ does not affect the Levenshtein distance. By Lemma \ref{Pro:lm3}, we have
\[
|D_t(\mathbf{x}) \cap D_t(\mathbf{y})| \leq \sum_{s=0}^t |D_s(\mathbf{u})| \cdot |D_{t-s}(\hat{\mathbf{x}}) \cap D_{t-s}(\hat{\mathbf{y}})|.
\]

Consider the case where $p \leq q$. Since $|\mathbf{u}|=p$, we have $|D_s(\mathbf{u})| = O(1)$ for all $s \geq 0$. By Theorem \ref{def:thm5}, we have $|D_{t-s}(\hat{\mathbf{x}}) \cap D_{t-s}(\hat{\mathbf{y}})| = O(n^{t-s-2})$ for any $s \in [0, t]$. Thus,
\[
|D_t(\mathbf{x}) \cap D_t(\mathbf{y})| \leq |D_1(\mathbf{u})| \cdot |D_{t-1}(\hat{\mathbf{x}}) \cap D_{t-1}(\hat{\mathbf{y}})| + |D_t(\hat{\mathbf{x}}) \cap D_t(\hat{\mathbf{y}})| + O(n^{t-4}).
\]
Now consider the case where $\hat{\mathbf{x}}_{[1,5]}$ and $\hat{\mathbf{y}}_{[1,5]}$ are not of the form  $(a,b,c,a,b)$ and $(b,a,c,b,a)$, respectively, for any distinct $a,b,c\in\Sigma_q$. By Remark~\ref{Upp:rmk1}, for $p\leq q$, we obtain 
\begin{align}
|\mathcal{X}|
&\leq p \cdot \frac{6}{(t-3)!} (n-p)^{t-3} + \frac{6}{(t-2)!} (n-p)^{t-2} - \frac{3t+15}{(t-3)!} (n-p)^{t-3} + O(n^{t-4})\nonumber \\
&= \frac{6}{(t-2)!} n^{t-2} - \frac{3t+15}{(t-3)!} n^{t-3} + O(n^{t-4}) < N_q^{(0)}(n, t),\label{Upp:eq30}
\end{align}
since $N_q^{(0)}(n, t) = \frac{6}{(t-2)!} n^{t-2} - \frac{3t+13}{(t-3)!} n^{t-3} + O(n^{t-4})$. 

Therefore, in both \textbf{Cases A} and \textbf{B}, it suffices to consider the scenario where $\hat{\mathbf{x}}_{[1,5]} = (a, b, c, a, b)$ and $\hat{\mathbf{y}}_{[1,5]} = (b, a, c, b, a)$ for some distinct $a, b, c \in \Sigma_q$. By Lemma \ref{Pro:lm4}, for notational simplicity, we set $a=0$, $b=1$, and $c=q-1$. Finally, by Lemma \ref{Pro:lm5}, we have $|D_t(\mathbf{u},\mathbf{v},\mathbf{w})\cap D_t(\mathbf{u},\mathbf{v}',\mathbf{w})|\leq |D_t(\mathbf{u}^c,\mathbf{v},\mathbf{w})\cap D_t(\mathbf{u}^c,\mathbf{v}',\mathbf{w})|$, so we may further restrict to the case $\mathbf{u}=\mathbf{u}^c$. 

\textbf{Case A:}  
We prove by induction on \(p\) that for all \(1\le p\le q-3\) with $q\geq 4$,  the following assertions hold:
\begin{equation}\nonumber%\label{Upp:eq31}
|\mathcal{X}|\le N_q^{(p)}(n,t), 
\end{equation}
and moreover,
\begin{enumerate}
  \item if $\mathbf{x}_{[p+6,n]}\neq \mathbf{y}_{[p+6,n]}$, then $|\mathcal{X}|\leq  N_q^{(p)}(n,t)-n^{t-4}+O(n^{t-5})$;
  \item if $\mathbf{x}_{[p+6,n]}= \mathbf{y}_{[p+6,n]}\notin \mathbf{C}_q(n-p-5)$, then $|\mathcal{X}|\leq  N_q^{(p)}(n,t)-n^{t-1-q}+O(n^{t-2-q})$;
  \item if $\mathbf{x}_{[p+6,n]}= \mathbf{y}_{[p+6,n]}\in \mathbf{C}_q(n-p-5)$, and either $\{x_{i+p+5}\mid i\in[q-3]\}\neq [2,q-2]$ or $x_{p+q+3}\neq q-1$, then $|\mathcal{X}|\leq  N_q^{(p)}(n,t)-n^{t-q}+O(n^{t-1-q})$;
  \item suppose $\mathbf{x}_{[p+6,n]}= \mathbf{y}_{[p+6,n]}=\mathbf{c}_{q}(n-p-5,\sigma)$ with $\sigma=(2,3,\cdots,q-1,0,1)$, and $\mathbf{x}_{[1,p]}\in \mathbf{C}_q(p)$. If $\mathbf{x}_{[1,p]}\neq (q-1-p,\cdots,q-2)$, then $|\mathcal{X}|\leq  N_q^{(p)}(n,t)-n^{t-1-q}+O(n^{t-2-q})$.
\end{enumerate}

%Since the proofs of items 1-3 share the same idea, we only prove item 1; the others are analogous.
For convenience, we formalize the following two conditions.

\emph{Condition 3:}
\(
\mathbf{x}_{[p+6,n]}= \mathbf{y}_{[p+6,n]}\in \mathbf{C}_q(n-p-5),
\)
and either
\(
\{x_{i+p+5}\mid i\in[q-3]\}\neq [2,q-2]
~\text{or}~
x_{p+q+3}\neq q-1.
\)

\emph{Condition 4:}
\(
\mathbf{x}_{[p+6,n]}= \mathbf{y}_{[p+6,n]}=\mathbf{c}_{q}(n-p-5,\sigma)\) with $\sigma=(2,3,\cdots,q-1,0,1)$,
and
\(
\mathbf{x}_{[1,p]}\in \mathbf{C}_q(p)\) with
\(
\mathbf{x}_{[1,p]}\neq (q-1-p,\cdots,q-2).
\)

%For convenience, we define the \emph{Condition $3$} to be $\mathbf{x}_{[p+6,n]}= \mathbf{y}_{[p+6,n]}\in \mathbf{C}_q(n-p-5)$, and either $\{x_{i+p+5}\mid i\in[q-3]\}\neq [2,q-2]$ or $x_{p+q+3}\neq q-1$; define the \emph{Condition $4$} to be assume that $\mathbf{x}_{[p+6,n]}= \mathbf{y}_{[p+6,n]}\in \mathbf{C}_q(n-p-5)$, $\{x_{i+p+5}|i\in[q-3]\}=[2,q-2]$, $x_{p+q+3}=q-1$, $\mathbf{x}_{[1,p]}\in \mathbf{C}_q(p)$, and $\mathbf{x}_{[1,p]}\neq (q-1-p,\cdots,q-2)$.

\emph{The base case $p=1$}.

By the preceding discussion, we have
\(
\mathbf{x}_{[1,6]}=(x_1,0,1,q-1,0,1)\) and \( \mathbf{y}_{[1,6]}=(x_1,1,0,q-1,1,0).
\)
Let $a=x_1=y_1$. By Lemmas \ref{Pro:lm1} and \ref{Pro:lm2}, we have 
\(
|\mathcal{X}|=|\mathcal{X}^{a}|+\sum_{i\in \Sigma_q\setminus\{a\}}|\mathcal{X}^{i}|,
\)
where $\mathcal{X}^{a}=a\circ \big(D_t(\mathbf{x}_{[2,n]})\cap D_t(\mathbf{y}_{[2,n]})\big)$ and $\mathcal{X}^{i}=i\circ\big(D_{t-1-\ell_i}(\mathbf{x}_{[3+\ell_i,n]})\cap D_{t-1-\ell_i^*}(\mathbf{y}_{[3+\ell_i^*,n]})\big)$ with $\ell_i,\ell_i^*\geq 0$ for $i\in \Sigma_q\setminus\{a\}$. Here, the indices $\ell_i$ (respectively, $\ell_i^*$) are pairwise distinct.

When $p=1$, we have $\hat{\mathcal{X}}=D_{t-1}(\mathbf{x}_{[2,n]})\cap D_{t-1}(\mathbf{y}_{[2,n]})$. Then, $\hat{\mathcal{X}}^i=\mathcal{X}^i$ for all $i\neq a$, since $x_1=y_1=a$. Therefore,
\(
|\hat{\mathcal{X}}|=\sum_{i\in \Sigma_q\setminus\{a\}}|\mathcal{X}^{i}|+|\hat{\mathcal{X}}^a|.
\)
Since $d_L(\mathbf{x},\mathbf{y})\geq 2$ and $x_1=y_1$, we have $d_L(\mathbf{x}_{[2,n]},\mathbf{y}_{[2,n]})\geq 2$. Thus, for sufficiently large $n$, by Lemma \ref{Upp:lm6} we have
\begin{equation}
|\mathcal{X}^{a}|\leq N_q^{(0)}(n-1,t), \label{Upp:eq32}
\end{equation}
and we also obtain
\begin{align}
|\mathcal{X}^{a}|\overset{(a)}{\leq}& N_q^{(0)}(n-1,t)-n^{t-4}+O(n^{t-5}),~~\text{whenever $\mathbf{x}_{[7,n]}\neq \mathbf{y}_{[7,n]}$}, \label{Upp:eq33}\\
|\mathcal{X}^{a}|\overset{(b)}{\leq}& N_q^{(0)}(n-1,t)-n^{t-1-q}+O(n^{t-2-q}),~~\text{whenever $\mathbf{x}_{[7,n]}= \mathbf{y}_{[7,n]}\notin \mathbf{C}_q(n-6)$}, \label{Upp:eq33b}\\
|\mathcal{X}^{a}|\overset{(c)}{\leq}& N_q^{(0)}(n-1,t)-n^{t-q}+O(n^{t-1-q}),~~\text{whenever Condition $3$ holds}, \label{Upp:eq33c}
%|\mathcal{X}^{a}|\overset{(d)}{\leq}& N_q^{(0)}(n-1,t)-n^{t-1-q}+O(n^{t-5}),~~\text{whenever $\mathbf{x}, \mathbf{y}$ satisfy Condition $4$}, \label{Upp:eq33d}
\end{align}
where $(a)$ follows from Eq. \eqref{Upp:eq27}, $(b)$ follows from Eq. \eqref{Upp:eq21}, and $(c)$ follows from Eqs. \eqref{Upp:eq23} and \eqref{Upp:eq26}.

Combining the definition of $\hat{\mathcal{X}}$ with $\hat{\mathbf{x}}_{[1,5]}=(0,1,q-1,0,1)$ and $\hat{\mathbf{y}}_{[1,5]}=(1,0,q-1,1,0)$, and applying Eq.~\eqref{Upp:eq16} from the proof of \emph{Case A2-3} in Lemma~\ref{Upp:lm6}, we obtain 
\begin{align}\label{Upp:eq33d}
|\hat{\mathcal{X}}^0|+|\hat{\mathcal{X}}^1|+|\hat{\mathcal{X}}^{q-1}|\leq N_q^{(0)}(n-1,t-1)-\sum\limits_{i=3}^{q-1}D_q(n-4-i,t-3-i).
\end{align}
Note that $|\hat{\mathcal{X}}^{0}|$, $|\hat{\mathcal{X}}^{1}|$, and $|\hat{\mathcal{X}}^{q-1}|$ dominate the remaining terms $|\hat{\mathcal{X}}^{i}|$ for $2\le i\le q-2$, since $\hat{\mathbf{x}}_{[1,3]}=(0,1,q-1)$ and $\hat{\mathbf{y}}_{[1,3]}=(1,0,q-1)$. We choose \(b\in[2,q-2]\) as follows. If every symbol of \(\Sigma_q\) appears in \(\hat{\mathbf{x}}\), let \(b\) be the symbol whose first occurrence is the latest; otherwise, let \(b\) be any symbol absent from $\hat{\mathbf{x}}$. Then $|\hat{\mathcal{X}}^a|\geq |\hat{\mathcal{X}}^b|$. Hence, by the same argument as in the proof of Eq. \eqref{Upp:eq16}, we have
\begin{align}\label{Upp:eq33e}
\sum_{i\in [2,q-1]\setminus\{b\}}|\hat{\mathcal{X}}^i|\leq \sum\limits_{i=3}^{q-2}D_q(n-4-i,t-3-i).
\end{align}
Thus, by Eqs. \eqref{Upp:eq33d} and \eqref{Upp:eq33e},  we have 
\begin{align}
\sum_{i\in \Sigma_q\setminus\{a\}}|\mathcal{X}^{i}|&= |\hat{\mathcal{X}}|-|\hat{\mathcal{X}}^a|\leq|\hat{\mathcal{X}}|-|\hat{\mathcal{X}}^b|\leq  N_q^{(0)}(n-1,t-1)-D_q(n-q-3,t-q-2),\label{Upp:eq34}
\end{align}
Furthermore, if Condition~4 holds, then \(a\neq q-2\) and \(\mathbf{y}_{[p+6,n]}=\mathbf{c}_{q}(n-p-5,\sigma)\) with $\sigma=(2,3,\cdots,q-1,0,1)$, and we have the lower bound \(|\hat{\mathcal{X}}^a|\geq D_q(n-q-2,t-q-1)\). Consequently, since $|\hat{\mathcal{X}}|\leq N_q^{(0)}(n-1,t-1)$, we have 
\begin{align}
\sum_{i\in \Sigma_q\setminus\{a\}}|\mathcal{X}^{i}|&= |\hat{\mathcal{X}}|-|\hat{\mathcal{X}}^a|\leq  N_q^{(0)}(n-1,t-1)-D_q(n-q-2,t-q-1),~\text{whenever Condition $4$ holds}.~\label{Upp:eq34b}
\end{align}

Therefore, we conclude that
\begin{align}
|\mathcal{X}|&\overset{(a)}{\leq} N_q^{(0)}(n-1,t)+N_q^{(0)}(n-1,t-1)-D_q(n-q-3,t-q-2)=N_q^{(1)}(n,t), \label{Upp:eq35}\\
|\mathcal{X}|&\overset{(b)}{\leq} N_q^{(1)}(n,t)-n^{t-4}+O(n^{t-5}),~~\text{whenever $\mathbf{x}_{[7,n]}\neq \mathbf{y}_{[7,n]}$}, \nonumber\\%\label{Upp:eq36}
|\mathcal{X}|&\overset{(c)}{\leq} N_q^{(1)}(n,t)-n^{t-1-q}+O(n^{t-2-q}),~~\text{whenever $\mathbf{x}_{[7,n]}= \mathbf{y}_{[7,n]}\notin \mathbf{C}_q(n-6)$}, \nonumber\\
|\mathcal{X}|&\overset{(d)}{\leq} N_q^{(1)}(n,t)-n^{t-q}+O(n^{t-1-q}),~~\text{whenever Condition $3$ holds}, \nonumber\\
|\mathcal{X}|&\overset{(e)}{\leq} N_q^{(1)}(n,t)-n^{t-1-q}+O(n^{t-2-q}),~~\text{whenever Condition $4$ holds},\nonumber
\end{align}
where $(a)$ follows from Eqs. \eqref{Upp:eq28}, \eqref{Upp:eq32}, and \eqref{Upp:eq34}; $(b)$ follows from Eqs. \eqref{Upp:eq28}, \eqref{Upp:eq33}, and \eqref{Upp:eq34}; $(c)$ follows from Eqs. \eqref{Upp:eq28}, \eqref{Upp:eq33b}, and \eqref{Upp:eq34}; $(d)$ follows from Eqs. \eqref{Upp:eq28}, \eqref{Upp:eq33c}, and \eqref{Upp:eq34}; $(e)$ follows from Eqs. \eqref{intr:eq2}, \eqref{Upp:eq28}, \eqref{Upp:eq32}, and \eqref{Upp:eq34b}.

\emph{Inductive step.}

Fix an integer \(s\) with \(2\le s\le q-4\), and assume the assertion holds for all \(p\le s\). We then prove it for \(p=s+1\).

For convenience, let $U=\{u_1,u_2,\dots,u_{s+1}\}$ denote the set of symbols appearing in the common prefix $\mathbf{u}$. When $p=s+1$, we have 
\(
\hat{\mathcal{X}}=D_{t-s-1}(\mathbf{x}_{[s+2,n]})\cap D_{t-s-1}(\mathbf{y}_{[s+2,n]}),
\) 
\(x_{s+2}\neq y_{s+2}\),
\(
\mathbf{x}_{[s+2,s+6]}=(0,1,q-1,0,1),\) and \( \mathbf{y}_{[s+2,s+6]}=(1,0,q-1,1,0).
\)

By Lemmas \ref{Pro:lm1} and \ref{Pro:lm2}, we have
\(
|\mathcal{X}|=\sum_{i=1}^{s+1}|\mathcal{X}^{u_i}|+\sum_{a\in \Sigma_q\setminus U}|\mathcal{X}^{a}|,
\)
where $$|\mathcal{X}^{u_i}|=|D_{t-i+1}(\mathbf{u}_{[i+1,s+1]},\hat{\mathbf{x}})\cap D_{t-i+1}(\mathbf{u}_{[i+1,s+1]},\hat{\mathbf{y}})|$$ for $i\in [s+1]$. Since $\mathbf{x}_{[1,s+1]}=\mathbf{y}_{[1,s+1]}$ and $d_L(\mathbf{x},\mathbf{y})\geq 2$, we have $d_L(\mathbf{x}_{[s+2,n]},\mathbf{y}_{[s+2,n]})\geq 2$. When Condition $4$ holds, we distinguish the following two cases: \textbf{Case 4.1}: $x_1\neq q-2-s$ and $\mathbf{x}_{[2,s+1]}=(q-1-s,q-s,\cdots,q-2)$; \textbf{Case 4.2}: $\mathbf{x}_{[2,s+1]}\neq (q-1-s,q-s,\cdots,q-2)$. 

If \textbf{Case~4.2} of Condition~$4$ holds for \(\mathbf{x}\) and \(\mathbf{y}\), then Condition~$4$ also holds for the truncated sequences \(\mathbf{x}_{[2,n]}\) and \(\mathbf{y}_{[2,n]}\). Likewise, if Condition~$3$ holds for \(\mathbf{x}\) and \(\mathbf{y}\), then Condition~$3$ also holds for the truncated sequences \(\mathbf{x}_{[2,n]}\) and \(\mathbf{y}_{[2,n]}\). By the inductive hypothesis, we have
\begin{align}
|\mathcal{X}^{u_i}|&\leq N_q^{(s+1-i)}(n-i,t-i+1),~~\text{for each $i\in[1,s+1]$,}\label{Upp:eq37}\\
|\mathcal{X}^{u_1}|&\leq N_q^{(s)}(n-1,t)-n^{t-4}+O(n^{t-5}),~~\text{whenever $\mathbf{x}_{[s+7,n]}\neq \mathbf{y}_{[s+7,n]}$}, \label{Upp:eq38}\\
|\mathcal{X}^{u_1}|&\leq N_q^{(s)}(n-1,t)-n^{t-1-q}+O(n^{t-2-q}),~~\text{whenever $\mathbf{x}_{[s+7,n]}= \mathbf{y}_{[s+7,n]}\notin \mathbf{C}_q(n-s-6)$}, \label{Upp:eq38b}\\
|\mathcal{X}^{u_1}|&\leq N_q^{(s)}(n-1,t)-n^{t-q}+O(n^{t-1-q}),~~\text{whenever Condition $3$ holds}, \label{Upp:eq38c}\\
|\mathcal{X}^{u_1}|&\leq N_q^{(s)}(n-1,t)-n^{t-1-q}+O(n^{t-2-q}),~~\text{when \textbf{Case 4.2} of Condition $4$ holds}. \label{Upp:eq38d}
\end{align}
Observe that $\hat{\mathcal{X}}^a=\mathcal{X}^a$ for all $a\in \Sigma_q\setminus U$. This gives the decomposition
\[
|\hat{\mathcal{X}}|=\sum_{a\in \Sigma_q\setminus U}|\mathcal{X}^{a}|+\sum_{b\in U}|\hat{\mathcal{X}}^b|.
\]
For sufficiently large $n$, applying the same method as in the case \(p=1\) together with Eq. \eqref{Upp:eq16} from the proof of \textbf{Case A2-3} in Lemma~\ref{Upp:lm6}, 
we obtain 
\begin{align}
\sum_{a\in \Sigma_q\setminus U}|\mathcal{X}^{a}|=&|\mathcal{X}|-\sum_{b\in U}|\hat{\mathcal{X}}^b|\leq N_q^{(0)}(n-s-1,t-s-1)-\sum_{j=1}^{s+1}D_q(n-4-s-q+j,t-3-s-q+j),\label{Upp:eq39}
\end{align}
Moreover, if \textbf{Case 4.1} of Condition~$4$ holds, then \(x_1\neq q-2-s\), $\mathbf{x}_{[2,s+1]}=(q-1-s,q-s,\cdots,q-2)$, and $\mathbf{x}_{[s+7,n]}=\hat{\mathbf{x}}_{[6,n-s-7]}=\hat{\mathbf{y}}_{[6,n-s-7]}=\mathbf{y}_{[s+7,n]}=\mathbf{c}_{q}(n-s-6,\sigma)$ with $\sigma=(2,3,\cdots,q-1,0,1)$. Thus, $U=\{x_1,q-1-s,q-s,\cdots,q-2\}$ with $x_1\in [0,q-3-s]\cup \{q-1\}$. By the definition of $\hat{\mathcal{X}}$, we have the lower bound $\sum_{b\in U}|\hat{\mathcal{X}}^b|\geq D_q(n-2-q,t-1-q)+\sum_{j=1}^{s}D_q(n-4-s-q+j,t-3-s-q+j)$. Consequently, since $|\hat{\mathcal{X}}|\leq N_q^{(0)}(n-s-1,t-s-1)$, we have
\begin{align}
\sum_{a\in \Sigma_q\setminus U}|\mathcal{X}^{a}|=&|\mathcal{X}|-\sum_{b\in U}|\hat{\mathcal{X}}^b|\leq N_q^{(0)}(n-s-1,t-s-1)-D_q(n-2-q,t-1-q)\nonumber\\
&~-\sum_{j=1}^{s}D_q(n-4-s-q+j,t-3-s-q+j),~~~~~\text{when \textbf{Case 4.1} of Condition~$4$ holds}\label{Upp:eq39b}.
\end{align}

%\begin{align}
%N_q^{(i)}(n,t)=\sum\limits_{j=0}^{i-1}N_q^{(j)}(n-i+j,t-i+j+1)+N_q^{(0)}(n-i,t-i)-\sum\limits_{j=1}^{i}D_q(n-3-i-q+j,t-2-i-q+j).\label{Upp:eq28}
%\end{align}  

Therefore, we have
\begin{align}
|\mathcal{X}|&\overset{(a)}{\leq} \sum_{i=0}^{s}N_q^{(i)}(n-s-1+i,t-s+i)+N_q^{(0)}(n-s-1,t-s-1)-\sum_{i=0}^{s}D_q(n-s-q-3+i,t-s-q-2+i) \nonumber\\
&=N_q^{(s+1)}(n,t),\label{Upp:eq40}\\
|\mathcal{X}|&\overset{(b)}{\leq} N_q^{(s+1)}(n,t)-n^{t-4}+O(n^{t-5}),~~\text{whenever $\mathbf{x}_{[s+7,n]}\neq \mathbf{y}_{[s+7,n]}$}, \nonumber\\
|\mathcal{X}|&\overset{(c)}{\leq} N_q^{(s+1)}(n,t)-n^{t-1-q}+O(n^{t-2-q}),~~\text{whenever $\mathbf{x}_{[s+7,n]}=\mathbf{y}_{[s+7,n]}\notin \mathbf{C}_q(n-s-6)$}, \nonumber\\
|\mathcal{X}|&\overset{(d)}{\leq} N_q^{(s+1)}(n,t)-n^{t-q}+O(n^{t-1-q}),~~\text{whenever Condition $3$ holds}, \nonumber\\
|\mathcal{X}|&\overset{(e)}{\leq} N_q^{(s+1)}(n,t)-n^{t-1-q}+O(n^{t-2-q}),~~\text{whenever Condition $4$ holds}, \nonumber%\label{Upp:eq40}
\end{align}
where $(a)$ follows from Eqs. \eqref{Upp:eq28}, \eqref{Upp:eq37}, and \eqref{Upp:eq39}; $(b)$ follows from Eqs. \eqref{Upp:eq28}, \eqref{Upp:eq37}, \eqref{Upp:eq38}, and \eqref{Upp:eq39}; $(c)$ follows from Eqs. \eqref{Upp:eq28}, \eqref{Upp:eq37}, \eqref{Upp:eq38b}, and \eqref{Upp:eq39}; $(d)$ follows from Eqs. \eqref{Upp:eq28}, \eqref{Upp:eq37}, \eqref{Upp:eq38c}, and \eqref{Upp:eq39}; $(e)$ follows from Eqs. \eqref{intr:eq2}, \eqref{Upp:eq28}, \eqref{Upp:eq37}, \eqref{Upp:eq38d}, \eqref{Upp:eq39}, and \eqref{Upp:eq39b}.

By induction,  for all integers $p$ with $1\leq p\leq q-3$, we conclude that $|\mathcal{X}|\leq N_q^{(p)}(n,t)$; $|\mathcal{X}|\leq N_q^{(p)}(n,t)-n^{t-4}+O(n^{t-5})$ whenever $\mathbf{x}_{[p+6,n]}\neq \mathbf{y}_{[p+6,n]}$; $|\mathcal{X}|\leq  N_q^{(p)}(n,t)-n^{t-1-q}+O(n^{t-2-q})$ whenever $\mathbf{x}_{[p+6,n]}= \mathbf{y}_{[p+6,n]}\notin \mathbf{C}_q(n-p-5)$; $|\mathcal{X}|\leq  N_q^{(p)}(n,t)-n^{t-q}+O(n^{t-1-q})$ whenever Condition $3$ holds;  $|\mathcal{X}|\leq  N_q^{(p)}(n,t)-n^{t-1-q}+O(n^{t-2-q})$ whenever Condition $4$ holds.

To characterize the sequences achieving the bound $N_q^{(q-3)}(n,t)$, we must examine all possible common prefixes $\mathbf{u}$, including non-cyclic ones.
Consider the case where $p=q-3$. If $\mathbf{x}_{[1,q-3]}\notin \mathbf{C}_q(q-3)$, then there exist two distinct indices $i<j\in[q-3]$ such that $u_i=u_j$, and there exists an element $a\in \Sigma_q\setminus\big(\{u_i\mid i\in[q-3]\}\cup \{0,1,q-1\}\big)$. %Consequently, 
      %among the contributions of $\mathcal{X}^{u_i}$ and $\mathcal{X}^{u_j}$, only one can retain the original bound $N_q^{(q-3-i)}(n-i,t-i+1)$; 
Since \(u_i=u_j\) with \(i<j\), the two components \(\mathcal{X}^{u_i}\) and \(\mathcal{X}^{u_j}\) correspond to the same symbol, so they cannot both contribute the distinct upper bounds \(N_q^{(q-3-i)}(n-i,t-i+1)\) and \(N_q^{(q-3-j)}(n-j,t-j+1)\). More precisely, the bound corresponding to the later index \(j\) is lost. Because the symbol \(a\) does not appear among \(\{u_i\mid i\in[q-3]\}\cup\{0,1,q-1\}\), its first occurrence in \(\mathbf{x}\) occurs no earlier than position \(q+3\), and hence its contribution is bounded above by the smaller quantity \(D_q(n-q-3,t-q-2)\). Thus,
\begin{align}
|\mathcal{X}|&\leq N_q^{(q-3)}(n,t)-N_q^{(q-3-j)}(n-j,t-j+1)+D_q(n-q-3,t-q-2)\nonumber\\
&\overset{(a)}{\leq} N_q^{(q-3)}(n,t)-\frac{6}{(t-j-1)!}n^{t-j-1}+O(n^{t-j-2})\overset{(b)}{\leq} N_q^{(q-3)}(n,t)-\frac{6}{(t-q+2)!}n^{t-q+2}+O(n^{t-q+1}),\label{Upp:eq39c}
\end{align}
where $(a)$ follows from Lemma \ref{Low:lm4}, Eq. \eqref{intr:eq2}, and $2\leq j\leq q-3$, $(b)$ follows from $2\leq j\leq q-3$.

To summarize, we have the following conclusions.
\begin{enumerate}
  \item By Lemma~\ref{Low:lm4}, for all \(0\le p\le q-4\) and \(q\ge4\), we have 
\begin{equation}\label{Upp:eq41}
|\mathcal{X}| \le N_q^{(p)}(n,t) \leq N_q^{(q-3)}(n,t)-\frac{5}{(t-q-2)!}n^{t-q-2}+O(n^{t-q-3})\leq N_q^{(q-3)}(n,t).
\end{equation}
  \item Let $p=q-3$ and $t\geq q+2$. For sufficiently large $n$, by Eq. \eqref{Upp:eq39c} and the results of \textbf{Case A} for $p=q-3$,  we have
  \begin{align}
  |\mathcal{X}|\leq N_q^{(q-3)}(n,t),\label{Upp:eq42}
  \end{align}  
  and equality holds in Eq. \eqref{Upp:eq42} only if there exists an ordering $\sigma=(c_1,\dots,c_{q-2},c_{q-1},c_q)$ of $\Sigma_q$ such that
  $\mathbf{x}_{[1,q-3]}=\mathbf{y}_{[1,q-3]}=(c_1,\cdots,c_{q-3}),\{\mathbf{x}_{[q-2,q+2]},\mathbf{y}_{[q-2,q+2]}\}=\{(c_q,c_{q-1},c_{q-2},c_{q},c_{q-1}),
      (c_{q-1},c_{q},c_{q-2},c_{q-1},c_{q})\}$, and $\mathbf{x}_{[q+3,n]}=\mathbf{y}_{[q+3,n]}=\mathbf{c}_q(n-q-2,\sigma)$. 
\end{enumerate}

\textbf{Case B:} We consider two subcases based on the value of $p$: \textbf{Case B1:} $p=q-2$; \textbf{Case B2:} $p=q-1$.

\textbf{Case B1:} $p=q-2$.

Then we have $\mathbf{u}\in \mathbf{C}_q(q-2)$, $\mathbf{x}_{[q-1,q+3]}=(0,1,q-1,0,1)$, and $\mathbf{y}_{[q-1,q+3]}=(1,0,q-1,1,0)$. Let $U=\{u_1,u_2,\dots,u_{q-2}\}$ and $\{b,c\}=\Sigma_q\setminus U$. By Lemmas $\ref{Pro:lm1}$ and $\ref{Pro:lm2}$, we have $|\mathcal{X}|=\sum\limits_{i\in [q-2]}|\mathcal{X}^{u_i}|+|\mathcal{X}^{b}|+|\mathcal{X}^{c}|$.

Recall that $d_L(\hat{\mathbf{x}}, \hat{\mathbf{y}})=d_L(\mathbf{x}_{[q-1,n]},\mathbf{y}_{[q-1,n]}) = d_L(\mathbf{x}, \mathbf{y}) \geq 2$. By the result of \textbf{Case A} (with $p \leq q-3$), for each $i\in[q-2]$, we obtain
\begin{equation}
|\mathcal{X}^{u_i}|\leq N_q^{(q-2-i)}(n-i,t-i+1). \label{Upp:eq43}
\end{equation}

We now bound the sum $|\mathcal{X}^b|+|\mathcal{X}^c|$. Note that $\hat{\mathcal{X}}^a=\mathcal{X}^a$ for all $a\in\Sigma_q\setminus U$, where $\hat{\mathcal{X}}=D_{t-q+2}(\mathbf{x}_{[q-1,n]})\cap D_{t-q+2}(\mathbf{y}_{[q-1,n]})$. Since the prefix of $\mathbf{x}_{[q-1,n]}$ is $(1,0,q-1,1,0)$ and that of $\mathbf{y}_{[q-1,n]}$ is $(0,1,q-1,0,1)$, we have $\min\{|\hat{\mathcal{X}}^0|,|\hat{\mathcal{X}}^1|\}\geq |\hat{\mathcal{X}}^i|$ for all $i\in[2,q-1]$. Therefore,
\begin{align}
|\mathcal{X}^b|+|\mathcal{X}^c|\leq |\hat{\mathcal{X}}^0|+|\hat{\mathcal{X}}^1|\overset{(a)}{\leq}& 2\big(D_q(n-q,t-q+1)+D_q(n-q-3,t-q)-D_q(n-q-2,t-q+1)\big),\label{Upp:eq44}
\end{align}
where $(a)$ follows from Eq. \eqref{Upp:eq18} in Lemma \ref{Upp:lm6}. 

By Lemma \ref{Low:lm4}, we have $N_q^{(i)}(m,l)\leq N_q^{(q-3)}(m,l)$ for all $i\leq q-3$ and $m\geq l+3q$.  Combining this with Eqs. \eqref{Upp:eq43} and \eqref{Upp:eq44}, we obtain
\begin{align}
|\mathcal{X}|&\leq \sum_{i=0}^{q-3}N_q^{(q-3)}(n-q+2+i,t-q+3+i)+|\hat{\mathcal{X}}^0|+|\hat{\mathcal{X}}^1| \nonumber\\
&\overset{(a)}{=} N_q^{(q-3)}(n,t) - \left[ N_q^{(q-3)}(n-q+1,t-q+2)+N_q^{(q-3)}(n-q,t-q+1)- \left(|\hat{\mathcal{X}}^0|+|\hat{\mathcal{X}}^1|\right)\right], \label{Upp:eq45}
\end{align}
where $(a)$ follows from the recurrence  $N_q^{(i-1)}(n,t)=\sum_{k=0}^{q-1}N_q^{(i-1)}(n-1-k,t-k)$ for sufficiently large $n$.
 
For sufficiently large $n$,
\begin{align*}
&N_q^{(q-3)}(n-q+1,t-q+2)+N_q^{(q-3)}(n-q,t-q+1) \overset{(a)}{=} \frac{6}{(t-q)!}n^{t-q} - \frac{3(t-q)+6q+7}{(t-q-1)!}n^{t-q-1} + O(n^{t-q-2}),\\
&|\hat{\mathcal{X}}^0|+|\hat{\mathcal{X}}^1|\overset{(b)}{\leq}  \frac{6}{(t-q)!}n^{t-q} - \frac{3(t-q)+6q+9}{(t-q-1)!}n^{t-q-1} + O(n^{t-q-2}),
\end{align*}
where $(a)$ follows from Lemma \ref{Low:lm4}, $(b)$ follows from Eq. \eqref{def:eq1} and Lemma \ref{Pro:lm10}. Thus,
\[
N_q^{(q-3)}(n-q+1,t-q+2)+N_q^{(q-3)}(n-q,t-q+1) - \left(|\hat{\mathcal{X}}^0|+|\hat{\mathcal{X}}^1|\right) \geq  \frac{2}{(t-q-1)!}n^{t-q-1} + O(n^{t-q-2}).
\]
Combining this with Eq. \eqref{Upp:eq45}, we conclude that for sufficiently large $n$,
\begin{align}
|\mathcal{X}| \leq N_q^{(q-3)}(n,t) - \frac{2}{(t-q-1)!}n^{t-q-1} + O(n^{t-q-2}) \leq N_q^{(q-3)}(n,t).\label{Upp:eq46}
\end{align}

\textbf{Case B2:} $p=q-1$.

Then $\mathbf{u}\in \mathbf{C}_q(q-1)$, $\mathbf{x}_{[q,q+4]}=(0,1,q-1,0,1)$, and $\mathbf{y}_{[q,q+4]}=(1,0,q-1,1,0)$. Let $\hat{U}=\{u_1,u_2,\dots,u_{q-1}\}$ and $\{b\}=\Sigma_q\setminus \hat{U}$. By Lemmas $\ref{Pro:lm1}$ and $\ref{Pro:lm2}$, we have $|\mathcal{X}|=\sum\limits_{i\in [q-1]}|\mathcal{X}^{u_i}|+|\mathcal{X}^{b}|$.

Recall that $d_L(\hat{\mathbf{x}}, \hat{\mathbf{y}})=d_L(\mathbf{x}_{[q,n]},\mathbf{y}_{[q,n]}) = d_L(\mathbf{x}, \mathbf{y}) \geq 2$. By the results of \textbf{Case A} and \textbf{Case B1}, we obtain 
\begin{align}
|\mathcal{X}^{u_1}|\leq N_q^{(q-3)}(n-1,t) - \frac{2}{(t-q-1)!}n^{t-q-1} +O(n^{t-q-2}), \quad
|\mathcal{X}^{u_i}|\leq N_q^{(q-1-i)}(n-i,t-i+1), \label{Upp:eq47}
\end{align}
for $i\in[2,q-1]$. 

We now bound $|\mathcal{X}^b|$. Note that $\hat{\mathcal{X}}^b=\mathcal{X}^b$, where $\hat{\mathcal{X}}=D_{t-q+1}(\mathbf{x}_{[q,n]})\cap D_{t-q+1}(\mathbf{y}_{[q,n]})$. By Eq. \eqref{Upp:eq18} in Lemma \ref{Upp:lm6}, we have
\begin{align}\label{Upp:eq48}
|\mathcal{X}^b| \leq \max\{|\hat{\mathcal{X}}^0|,|\hat{\mathcal{X}}^1|\} = D_q(n-q-1,t-q)+D_q(n-q-4,t-q-1)-D_q(n-q-3,t-q).
\end{align}

Combining this with Eqs. \eqref{Upp:eq47} and \eqref{Upp:eq48}, we obtain
\[
|\mathcal{X}|\leq \sum_{i=0}^{q-3}N_q^{(i)}(n-q+1+i,t-q+2+i)+N_q^{(q-3)}(n-1,t)+|\mathcal{X}^b|-\frac{2}{(t-q-1)!}n^{t-q-1} +O(n^{t-q-2}).
\]

By Lemma \ref{Low:lm4}, we have $N_q^{(i)}(m,l)\leq N_q^{(q-3)}(m,l)$ for all $i\leq q-3$ and $m\geq l+3q$. Thus, 
\[
|\mathcal{X}|\leq N_q^{(q-3)}(n,t) - N_q^{(q-3)}(n-q,t-q+1) + |\mathcal{X}^b|-\frac{2}{(t-q-1)!}n^{t-q-1} +O(n^{t-q-2}).
\]

By Lemmas \ref{Pro:lm10} and \ref{Low:lm4}, we have 
\begin{align*}
N_q^{(q-3)}(n-q,t-q+1)= \frac{6}{(t-q-1)!}n^{t-q-1} + O(n^{t-q-2}),\qquad |\mathcal{X}^b|\leq  \frac{3}{(t-q-1)!}n^{t-q-1} + O(n^{t-q-2}).
\end{align*}

Therefore, for sufficiently large $n$,
\begin{align}
|\mathcal{X}| \leq N_q^{(q-3)}(n,t) - \frac{5}{(t-q-1)!}n^{t-q-1} + O(n^{t-q-2}) \leq N_q^{(q-3)}(n,t).\label{Upp:eq49}
\end{align}

\textbf{Case C:} We split into two subcases depending on whether $\hat{\mathbf{x}}_{[1,5]}$ and $\hat{\mathbf{y}}_{[1,5]}$ are of the form $(a,b,c,a,b)$ and $(b,a,c,b,a)$, respectively.

\emph{1) $\hat{\mathbf{x}}_{[1,5]}$ and $\hat{\mathbf{y}}_{[1,5]}$ are of the forms $(a,b,c,a,b)$ and $(b,a,c,b,a)$, respectively.} 

By Lemmas $\ref{Pro:lm1}$ and $\ref{Pro:lm2}$, we have $|\mathcal{X}|=\sum\limits_{i\in [q]}|\mathcal{X}^{u_i}|$. We prove by induction that for all $p \geq q-1$ and sufficiently large $n$, 
$$|\mathcal{X}| \leq N_q^{(q-3)}(n,t)-\frac{5}{(t-q-1)!}n^{t-q-1} + O(n^{t-q-2})\leq N_q^{(q-3)}(n,t).$$ 
The base case $p=q-1$ has already been established in \textbf{Case B} (see, Eq. \eqref{Upp:eq49}). For the inductive step, assume the assertion holds for all $q-1\leq p \leq s$, where $s \geq q-1$. We consider the case $p = s+1$. For convenience, let $a=0,b=1,c=q-1$. Then $\mathbf{u}\in \mathbf{C}_q(s+1)$, $\mathbf{x}_{[s+2,s+6]}=(0,1,q-1,0,1)$, $\mathbf{y}_{[s+2,s+6]}=(1,0,q-1,1,0)$, and $\Sigma_q=\{u_1,u_2,\dots,u_{q}\}$. 

Recall that $d_L(\hat{\mathbf{x}}, \hat{\mathbf{y}})=d_L(\mathbf{x}_{[s+2,n]},\mathbf{y}_{[s+2,n]}) = d_L(\mathbf{x}, \mathbf{y}) \geq 2$. 
Thus, for $i\in [2,q]$, we have
\begin{align}
|\mathcal{X}^{u_1}| &\overset{(a)}{\leq} N_q^{(q-3)}(n-1,t)-\frac{5}{(t-q-1)!}(n-1)^{t-q-1} + O(n^{t-q-2}), \label{Upp:eq50}\\
|\mathcal{X}^{u_i}| &\overset{(b)}{\leq} N_q^{(q-3)}(n-i,t-i+1),\label{Upp:eq51}
\end{align}
where $(a)$ follows from Eq.~\eqref{Upp:eq49} when $s=q-1$, and from the induction hypothesis when $s\ge q$; $(b)$ follows from \textbf{Cases A} or \textbf{B} when $s+1-i\le q-1$, and from the induction hypothesis when $s+1-i\ge q$.
 
Hence, for sufficiently large $n$, combining Eqs.~\eqref{Upp:eq50} and~\eqref{Upp:eq51} yields
\[
|\mathcal{X}| \leq \sum_{i=1}^q N_q^{(q-3)}(n-i,t-i+1)-\frac{5}{(t-q-1)!}n^{t-q-1} + O(n^{t-q-2}) = N_q^{(q-3)}(n,t)-\frac{5}{(t-q-1)!}n^{t-q-1} + O(n^{t-q-2}).
\]

By induction, we conclude that for all integers $p \geq q$ and sufficiently large $n$,
\begin{align}
 |\mathcal{X}| \leq N_q^{(q-3)}(n,t)-\frac{5}{(t-q-1)!}n^{t-q-1} + O(n^{t-q-2}). \label{Upp:eq52}
\end{align}

\emph{2) $\hat{\mathbf{x}}_{[1,5]}$ and $\hat{\mathbf{y}}_{[1,5]}$ are not of the form  $(a,b,c,a,b)$ and $(b,a,c,b,a)$, respectively.} 

Then $\mathbf{u}\in \mathbf{C}_q(s+1)$ and $\Sigma_q=\{u_1,u_2,\dots,u_{q}\}$. By Lemmas $\ref{Pro:lm1}$ and $\ref{Pro:lm2}$, we have $|\mathcal{X}|=\sum\limits_{i\in [q]}|\mathcal{X}^{u_i}|$.

We prove by induction that for all $p \geq q+1$ and sufficiently large $n$, 
$$|\mathcal{X}| \leq \frac{6}{(t-2)!} n^{t-2} - \frac{3t+15}{(t-3)!} n^{t-3} + O(n^{t-4}) < N_q^{(0)}(n, t).$$ 
The base case $1\leq p\leq q$ has already been established in Eq.  \eqref{Upp:eq30}. For the inductive step, assume the claim holds for all $q\leq p \leq s$, where $s \geq q$. We consider the case $p = s+1$.

Recall that $d_L(\hat{\mathbf{x}}, \hat{\mathbf{y}})=d_L(\mathbf{x}_{[s+2,n]},\mathbf{y}_{[s+2,n]}) = d_L(\mathbf{x}, \mathbf{y}) \geq 2$. 
Thus, by the induction hypothesis,
\begin{align}
|\mathcal{X}^{u_1}| &\leq \frac{6}{(t-2)!} {(n-1)}^{t-2} - \frac{3t+15}{(t-3)!} {(n-1)}^{t-3} + O(n^{t-4}), \nonumber\\
|\mathcal{X}^{u_2}| &\leq \frac{6}{(t-3)!} {(n-2)}^{t-3} + O(n^{t-4}), \nonumber\\
|\mathcal{X}^{u_i}| &= O(n^{t-4}),\nonumber
\end{align}
for $i\in [3,q]$.

Hence, for sufficiently large $n$, we obtain
\begin{align*}
|\mathcal{X}| &\leq \frac{6}{(t-2)!} {(n-1)}^{t-2} - \frac{3t+15}{(t-3)!} {(n-1)}^{t-3}+ \frac{6}{(t-3)!} {(n-2)}^{t-3} +O(n^{t-4})\\
&=\frac{6}{(t-2)!} n^{t-2} - \frac{3t+15}{(t-3)!} n^{t-3} + O(n^{t-4}) \leq N_q^{(0)}(n,t).
\end{align*}

By induction, we conclude that for all integers $p \geq q$ and sufficiently large $n$,
\begin{align}
 |\mathcal{X}| \leq\frac{6}{(t-2)!} n^{t-2} - \frac{3t+15}{(t-3)!} n^{t-3} + O(n^{t-4}) \leq  N_q^{(0)}(n,t). \label{Upp:eq53}
\end{align}

Note that $N_q^{(q-3)}(n,t)=N_q(n,t)$. Combining the results from \textbf{Cases} \textbf{A}, \textbf{B}, and \textbf{C} (see, Eqs. \eqref{Upp:eq41}, \eqref{Upp:eq42}, \eqref{Upp:eq46}, \eqref{Upp:eq49}, \eqref{Upp:eq52}, and \eqref{Upp:eq53}), we have proved that for any pair of sequences $\mathbf{x},\mathbf{y}\in\Sigma_q^n$ with $d_L(\mathbf{x},\mathbf{y})\geq 2$,
\[
|D_t(\mathbf{x}) \cap D_t(\mathbf{y})| \leq N^{(q-3)}(n,t)=N_q(n,t)
\]
for all sufficiently large $n$. 

Combined with the lower bound from Lemma~\ref{Low:lm4} and the upper bound from Lemma~\ref{Upp:lm6}, this establishes the exact asymptotic result
\begin{align}\label{Upp:eq54}
N_q(n,2,t) = N_q(n,t)
\end{align}
for sufficiently large $n$ and $t \geq 2$.

Moreover, by Lemma~\ref{Pro:lm4}, permutations and reversals preserve both the Levenshtein distance and the size of the intersection of deletion balls. For $t\ge q+2$, equality in Eq. \eqref{Upp:eq54} holds if and only if $\{\mathbf{x},\mathbf{y}\}\in\mathcal{P}$, where
\[
\mathcal{P}=
\left\{ \{\pi(\mathbf{x}^{(q-3)}), \pi(\mathbf{y}^{(q-3)})\} : \pi\in S_q \right\}
\cup
\left\{ \{(\pi(\mathbf{x}^{(q-3)}))^R, (\pi(\mathbf{y}^{(q-3)}))^R\} : \pi\in S_q \right\},
\]
and the sequences $\mathbf{x}^{(q-3)}$ and $\mathbf{y}^{(q-3)}$ are as defined in Eq. \eqref{Upp:eq29}. The two sets in the union are disjoint, and each has cardinality $q!$ because for distinct $\pi,\beta\in S_q$, the pairs $\{\pi(\mathbf{x}^{(q-3)}),\pi(\mathbf{y}^{(q-3)})\}$ and $\{\beta(\mathbf{x}^{(q-3)}),\beta(\mathbf{y}^{(q-3)})\}$ are distinct, and the same holds for the reversed families. Hence $|\mathcal{P}|=2(q!)$. Consequently, for sufficiently large $n$ and $t\ge q+2$, there are exactly $2(q!)$ unordered pairs of sequences in $\mathcal{P}$ whose intersection attains the maximum value $N_q(n,2,t)$.
\end{proof}

From the proofs of Lemma \ref{Upp:lm6} and Theorem \ref{def:thm6}, we obtain the following characterization of the intersection size $|D_t(\mathbf{x})\cap D_t(\mathbf{y})|$ for any $\mathbf{x},\mathbf{y}\in\Sigma_q^n$ with $d_L(\mathbf{x},\mathbf{y})\geq 2$. 

\begin{remark}
Let $\mathbf{x}=(x_1,\dots,x_n), \mathbf{y}=(y_1,\dots,y_n)\in \Sigma_q^n$ with $d_L(\mathbf{x},\mathbf{y})\geq 2$, and define $\mathcal{X}=D_t(\mathbf{x})\cap D_t(\mathbf{y})$.
\begin{itemize}
\item  By Eqs. \eqref{Upp:eq30} and \eqref{Upp:eq53} and item $1$ of Remark \ref{Upp:rmk1}, for $t \geq 3$, $|\mathcal{X}|$ matches the first two  terms of $N_q(n,2,t)$ only if there exist distinct \(a,b,c\in\Sigma_q\) and an index \(i\in[0,n-5]\) such that
$\mathbf{x}_{[i+1,i+5]}=(a,b,c,a,b)$ and $\mathbf{y}_{[i+1,i+5]}=(b,a,c,b,a)$.
\item  By Eqs. \eqref{Upp:eq30} and \eqref{Upp:eq53} and item $1$ of Remark \ref{Upp:rmk1}, if there do not exist three distinct $a,b,c\in\Sigma_q$ and an index $i\in[0,n-5]$ such that $\mathbf{x}_{[i+1,i+5]}=(a,b,c,a,b)$ and $\mathbf{y}_{[i+1,i+5]}=(b,a,c,b,a)$, then $|\mathcal{X}|\leq  \frac{6}{(t-2)!}n^{t-2}-\frac{3t+15}{(t-3)!}n^{t-3} +O(n^{t-4})$.
%\item For $t= q+1$, $|\mathcal{X}|$ has the first $q$  terms of $N_q(n,2,t)$ if and only if  $\mathbf{x}_{[1,i]}=\mathbf{y}_{[1,i]}=(c_1,\cdots,c_i)$, $\mathbf{x}_{[i+1,i+5]}=(a,b,c_{q-2},a,b)$, $\mathbf{y}_{[i+1,i+5]}=(b,a,c_{q-2},b,a)$, and $\mathbf{x}_{[i+6,n]}=\mathbf{y}_{[i+6,n]}=\mathbf{c}_q(n-5-i,\sigma)$, where $\sigma$ is an ordering of $\Sigma_q$ of the form
%\[
%\begin{aligned}
%\sigma &= (c_1,\dots,c_{q-2},a,b), \text{or} \quad \sigma &= (c_1,\dots,c_{q-2},b,a),
%\end{aligned}
%\]
%with $\{c_j \mid j\in[q-2]\} = \Sigma_q \setminus \{a,b\}$ and $i\in [0,q-3]$. 
%\item For $t\geq q+2$, $|\mathcal{X}|=N_q(n,2,t)$ if and only if $\mathbf{x}_{[1,q-3]}=\mathbf{y}_{[1,q-3]}=(c_1,\cdots,c_{q-3}),\mathbf{x}_{[q-2,q+2]}=(a,b,c_{q-2},a,b)$, $\mathbf{y}_{[q-2,q+2]}=(b,a,c_{q-2},b,a)$, and $\mathbf{x}_{[q+3,n]}=\mathbf{y}_{[q+3,n]}=\mathbf{c}_q(n-q-2,\sigma)$, where $\sigma$ is an ordering of $\Sigma_q$ of the form
%\[
%\begin{aligned}
%\sigma &= (c_1,\dots,c_{q-2},a,b), \quad\text{or} \quad \sigma &= (c_1,\dots,c_{q-2},b,a),
%\end{aligned}
%\]
%with $\{c_i \mid i\in[q-2]\} = \Sigma_q \setminus \{a,b\}$. 
\end{itemize}
\label{Upp:rmk2}
\end{remark}

Finally, we characterize the asymptotic behavior of $N_q(n,2,t)$ for $t \geq 2$, and compare it with the corresponding asymptotics of $N_3(n,2,t)$ and $N_2(n,2,t)$. We recall that the asymptotics of $N_i(n,2,t)$ for $i\in \{2,3\}$ were established in \cite{Wang}.

\begin{lemma}[Wang, Li, and Fu \cite{Wang}]
For $t\geq 3$,
\begin{align*}
 N_3(n,2,t)= \frac{6}{(t-2)!}n^{t-2}-\frac{(3t+13)}{(t-3)!}n^{t-3}+O(n^{t-4}),\qquad
 N_2(n,2,t)= \frac{6}{(t-2)!}n^{t-2}-\frac{(9t+3)}{(t-3)!}n^{t-3}+O(n^{t-4}).
\end{align*}
Moreover,
\begin{equation*}
N_3(n,2,t)-N_2(n,2,t)=\frac{6t-10}{(t-3)!}n^{t-3}+O(n^{t-4}).
\end{equation*}
\label{Upp:lm7}
\end{lemma}

Combining Remark~\ref{Low:rmk2}, Lemmas~\ref{Low:lm4} and~\ref{Upp:lm7}, and Theorem~\ref{Upp:thm1}, we establish the asymptotic behavior of $N_q(n,2,t)$ in the following theorem.

\begin{theorem}
\label{Upp:thm2}
Let $t\geq 2, q\geq 3$, and let $n$ be sufficiently large. Then
\begin{align*}
 N_q(n,2,t)=\frac{6}{(t-2)!}n^{t-2}-\frac{3t+13}{(t-3)!}n^{t-3}+\frac{3t^2+25t+64}{4(t-4)!}n^{t-4}+O(n^{t-5}),
\end{align*}
Moreover, for $q\geq 4$,
\[
N_q(n,2,t)-N_{q-1}(n,2,t) =
\begin{cases}
0, & \text{if } q\ge t+1, \\
5, & \text{if } q=t, \\
\dfrac{(6t-6q+5)n^{t-q}}{(t-q)!} + O(n^{t-q-1}), & \text{if } q\le t-1.
\end{cases}
\]
Furthermore, for $q\geq 4$ and $t\geq 4$,
\begin{align*}
N_q(n,2,t)-N_2(n,2,t)=\frac{6t-10}{(t-3)!}n^{t-3}+O(n^{t-4}),\qquad
N_q(n,2,t)-N_3(n,2,t)=\frac{6t-19}{(t-4)!}n^{t-4}+O(n^{t-5}).
\end{align*}

\end{theorem}

\section{Conclusion}
\label{sec6}

In this paper, we have studied  the sequence reconstruction problem over $q$-ary deletion channels with $q\geq 3$. We have established the exact value of $N_q(n,2,t)$ for all $t\geq 2$, $q\geq 4$, and sufficiently large $n$. Specifically, we show that
\(
N_q(n,2,t) = N_q^{(q-3)}(n,t),
\)
where $N_q^{(q-3)}(n,t)$ is the maximum size of the intersection of two deletion balls of radius $t$. This maximum is achieved by the pair of sequences
\(
\{\mathbf{x},\mathbf{y}\}=\{(c_1,\dots,c_{q-3},c_q,c_{q-1},c_{q-2},c_q,c_{q-1},\mathbf{w}), (c_1,\dots,c_{q-3},c_{q-1},c_q,c_{q-2},c_{q-1},c_q,\mathbf{w})\},
\)
where $\mathbf{w}=(w_1,w_2,\dots,w_{n-q-2})$ is the periodic sequence with period $q$ satisfying $\mathbf{w}=\mathbf{c}_q(n-q-2,(c_1,c_2,\cdots,c_{q-2},c_{q-1},c_{q}))$.

Furthermore, we fully characterize the structure of sequence pairs with Levenshtein distance at least two that maximize the intersection size $N_q(n,2,t)$. In particular,
\begin{itemize}
  \item If \(|D_t(\mathbf{x}) \cap D_t(\mathbf{y})|\) shares the first two terms with \(N_q(n,2,t)\), then \(\mathbf{x}\) and \(\mathbf{y}\) must contain, at the same positions, length-5 subsequences of the forms \((a,b,c,a,b)\) and \((b,a,c,b,a)\), respectively, for some distinct \(a,b,c \in \Sigma_q\).
  \item For $t\geq q+2$ and $q\geq 3$, $|D_t(\mathbf{x})\cap D_t(\mathbf{y})|=N_q(n,2,t)$ if and only if  $\{\mathbf{x},\mathbf{y}\}\in\mathcal{P}$, where
\[
\mathcal{P}=
\left\{ \{\pi(\mathbf{x}^{(q-3)}), \pi(\mathbf{y}^{(q-3)})\} : \pi\in S_q \right\}
\cup
\left\{ \{(\pi(\mathbf{x}^{(q-3)}))^R, (\pi(\mathbf{y}^{(q-3)}))^R\} : \pi\in S_q \right\},
\]
and the sequences $\mathbf{x}^{(q-3)}=(q-2,\cdots, 2,0,1,q-1,0,1,\mathbf{w}^{(q-3)})$ and $\mathbf{y}^{(q-3)}=(q-2,\cdots, 2,1,0,q-1,1,0,\mathbf{w}^{(q-3)})$ with  $\mathbf{w}^{(q-3)} = \mathbf{c}_q(n-q-2,\sigma)$ and $\sigma=(2,3,\dots,q-1,0,1)$.
\end{itemize}

In terms of asymptotic behavior, we derive the explicit expansion
\[
N_q(n,2,t)=\frac{6}{(t-2)!}n^{t-2}-\frac{3t+13}{(t-3)!}n^{t-3}+\frac{3t^2+25t+64}{4(t-4)!}n^{t-4}+O(n^{t-5}),
\]
for all $q\geq 3$ and $t\geq 2$. We also show that $N_q(n,2,t)$ and $N_{q-1}(n,2,t)$ share the same first $q-1$  terms, and their difference in the $n^{t-q}$ term is $\frac{6t-6q+5}{(t-q)!}$ for $t\geq q$. Note that we have proved the identity \(N_q(n,2,t)=N_q^{(q-3)}(n,t)\) only for sufficiently large \(n\). We conjecture that the same equality holds for all \(n>k t q\) for some constant \(k\), all \(q\ge4\), and all $t\geq 2$. The problem remains open for minimum Levenshtein distance \(d\ge3\) and \(q\ge3\). In future work, we will extend our techniques to determine the exact values of \(N_q(n,d,t)\) for all \(d\ge3\), \(q\ge3\), and sufficiently large \(n\).

\appendices
\section{Proofs of Lemmas $\ref{Upp:lm1}$ and $\ref{Upp:lm3}$}\label{APP-A}
The purpose of this appendix is to provide the proofs of Lemmas $\ref{Upp:lm1}$ and $\ref{Upp:lm3}$.

\begin{proof}{ \rm (The proof of Lemma $\ref{Upp:lm1}$)}
We prove by induction on $n$ that $N_q(n,1,2,1)\leq 3$ for any $n\geq 4$. When $n=4$, by using a computerized search, we find that $N_q(4,1,2,1)=3$, which indicates that the base case of the induction holds.

Assume that for some fixed $m\geq 4$, the proposition $N_q(m,1,2,1)\leq 3$ holds.
Consider the case $n = m+1$. Let $\mathbf{x}=(x_1,...,x_{m+2})$, $\mathbf{y}=(y_1,...,y_{m+1})$, and $\mathcal{X}=D_{2}(\mathbf{x})\cap D_{1}(\mathbf{y})$.

Now, we will prove that $|\mathcal{X}|\leq 3$ by the following two cases. 
\begin{itemize}
  \item If \(x_1 \neq y_1\), then
\(
|\mathcal{X}^{x_1}| = | D_{2}(\mathbf{x}_{[2,m+2]}) \cap D_{0-\ell_1^*}(\mathbf{y}_{[3+\ell_1^*,m+1]}) |,
\)
\(
|\mathcal{X}^{y_1}| = | D_{1-\ell_2}(\mathbf{x}_{[3+\ell_2,m+2]}) \cap D_{1}(\mathbf{y}_{[2,m+1]}) |,
\)
and 
\(
|\mathcal{X}^{a}| = 0,
\)
for any \(a \in \Sigma_q \setminus (\{x_1,y_1\} \cup \{y_2\})\),
with the further specification that when $y_2\notin \{x_1,y_1\}$, we have
\(
|\mathcal{X}^{y_2}| = | D_{1-\ell_3}(\mathbf{x}_{[3+\ell_3,m+2]}) \cap D_{0}(\mathbf{y}_{[3,m+1]}) |,
\)
where \(\ell_1^*, \ell_2, \ell_3 \geq 0\). Here, $\ell_1^*=0$  implies $y_2=x_1$;  $\ell_2=0$ implies $x_2=y_1$; $\ell_3=0$ implies $x_2=a$. Thus,
\begin{align*}
|\mathcal{X}^{x_1}|&\leq |D_{0-\ell_1^*}(\mathbf{y}_{[3+\ell_1^*,m+1]})|\leq |D_{0}(\mathbf{y}_{[3,m+1]})|=1,\\
|\mathcal{X}^{y_1}|&\leq |D_{1-\ell_2}(\mathbf{x}_{[3+\ell_2,m+2]})\cap D_{1}(\mathbf{y}_{[2,m+1]})|\leq \max\{|D_{1}(\mathbf{x}_{[3,m+2]})\cap D_{1}(\mathbf{y}_{[2,m+1]})|,|D_{0}(\mathbf{x}_{[4,m+2]})|\}\overset{(a)}{\leq} 2,\\
|\mathcal{X}^{y_2}|&\leq |D_{0}(\mathbf{y}_{[3,m+1]})|=1.
\end{align*}
Here, $(a)$ follows from the fact that  \(|D_{1}(\mathbf{x}_{[3,m+2]})\cap D_{1}(\mathbf{y}_{[2,m+1]})|\le N_q(m,1,1)=2,\)
since \(\mathbf{y}\notin D_1(\mathbf{x})\) and \(x_2=y_1\) together imply \(\mathbf{x}_{[3,m+2]}\neq \mathbf{y}_{[2,m+1]}\).
By Lemmas~\ref{Pro:lm1} and~\ref{Pro:lm2}, it follows that
\(
|\mathcal{X}| = |\mathcal{X}^{x_1}|+|\mathcal{X}^{y_1}|+\sum_{a\notin\{x_1,y_1\}}|\mathcal{X}^a| \le 4.
\)
Equality would require \(\ell_1^*=0\), \(\ell_2=0\), and \(y_2\notin\{x_1,y_1\}\). However, \(\ell_1^*=0\) implies \(y_2=x_1\) which results in a contradiction. Thus, for \(x_1\ne y_1\), we have \(|\mathcal{X}|\le3\).

  \item If \(x_1 = y_1\), then
\(
|\mathcal{X}^{x_1}| = | D_{2}(\mathbf{x}_{[2,m+2]}) \cap D_{1}(\mathbf{y}_{[2,m+1]}) |,
\)
and for some \(a \in \Sigma_q \setminus \{x_1\}\),
\[
|\mathcal{X}^{a}| = | D_{1-\ell_2}(\mathbf{x}_{[3+\ell_2,m+2]}) \cap D_{0-\ell_2^*}(\mathbf{y}_{[3+\ell_2^*,m+1]}) |,
\]
with \(\ell_2, \ell_2^* \geq 0\). Here, $\ell_2=0$ implies $x_2=a$; $\ell_2^*=0$ implies $y_2=a$. If \(\ell_2 \geq 2\) or \(\ell_2^* \geq 1\), then \(|\mathcal{X}^{a}| = 0\).
If \(\ell_2 \in\{0,1\}\) and \(\ell_2^* = 0\), a direct verification also gives \(|\mathcal{X}^{a}| = 0\) since $d_L(\mathbf{x},\mathbf{y})\ge 1$. Hence, when \(x_1 = y_1\), we always have \(|\mathcal{X}^{a}| = 0\). Now, since \(d_L(\mathbf{x},\mathbf{y}) \geq 1\) and \(x_1 = y_1\), it follows that \(d_L(\mathbf{x}_{[2,m+2]},\mathbf{y}_{[2,m+1]}) \geq 1\). By the induction hypothesis, $|\mathcal{X}^{x_1}| \leq 3$. Consequently,
\(
|\mathcal{X}| = |\mathcal{X}^{x_1}| + \sum_{a \in \Sigma_q \setminus \{x_1\}} |\mathcal{X}^{a}| \leq 3.
\)
\end{itemize}

This completes the inductive step for $n = m+1$. By induction, $N_q(n,1,2,1) \leq 3$ for all $n \geq 4$.
\end{proof}

\begin{proof}{ \rm (The proof of Lemma $\ref{Upp:lm3}$)}
Let \(\mathbf{x}= (1,0,2,1,\mathbf{w})\in \Sigma_q^{n+1}\) and \(\mathbf{y}= (0,1,2,\mathbf{w})\in \Sigma_q^{n}\),
where \(\mathbf{w} = (w_1,\dots,w_{n-3})\) with \(w_i \equiv i+2 \pmod q\).
Since the Levenshtein distance between \(\mathbf{x}_{[1,4]}\) and \(\mathbf{y}_{[1,3]}\) is \(1\), we have \(d_L(\mathbf{x},\mathbf{y})\geq 1\).

Let \(\mathcal{X}= D_{t+1}(\mathbf{x}) \cap D_t(\mathbf{y})\). By Lemmas~\ref{Pro:lm1} and~\ref{Pro:lm2},
\(
|\mathcal{X}| = \sum_{i\in \Sigma_q} |\mathcal{X}^i| = |\mathcal{X}^0| + \sum_{i=1}^{q-1} |\mathcal{X}^i|.
\)
By Theorem~\ref{def:thm1}, we have
\(
|\mathcal{X}^0| = N_q(n-1,1,t).
\)
Moreover, $|\mathcal{X}^i|=|D_{t-i}(\mathbf{y}_{[i+2,n]})|= D_q(n-1-i,t-i)$ for $i\in [q-1]$. Thus, \(|\mathcal{X}| = N_q(n-1,1,t) + \sum_{i=1}^{q-1} D_q(n-1-i,t-i).\)

By Theorem \ref{def:thm1} and Eq. \eqref{def:eq1}, we have  $|\mathcal{X}| = D_q(n,t) - D_q(n-2,t) + D_q(n-3,t-1)$ for $n\geq \max\{4,t+3\}$ and $t\geq 1$. Thus, for $n\geq \max\{4,t+3\}$ and $t\geq 1$,
\begin{equation}
N_q(n,1,t+1,t)\geq |\mathcal{X}|=D_q(n,t)-D_q(n-2,t)+D_q(n-3,t-1)=M_q(n,t).\label{App:eq1}
\end{equation}

We now prove that for any \(t \geq 1\) and \(n \geq \max\{4, t+3\}\),
\[
N_q(n,1,t+1,t) \leq M_q(n,t)=D_q(n,t) - D_q(n-2,t) + D_q(n-3,t-1).
\]
Let \(\mathbf{u} = (u_1,\dots,u_{n+1}) \in \Sigma_q^{n+1}\) and \(\mathbf{v} = (v_1,\dots,v_n) \in \Sigma_q^{n}\) with \(d_L(\mathbf{u},\mathbf{v}) \ge 1\). For convenience, let \(\mathcal{Y}= D_{t+1}(\mathbf{u}) \cap D_t(\mathbf{v})\). It suffices to show
\begin{equation}
|\mathcal{Y}| \leq D_q(n,t) - D_q(n-2,t) + D_q(n-3,t-1). \label{App:eq2}
\end{equation}
We consider the following two cases:
\textbf{Case 1:} \(u_1 \neq v_1\);
\textbf{Case 2:} \(u_1 = v_1\).

\textbf{Case 1:} \(u_1 \neq v_1\). We distinguish the following two cases according to whether \(u_2 = v_1\):
\textbf{Case A1:} \(u_2 = v_1\);
\textbf{Case A2:} \(u_2 \neq v_1\).

\textbf{Case A1:}  \(u_2 = v_1 \neq u_1\).

By Lemmas~\ref{Pro:lm1} and~\ref{Pro:lm2}, we have
\(
|\mathcal{Y}| = \sum_{i \in \Sigma_q} |\mathcal{Y}^i|
= |\mathcal{Y}^{v_1}| + \sum_{a \in \Sigma_q \setminus \{v_1\}} |\mathcal{Y}^a|.
\)
Since \(u_2 = v_1\), \(u_1 \neq v_1\), and $d_L(\mathbf{u},\mathbf{v})\geq 1$, we have $\mathbf{u}_{[3,n+1]}\neq \mathbf{v}_{[2,n]}$. By Theorem \ref{def:thm1}, we have 
\(
|\mathcal{Y}^{v_1}| = |D_t(\mathbf{u}_{[3,n+1]}) \cap D_t(\mathbf{v}_{[2,n]})| \leq N_q(n-1,1,t).
\)

For each \(a \in \Sigma_q\setminus\{v_1\}\), we have
\(
|\mathcal{Y}^a| = |D_{t+1-\ell_a}(\mathbf{u}_{[\ell_a+2,n+1]}) \cap D_{t-\ell_a^*}(\mathbf{v}_{[2+\ell_a^*,n]})|
\)
with \(\ell_a \ge 0\) and \(\ell_a^* \ge 1\).  Since the indices \(\ell_a^*\) are pairwise distinct, it follows that
\[
\sum_{a\in \Sigma_q\setminus\{v_1\}} |\mathcal{Y}^a|
\le \sum_{a\in \Sigma_q\setminus\{v_1\}} |D_{t-\ell_a^*}(\mathbf{v}_{[2+\ell_a^*,n]})|
\le \sum_{i=1}^{q-1} D_q(n-1-i, t-i).
\]

Thus, we obtain
\begin{align}
|\mathcal{Y}|
\le N_q(n-1,1,t) + \sum_{i=1}^{q-1} D_q(n-1-i,\, t-i)
\overset{(a)}{\le} \sum_{i=0}^{q-1} M_q(n-1-i,\, t-i) = M_q(n,t), \label{App:eq3}
\end{align}
where $(a)$ follows from Lemma \ref{Upp:lm2}.

\textbf{Case A2:}  \(v_1 \neq u_1\) and \(v_1 \neq u_2\).

By Lemmas~\ref{Pro:lm1} and~\ref{Pro:lm2}, we have 
\(
|\mathcal{Y}| = |\mathcal{Y}^{v_1}| + \sum_{a \neq v_1} |\mathcal{Y}^a|.
\) Since \(u_1 \neq v_1\) and \(u_2 \neq v_1\), the first occurrence of \(v_1\) in \(\mathbf{u}\) appears at some index \(i\ge 3\) such that 
\(
|\mathcal{Y}^{v_1}| = |D_{t-1-\ell}(\mathbf{u}_{[4+\ell,n+1]}) \cap D_t(\mathbf{v}_{[2,n]})|,
\)
where \(\ell=i-3 \ge 0\). Consequently,
\begin{align}\label{App:eq4}
|\mathcal{Y}^{v_1}| \le |D_{t-1-\ell}(\mathbf{u}_{[4+\ell,n+1]})| \le D_q(n-2, t-1)\overset{(a)}{\le} D_q(n-1,t) - D_q(n-2,t) + D_q(n-3,t-1),
\end{align}
where $(a)$ follows from Lemma $\ref{Pro:lm6}$.

For  each \(a \in \Sigma_q\setminus\{v_1\}\),  we have
\(
|\mathcal{Y}^a| = |D_{t+1-\ell_a}(\mathbf{u}_{[\ell_a+2,n+1]}) \cap D_{t-\ell_a^*}(\mathbf{v}_{[2+\ell_a^*,n]})|
\)
with \(\ell_a \ge 0\) and \(\ell_a^* \ge 1\).  Since the indices \(\ell_a^*\) are pairwise distinct, it follows that
\[
\sum_{a \neq v_1} |\mathcal{Y}^a| \le \sum_{a \neq v_1} |D_{t-\ell_a^*}(\mathbf{v}_{[2+\ell_a^*,n]})| \le \sum_{i=1}^{q-1} D_q(n-1-i, t-i).
\]

Combining this with Eqs. \eqref{App:eq3} and  \eqref{App:eq4}, we obtain the desired bound Eq. \eqref{App:eq2}.

\textbf{Case 2:} \(u_1 = v_1\).

Let \(\mathbf{p}\) be the longest common prefix of \(\mathbf{u}\) and \(\mathbf{v}\), and set \(l = |\mathbf{p}|\) (with \(0\le l\le n-1\)).  Write
\(\mathbf{u} = (p_1,\dots,p_l, w_1,\dots,w_{n+1-l})\) and 
\(\mathbf{v} = (p_1,\dots,p_l, w'_1,\dots,w'_{n-l})\),
where \(w_1 \neq w'_1\) and $d_L(\mathbf{u},\mathbf{v})\geq 1$. 

We prove Eq. \eqref{App:eq2} by induction on \(l\)  for \(0 \le l \le n-1\), under the assumptions \(t\geq 1\) and \(n \geq \max\{4, t+3\}\).  The base case \(l=0\) is already established in \textbf{Case 1}; the case \(t=1\) follows from Lemma~\ref{Upp:lm1}, so we may assume \(t\ge2\) in the induction.

Assume the claim holds for all \(0\le l\le s\) with \(s\ge0\). We prove it for \(l=s+1\). 

In this case, write 
\( \mathbf{u} = (p_1,\dots,p_{s+1}, w_1,\dots,w_{n-s})\) and \(\mathbf{v} = (p_1,\dots,p_{s+1}, w'_1,\dots,w'_{n-s-1}) \)
with \(w_1 \neq w'_1, \mathbf{p}=(p_1,\dots,p_{s+1}), \mathbf{w}=(w_1,\dots,w_{n-s}), \mathbf{w}'=(w'_1,\dots,w'_{n-s-1})\). By Lemma~\ref{Pro:lm5}, we have $|D_{t+1}(\mathbf{u})\cap D_t(\mathbf{v})|\leq |D_{t+1}(\mathbf{p}^c,\mathbf{w})\cap D_{t}(\mathbf{p}^c,\mathbf{w}')|$. Thus, it suffices to consider the case $\mathbf{p}\in \mathbf{C}_q(s+1)$. For convenience, we let  $\mathcal{Z} = D_{t+1}(\mathbf{p},\mathbf{w}) \cap D_t(\mathbf{p},\mathbf{w}')$.

We distinguish the following two cases according to whether \(s<q\): \textbf{Case B1:} \(s < q\); \textbf{Case B2:} \(s \geq q\).

\textbf{Case B1:} \(s < q\).

Since \(s+1 \le q\) and $\mathbf{p}\in \mathbf{C}_q(s+1)$, we have \(\{p_1,\dots,p_{s+1}\}\subseteq \Sigma_q\). Let $A=\{p_i|i\in[s+1]\}$ and \(B = \Sigma_q \setminus A = \{b_1,\dots,b_{q-s-1}\}\). By Lemmas~\ref{Pro:lm1} and~\ref{Pro:lm2}, we have 
\(
|\mathcal{Z}| = \sum_{p \in A} |\mathcal{Z}^p| + \sum_{b \in B} |\mathcal{Z}^b|
= \sum_{i=1}^{s+1} |\mathcal{Z}^{p_i}| + \sum_{j=1}^{q-s-1} |\mathcal{Z}^{b_j}|.
\)

\emph{Terms for \(p_i\in A\).}

For each \(i \in [s+1]\),  applying the induction hypothesis yields
\[
|\mathcal{Z}^{p_i}| = \bigl| D_{t-i+2}(\mathbf{p}_{[i+1,s+1]},\mathbf{w}) \cap D_{t-i+1}(\mathbf{p}_{[i+1,s+1]},\mathbf{w}') \bigr|
\le M_q(n-i,\,t-i+1).
\]

\emph{Terms for \(b \in B\).} 

For each \(b \in B\), we have
\(|\mathcal{Z}^b| = \bigl| D_{t-s-\ell_b}(\mathbf{w}_{[2+\ell_b,\,n-s]}) \cap D_{t-s-1-\ell_b^*}(\mathbf{w}'_{[2+\ell_b^*,\,n-s-1]}) \bigr|\)
with \(\ell_b \ge 0\) and \(\ell_b^* \ge 0\).

\begin{itemize}
  \item If \(w'_1 \in A\), then for each \(b\in B\), its first occurrence in \(\mathbf{w}'\) is at a position at least \(2\) such that $\ell_b^*\geq 1$. Since the indices \(\ell_b^*\) (\(b\in B\)) are pairwise distinct, we obtain
  \[
  \sum_{b \in B} |\mathcal{Z}^b|
  \le \sum_{b \in B} \bigl| D_{t-s-1-\ell_b^*}(\mathbf{w'}_{[2+\ell_b^*,\,n-1-s]}) \bigr|
  \le \sum_{i=1}^{q-s-1} D_q(n-s-2-i,\,t-s-1-i).
  \]
  \item If \(w'_1 \notin A\), then \(w'_1 \in B\). For each \(b \in B \setminus \{w'_1\}\), we still have \(\ell_b^* \ge 1\). Since the indices \(\ell_b^*\) (\(b\in B \setminus \{w'_1\}\)) are pairwise distinct, we obtain
  \[
  \sum_{b \in B \setminus \{w'_1\}} |\mathcal{Z}^b|
  \le \sum_{i=1}^{q-s-2} D_q(n-s-2-i,\,t-s-1-i).
  \]
  For the value of $|\mathcal{Z}^{w'_1}|$, we discuss the following two cases. 
  \begin{enumerate}
    \item If \(w_2 = w'_1\), then \(\mathbf{w}_{[3,n-s]}\neq \mathbf{w}'_{[2,n-s-1]}\) since \(d_L(\mathbf{u},\mathbf{v}) \ge 1\). Thus, by Theorem \ref{def:thm1}, we have 
    \[
    |\mathcal{Z}^{w'_1}| = \bigl| D_{t-s-1}(\mathbf{w}_{[3,n-s]}) \cap D_{t-s-1}(\mathbf{w}'_{[2,n-s-1]}) \bigr|
    \le N_q(n-s-2,1,t-s-1).
    \]
    \item If \(w_2 \neq w'_1\) (together with \(w_1\neq w'_1\)), then the first occurrence of \(w'_1\) in \(\mathbf{w}\) is at a position at least \(2\) such that $\ell_{w'_1}\geq 2$. Thus,  
    \[
    |\mathcal{Z}^{w'_1}| \le \bigl| D_{t-s-\ell_{w'_1}}(\mathbf{w}_{[2+\ell_{w'_1},\,n-s]}) \bigr|
    \le D_q(n-s-3,\,t-s-2)
    \overset{(a)}{\le} N_q(n-s-2,1,t-s-1),
    \]
  where $(a)$ follows from Lemma \ref{Pro:lm6}.
  \end{enumerate}
\end{itemize}
Thus, we have
\[
|\mathcal{Z}| \le \sum_{i=1}^{s+1} M_q(n-i,t-i+1)
+ N_q(n-s-2,1,t-s-1)
+ \sum_{j=1}^{q-s-2} D_q(n-s-2-j,\,t-s-1-j),
\]
where the last sum is taken as \(0\) if \(q-s-2 < 1\). Applying Lemma~\ref{Upp:lm2} with parameters \(n' = n-s-1\) and \(t' = t-s-1\) (which satisfy the lemma's conditions because \(n \ge t+3\)) yields
\[
N_q(n-s-2,1,t-s-1) + \sum_{j=1}^{q-s-2} D_q(n-s-2-j,\,t-s-1-j)
\le \sum_{j=0}^{q-s-2} M_q(n-s-2-j,\,t-s-1-j).
\]
Consequently, 
\(
|\mathcal{Z}|
\le \sum_{i=1}^{s+1} M_q(n-i,t-i+1)
+ \sum_{j=0}^{q-s-2} M_q(n-s-2-j,\,t-s-1-j)
= M_q(n,t).
\)

\textbf{Case B2:} \(s \ge q\).

Since \(s+1 \ge q+1\) and $\mathbf{p}\in \mathbf{C}_q(s+1)$, we have \(\{p_1,\dots,p_{q}\}=\Sigma_q\). By Lemmas~\ref{Pro:lm1} and~\ref{Pro:lm2}, we have 
\(
|\mathcal{Z}| = \sum_{i=1}^{q} |\mathcal{Z}^{p_i}|.
\)
For each \(i\in[q]\), the induction hypothesis yields
\[
|\mathcal{Z}^{p_i}| = \bigl| D_{t-i+2}(\mathbf{p}_{[i+1,s+1]},\mathbf{w}) \cap D_{t-i+1}(\mathbf{p}_{[i+1,s+1]},\mathbf{w}') \bigr|
\le M_q(n-i,\,t-i+1).
\]
Thus, we have 
\(
|\mathcal{Z}|= \sum_{i=1}^{q} |\mathcal{Z}^{p_i}|
\le \sum_{i=1}^{q} M_q(n-i,t-i+1)= M_q(n,t).
\)

Therefore, Eq.~\eqref{App:eq2} holds in all cases, which completes the induction. Consequently,
\(
N_q(n,1,t+1,t)\le M_q(n,t)
\)
for $n\geq \max\{t+3,4\}$ and $t\geq 1$. Combining this with Eq.~\eqref{App:eq1} yields equality, proving Lemma~\ref{Upp:lm3}.
\end{proof}

\section*{Acknowledgments}
%\textcolor{red}{The authors would like to express their sincere gratefulness to the editor and the three anonymous reviewers for their valuable suggestions and comments which have greatly improved this paper.}
The work of X.~Wang is supported by the National Natural Science Foundation of China (Grant No. 12001134). The work of F.-W.~Fu is supported by the National Key Research and Development Program of China (Grant No. 2022YFA1005000), the National Natural Science Foundation of China (Grant No. 62371259),  the Fundamental Research Funds for the Central Universities of China (Nankai University), and the Nankai Zhide Foundation.

\end{document}